\documentclass[a4paper,
               DIV=15,
               11pt,
               titlepage=off]{scrartcl}

\usepackage[utf8]{inputenc}
\usepackage[english]{babel}
\usepackage[sc,osf]{mathpazo}
\usepackage{amsthm,amsmath,amsfonts,amssymb}
\usepackage{braket}

\usepackage[dvipsnames]{xcolor}

\usepackage{dsfont,xspace,xparse}
\usepackage{newtxtt,microtype} 
\usepackage{tabularx,xifthen,afterpage,booktabs}
\usepackage{algpseudocode,algorithm}
\usepackage{newunicodechar}
\newunicodechar{−}{-}

\usepackage{subcaption}
\usepackage{graphicx} 
\graphicspath{{./Figures/}}
\usepackage{tikz}
\usetikzlibrary{positioning, arrows.meta, shapes, fit, calc}
\usetikzlibrary{matrix,arrows.meta}
\usetikzlibrary{quantikz2}

\usepackage[backend=biber,
            style=trad-alpha,
            sorting=nty,
            doi=false,
            isbn=false,
            url=false,   
            backref,
            backrefstyle=three,
            maxbibnames=20,
            maxcitenames=2]{biblatex}
\usepackage{csquotes}
\bibliography{refs}
\newbibmacro{string+doi}[1]{\iffieldundef{doi}{#1}{\href{https://dx.doi.org/\thefield{doi}}{#1}}}
\DeclareFieldFormat{title}{\usebibmacro{string+doi}{\mkbibemph{#1}}}
\DeclareFieldFormat[article]{title}{\usebibmacro{string+doi}{\mkbibquote{#1}}}
\DeclareFieldFormat[incollection]{title}{\usebibmacro{string+doi}{\mkbibquote{#1}}}
\DeclareFieldFormat[inproceedings]{title}{\usebibmacro{string+doi}{\mkbibquote{#1}}}
\renewbibmacro{in:}{}
\DefineBibliographyStrings{english}{%
  backrefpage = {page},
  backrefpages = {pages},
}

\usepackage{thmtools}
\usepackage[colorlinks=true,linkcolor=blue,urlcolor=blue,citecolor=blue,anchorcolor=green,pdfusetitle]{hyperref}

\usepackage[capitalise,nameinlink,noabbrev]{cleveref}
\newtheorem{theorem}{Theorem}[section]
\numberwithin{equation}{section}

\newtheorem{definition}[theorem]{Definition}
\newtheorem{proposition}[theorem]{Proposition}
\newtheorem{lemma}[theorem]{Lemma}

\newtheorem{remark}[theorem]{Remark}
\newtheorem*{remark*}{Remark}

\crefname{lemma}{lemma}{lemmas}
\crefname{proposition}{proposition}{propositions}
\crefname{definition}{definition}{definitions}
\crefname{theorem}{theorem}{theorems}
\crefname{example}{example}{examples}
\crefname{section}{section}{sections}
\crefname{appendix}{appendix}{appendices}
\crefname{figure}{figure}{figures}
\crefname{equation}{eq.}{eqs.}
\crefname{table}{table}{tables}
\crefname{algorithm}{protocol}{protocols}

\definecolor{cmykRed}{cmyk}{0, 0.9, 0.9, 0.1}   
\definecolor{cmykBlue}{cmyk}{0.8, 0.5, 0, 0.1}  
\definecolor{cmykGreen}{cmyk}{0.7, 0, 0.8, 0.2} 

\numberwithin{equation}{section}
\numberwithin{theorem}{section}

\DeclareMathOperator{\spec}{Spec}

\DeclareMathOperator{\diag}{diag}
\DeclareMathOperator{\poly}{poly}

\DeclareMathOperator{\tr}{Tr}

\DeclareMathOperator*{\E}{\mathbb{E}}

\newcommand{\Tr}[1]{\tr\left({#1}\right)}

\newcommand{\ketbra}[2]{\ket{#1}\!\!\bra{#2}}

\newcommand*{\NC}{\mathsf{NC}}

\newcommand*{\AC}{\mathsf{AC}}

\newcommand*{\QNC}{\mathsf{QNC}}

\newcommand*{\QAC}{\mathsf{QAC}}

\newcommand{\vecs}{\mathrm{\mathbf{s}}}
\newcommand{\veci}{\mathbf{i}}
\newcommand{\veca}{\mathbf{a}}
\newcommand{\vecb}{\mathbf{b}}
\newcommand{\vecj}{\mathbf{j}}
\newcommand{\vect}{\mathbf{t}}
\newcommand{\veck}{\mathbf{k}}
\newcommand{\vecu}{\mathbf{u}}
\newcommand{\vecv}{\mathbf{v}}
\newcommand{\vecz}{\mathbf{z}}

\title{Worst-case depth hierarchy for shallow quantum circuits}

\author{
    \large{Min-Hsiu Hsieh} \thanks{Hon Hai (Foxconn) Quantum Computing Research Center. Email: \texttt{min-hsiu.hsieh@foxconn.com}.}
    \and
    \large{Michael de Oliveira} \thanks{Hon Hai (Foxconn) Quantum Computing Research Center. Email: \texttt{michael.de.oliveira@foxconn.com}.}
    \and
    \large{Sathyawageeswar Subramanian} \thanks{University of Oxford. Email: \texttt{sathya.subramanian@cs.ox.ac.uk}.}
    \and 
    \large{Xingjian Zhang} \thanks{The University of Hong Kong, University of Technology Sydney. Email: \texttt{Xingjian.Zhang@uts.edu.au}.}
}
\date{}

\begin{document}
\maketitle

\begin{abstract}
Circuit depth is a central resource in complexity theory. In the classical setting, bounded-depth circuits admit well-understood hierarchy theorems. In contrast, constant-depth quantum computation has primarily been studied via separations from classical models, leaving its internal structure comparatively unexplored.

We prove an explicit depth hierarchy theorem for $\QNC^0$. For each $d \ge 12$, we construct a family of two-round interactive problems on which no depth-$(d-1)$ quantum circuit can achieve near-perfect success, regardless of gate set, circuit size, or ancillary qubits. In contrast, we prove that our construction allows honest strategies implementable by simple bounded fan-in quantum circuits of depth larger than $d$ by a small constant factor. Moreover, all bounded fan-in classical circuits of sublogarithmic depth (in the input size) fail to achieve perfect success on these tasks for every $d$, yielding a hierarchy of problems that show unconditional quantum advantage of $\QNC^0$ over $\NC^0$.

A key obstacle is the scarcity of lower bound techniques for quantum circuits. To address this, we develop methods to analyze how depth affects a circuit's ability to realize nonlocal correlations amongst its output qubits in a fine-grained manner. Our approach exploits the correspondence between constraint systems and nonlocal games, translating group-theoretic constructions into rigid operator-valued constraint systems and subsequently into non-local games. In particular, we construct constraint systems whose unique faithful (i.e.,\ injective) operator-valued solutions require every perfect strategy, and every near-perfect strategy up to a fixed precision, to implement multi-controlled phase operations. This reduces to a nonlocal unitary-synthesis problem, yielding depth lower bounds for both shallow quantum and classical circuits.

Our results show that increasing depth strictly increases computational power, exhibiting a robust internal hierarchy within $\QNC^0$ that is genuinely quantum.
\end{abstract}
\clearpage
\tableofcontents
\clearpage

\section{Introduction}
\label{sec:intro}
Circuit depth is a fundamental computational resource. For parallel computation, it serves as a proxy for time, corresponding to the number of sequential and synchronous time steps required to complete the computation. In the classical theory of circuit complexity, depth induces a rich stratification of expressivity and computational power. A long line of work proves explicit worst-case and average-case depth hierarchy theorems for constant-depth circuits and formulas \cite{Sipser83BorelCircuitComplexity,Hastad86STOC,Yao1989,Hastad16FOCS,Chillara2018,Hoza24}, showing that decreasing depth by even one can cause dramatic loss of computational power. These results firmly established shallow circuits as a tractable avenue for fine-grained quantification of how resources shape computational capability.

Shallow quantum circuits have recently emerged as an equally active area of research, shaping our understanding of the power and limitations of quantum computation. On the one hand, constant-depth quantum circuits formalize highly parallel quantum processing, and their limitations are determined by locality and information propagation constraints under bounded fan-in gates. On the other hand, even though the class of decision problems solvable by constant-depth quantum circuits of bounded fan-in ($\QNC^0$) and constant-depth classical circuits of bounded fan-in ($\NC^0$) is exactly the same, shallow quantum circuits can exhibit striking non-classical correlations. Consequently, they can achieve unconditional advantages over their classical counterparts for certain relation and sampling problems.

Depth is also an important practical concern for quantum circuits. In physical implementations of noisy intermediate-scale quantum (NISQ) devices, it captures the number of synchronous ``timesteps'' available before noise and decoherence dominate. In near-term quantum algorithms and quantum machine learning, including variational ans\"atze such as QAOA and quantum neural networks, depth qualitatively affects the class of realizable correlations, the accumulation of noise, and the optimization landscape. Circuit depth is thus a bottleneck for both expressivity and implementability.

Most prior work on the complexity of shallow quantum circuits focuses on comparisons across models, primarily quantum versus classical separations \cite{BGK18ShallowAdvantage,Watts19,le2019average,bravyi2020quantum,Grier20,grier2021interactivenoisy,Briet24decoding,deOliveira2025,Watts23,Grier26sampling}, and on constant-depth models with additional primitives such as unbounded fan-out \cite{Spalek05Fanout,Hoyer2005,TakahashiTani16CollapseQNC0,Grier25threshold,parham2025quantum}. This leaves open a basic structural question about the power of increasing quantum depth:

\begin{quote}
\emph{Does each additional layer of gates make quantum circuits strictly more powerful?}
\end{quote}

Simple counting arguments suggest a superficial answer to this question, in that for any fixed \emph{finite} gate set and for each finite input size, there are strictly more depth-$(d+1)$ circuits than depth-$d$ circuits, so there must exist functions computable in depth $d+1$ that cannot be computed in depth $d$. However, such a statement does not identify an explicit family of functions, does not yield a robust quantitative separation (e.g.,\ a circuit size lower bound), and is sensitive to the choice of gate set in an unstructured manner. The goal of a depth hierarchy theorem is to make this evident structural barrier explicit, quantitative, and precise.

\paragraph{Our contributions.}
We prove an explicit depth hierarchy within $\QNC^0$ by constructing, for each integer $d$, a family of problems that cannot be solved with perfect success probability at depth $d-1$, but become solvable at a slightly larger depth. We then translate these problems into classical-input, classical-output interactive tasks whose required correlations can only be generated by quantum circuits of sufficient depth, and are beyond the reach of shallow-depth classical circuits. Taken together, this shows that even within constant-depth quantum computation, increasing depth strictly increases computational power, while also yielding an unconditional quantum advantage over classical analogs; see \Cref{fig:depth-hierarchy-schematic}.

\begin{figure}[h]
\centering
\resizebox{0.85\textwidth}{!}{%
\begin{tikzpicture}[
    font=\small,
    laptop/.style={draw, very thick, rounded corners=2pt, fill=white},
    screenouter/.style={draw, thick, rounded corners=2pt},
    gate/.style={draw, thick, fill=white},
    depthtag/.style={font=\large\itshape, inner sep=1pt},
]

\newcommand{\laptopdraw}[2]{
\begin{scope}[shift={(#1,0)}]
    \draw[screenouter] (-1.03,-0.08) rectangle (1.03,1.23);

    \draw[laptop] (-0.95,0) rectangle (0.95,1.15);

    \draw[thick, rounded corners=1pt]
        (-1.02,-0.16) -- (1.02,-0.16) --
        (1.18,-0.52) -- (-1.18,-0.52) -- cycle;

    \draw[thick] (-0.92,-0.43) -- (0.92,-0.43);

    \node at (0,-0.92) {#2};
\end{scope}
}

\newcommand{\nandgate}[3]{
\begin{scope}[shift={(#1,#2)}, scale=#3]
    \draw[thick] (-0.25,-0.22) -- (0.00,-0.22)
        arc[start angle=-90,end angle=90,radius=0.22]
        -- (-0.25,0.22) -- cycle;

    \draw[thick] (0.27,0) circle (0.045);

    \draw[thick] (-0.55,0.12) -- (-0.25,0.12);
    \draw[thick] (-0.55,-0.12) -- (-0.25,-0.12);

    \draw[thick] (0.30,0) -- (0.42,0);
\end{scope}
}

\newcommand{\nandinside}[1]{
\begin{scope}[shift={(#1,0)}]

    \draw[thick] (-0.73,0.24) -- (0.1,0.24) -- (0.1,0.16) -- (0.7,0.16);

    \nandgate{-0.28}{0.75}{0.68}
    \nandgate{0.42}{0.48}{0.68}

    \draw[thick] (-0.73,0.832) -- (-0.654,0.832);
    \draw[thick] (-0.73,0.668) -- (-0.654,0.668);

    \draw[thick] (0.0,0.75) -- (0.07,0.75) -- (0.07,0.56);

    \draw[thick] (-0.73,0.398) -- (0.048,0.398);

    \draw[thick] (0.80,0.48) -- (0.7,0.48);

    \draw[thick] (-0.1,0.93) -- (0.30,0.93) -- (0.30,0.73) -- (0.7,0.73);

\end{scope}
}

\newcommand{\atomsymbol}[3]{
\begin{scope}[shift={(#1,#2)}, scale=#3]
    \filldraw[black] (0,0) circle (0.025);

    \draw[thick] (0,0) ellipse (0.18 and 0.06);
    \draw[thick, rotate=60] (0,0) ellipse (0.18 and 0.06);
    \draw[thick, rotate=-60] (0,0) ellipse (0.18 and 0.06);
\end{scope}
}

\newcommand{\qcircuitinside}[2]{
\begin{scope}[shift={(#1,0)}]
    \atomsymbol{-0.70}{0.92}{0.8}

    \foreach \y in {0.25,0.45,0.65,0.8} {
        \draw[thick] (-0.42,\y) -- (0.84,\y);
    }

    \ifnum#2=2
        \draw[gate] (-0.18,0.58) rectangle (0.02,0.92);
        \draw[gate] (0.35,0.18) rectangle (0.55,0.52);
    \fi

    \ifnum#2=3
        \draw[gate] (-0.28,0.58) rectangle (-0.08,0.92);
        \draw[gate] (0.12,0.18) rectangle (0.32,0.52);
        \draw[gate] (0.50,0.38) rectangle (0.70,0.72);
    \fi

    \ifnum#2=4
        \draw[gate] (-0.34,0.58) rectangle (-0.14,0.92);
        \draw[gate] (-0.02,0.18) rectangle (0.18,0.52);
        \draw[gate] (0.30,0.38) rectangle (0.50,0.72);
        \draw[gate] (0.58,0.58) rectangle (0.78,0.92);
    \fi
\end{scope}
}


\laptopdraw{0}{$o(\log n)$}
\nandinside{0}

\node[depthtag, anchor=east] at (-1.4,-0.92) {Depth:};

\laptopdraw{4.25}{$d=O(1)$}
\qcircuitinside{4.25}{2}

\laptopdraw{7.85}{$d+\Delta_1$}
\qcircuitinside{7.85}{3}

\laptopdraw{11.45}{$d+\Delta_1+\Delta_2$}
\qcircuitinside{11.45}{4}

\draw[gray!45, densely dashed, line width=0.5pt, rounded corners=3pt]
    (2.45,-1.38) rectangle (14.85,1.58);

\node[font=\Large\bfseries] at (2.08,0.55) {$\boldsymbol{\not\supset}$};
\node[font=\Large] at (6.05,0.55) {$\subsetneq$};
\node[font=\Large] at (9.65,0.55) {$\subsetneq$};
\node[font=\Large] at (13.95,0.55) {$\cdots$};

\end{tikzpicture}}
\caption{Schematic representation of the quantum depth hierarchies established in this work. The dashed box encloses the hierarchy obtained from the quantum-input classical-output relation problems. The interactive version preserves the same quantum hierarchy using only classical transcripts, while remaining unsolvable by sublogarithmic-depth classical circuits in input size $n$. Up to the constant-factor gap between our lower and upper depth thresholds ($\Delta_1,\Delta_2,...$), this yields an infinite discrete hierarchy.}
\label{fig:depth-hierarchy-schematic}
\end{figure}
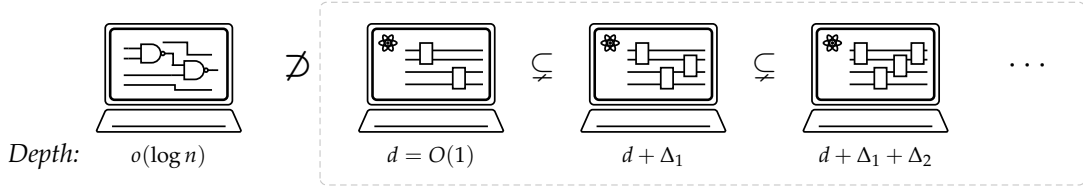

\subsection{Results}
\label{sec:mainresults}

Our first result is a robust worst-case depth hierarchy theorem for $\QNC^0$ in the plain non-interactive model, realized by an explicit relation problem with quantum inputs and classical outputs.

\begin{theorem}[See \Cref{thm:QPlain}]\label{thminf:qplain}
For every integer $d \geq 1$, we construct an explicit family of relation problems $\mathcal{R}^n_d \subseteq \mathcal{H}_d^n \times \{0,1\}^n$ whose valid inputs are stabilizer states in the Hilbert space $\mathcal{H}_d^n$, such that for all large enough $n \in \mathbb{N}$ the following properties hold.
\begin{itemize}
\item \textbf{Perfect completeness:}
There exists a family of quantum circuits $\{C_n\}_n$, over the $2$-qubit Clifford+$\mathsf{T}$ gate set with depth $c\cdot d$ and a constant number of ancillary qubits, that solves $\mathcal{R}^n_d$ with probability $1$ on every valid input $\ket\psi\in\mathcal{H}_d^n$. Here $c>1$ is a universal constant.

\item \textbf{Soundness against shallow circuits:}
For every $d'\leq d-1$, there exists $\varepsilon(d)>0$ such that every family of bounded fan-in circuits solving $\mathcal{R}^n_d$ with depth at most $d'$ can achieve success probability at most $1-\varepsilon(d)$, regardless of the gate set, circuit size, number of ancillary qubits, or access to quantum advice.
\end{itemize}
\end{theorem}

The relation problems $\mathcal{R}_d^n$ are derived from the unitary-synthesis problem of implementing specific quantum Boolean functions, in particular the multi-controlled phase operator $\mathsf{C}^{2^d-1}(\mathsf{Z})$.
Consequently, for each depth parameter $d$, we obtain an explicit family of quantum-classical relation problems that can be solved perfectly by quantum circuits whose depth exceeds $d$ by at most a constant factor.\footnote{For the Clifford+$\mathsf{T}$ gate set, one may take $c \leq 2+2\gamma(4)+\gamma(3)$, where $\gamma(k)$ denotes the depth required to synthesize $\mathsf{C}^{k}(\mathsf{Z})$ (see \Cref{append:Toff_upper}).} Conversely, $\varepsilon(d)$ reflects the difficulty of approximating $\mathsf{C}^{2^d-1}(\mathsf{Z})$ using depth-($d-1$) $\QNC^0$ circuits. Thus, for each $d$ we obtain a task that becomes perfectly solvable only once sufficient depth is available.

Moreover, the depth threshold $d$ also captures the non-Clifford resources required by the task: increasing $d$ corresponds to enforcing observables with more intricate control structure, thereby inducing a hierarchy in quantum magic. In particular, near-perfect success requires implementing increasingly complex observables that approximate multi-controlled phase operators. In the exact setting, this entails a $\mathsf{T}$-gate complexity of $\Omega(\log d)$.

\paragraph{Depth certification.} A key consequence of our formulation is that it yields a simple, unconditional test of quantum depth (or coherence time) that does not rely on computational, cryptographic or oracle-based assumptions \cite{chia2022classical,arora2023quantum}, nor on a number of interaction rounds that scales with the depth being certified \cite{Broadbent09}. The quantum inputs are stabilizer states prepared by elementary Clifford circuits, making the task experimentally accessible, while still probing the ability of an (untrusted) device to sustain coherent quantum evolution over randomly selected subsets of qubits across spatially separated regions.

\paragraph{Comparison to state synthesis.} In the setting of classical-input quantum-output tasks such as state synthesis, $\QNC^0$ depth lower bounds arise from light cone arguments (e.g., GHZ state preparation \cite{Watts19}). However, this is not dual to our setting in any straightforward sense. Furthermore, verifying state preparation typically requires either expensive state tomography with trusted quantum measurements or more elaborate procedures, such as swap tests.

Conceptually, our problem resembles semi-quantum games, where quantum states serve as inputs (or ``questions'') to non-communicating provers \cite{buscemi2012all}. Such games can be used to distinguish entangled states that do not violate Bell inequalities. Here, we use quantum inputs to distinguish the provers' computational capabilities, specifically circuit depth.

\bigskip

Our main result then dequantizes the quantum inputs in \Cref{thminf:qplain}, introducing a fully classical two-round interactive task that preserves the depth hierarchy. Importantly, it removes the need for quantum communication entirely and shows that bounded-fan-in classical circuits of sublogarithmic depth also cannot solve the problem perfectly, regardless of their size.

\begin{theorem}[Main result; see {\Cref{thm:hierarchy-classical-verifier}}]
\label{thm:main-informal}
For every integer\footnote{The minimal depth corresponds to the shallowest $\QNC^0$ circuit that passes the first round of the protocol.} $d \geq 12$, we construct an explicit family of two-round interactive problems $\{\mathcal{IR}_d^n\}_{d}$ in which each round consists of a promise search task specified by a relation $R_d \subseteq \{0,1\}^n \times \{0,1\}^{n}$, such that for large enough $n\in\mathbb{N}$ the following properties hold.
\begin{itemize}
    \item \textbf{Perfect completeness:} There exists a family of quantum circuits $\{C_n\}_n$, over the 2-qubit Clifford$+\mathsf{T}$ gate set with depth $c\cdot d$ and a constant number of ancillary qubits, that solves $\{\mathcal{IR}_d^n\}_{n}$ with probability $1$ on every valid input.

    \item \textbf{Soundness against shallow circuits:} 
    For every $d'\leq d-1$, there exists $\varepsilon(d)>0$ such that every family of bounded fan-in circuits solving $\mathcal{IR}^n_d$ with depth at most $d'$ can achieve success probability at most $1-\varepsilon(d)$, regardless of the gate set, circuit size, number of ancillary qubits, or access to quantum advice.
\end{itemize}
Furthermore, every family of classical $\NC^0$ circuits achieves success probability at most 17/18, regardless of circuit depth or size. 
\end{theorem}

Prior work has identified explicit relation problems that lie beyond constant-depth classical computation yet admit shallow quantum solutions. However, these problems are typically solvable by quantum circuits of fixed small depth (often $d<10$), and the number and variety of such examples remains limited. In the interactive setting, stronger quantum--classical separations have been obtained (e.g.,~\cite{Grier20,grier2021interactivenoisy}) under minimal computational assumptions. The closest related result to ours is that of~\cite{Zhang2024}, which established a separation between magic-free (Clifford) $\QNC^0$ circuits and general $\QNC^0$ circuits while preserving quantum advantage against $\NC^0$. However, these works are not sensitive to circuit depth or to the amount of non-Clifford resources, and current lower bound techniques appear insufficient to address such resource-dependent questions~\cite{slote24,parham2025quantum}.

In contrast, our construction yields an explicit infinite family of problems indexed by depth, where each level becomes solvable only once sufficient quantum depth is available, irrespective of other resources. This reveals a nontrivial internal structure within $\QNC^0$ and, to the best of our knowledge, establishes the first explicit unconditional depth hierarchy for the class. The hierarchy is genuinely quantum: every level exhibits an unconditional separation from shallow classical computation.

\paragraph{Fine-grained lower bounds.} A key ingredient underlying our results is the development of new fine-grained lower bounds for $\QNC^0$. A central challenge is to characterize the correlations that shallow quantum circuits can generate between distant regions and how these scale with circuit depth in multi-output settings. Our approach introduces an operator-valued constraint system obtained by embedding the group $\mathbb Z_2^n$ into a larger group combining Pauli and affine transformations. This construction gives rise to nonlocal games that force the realization of quantum Boolean functions with increasingly complex global structure, which we relate to multi-controlled phase operators. Combined with new unitary-synthesis lower bounds, this yields explicit depth lower bounds across space-like separated regions of the circuit. Our framework thus provides a complementary perspective on quantum depth separations \cite{parham2025quantum,slote24} and a more algebraic route to proving such lower bounds.

\paragraph{Self-testing non-Clifford resources.} As a notable ingredient, we develop a self-testing-type argument for approximately satisfying solutions of the underlying constraint system. We show that approximate success in the protocol implies that the provers' observables are close to the unique operator-valued assignment satisfying the defining relations of the constraint system. Consequently, the protocol certifies the intended algebraic structure, up to local isometries, in the standard sense of self-testing. A key feature of our construction is that it robustly self-tests genuinely non-Clifford resources. In particular, the certified observables include $\mathsf{C}^{m-1}(\mathsf{Z})$ operators, whose implementation requires increasing amounts of non-stabilizer resources (magic). To our knowledge, this is the first result of its kind. Prior works, including \cite{slofstra2020tsirelson,Zhang2024}, can classically witness the presence of non-Clifford resources but do not uniquely determine their realization.

\paragraph{Beyond interactions.} Finally, our construction can also be instantiated in the non-interactive (plain) setting by merging the two stages of the interactive protocol into a single relation problem. In this formulation, the verifier’s computation reduces to a layer of parity computations that can be absorbed into the provers’ quantum strategy, so that the resulting computation is implemented by circuits in the class $\QNC^0[\oplus]$ over a fixed finite universal gate set. This model is natural and practically relevant: it captures key features of current quantum hardware and is closely related to models arising in measurement-based quantum computation \cite{browne2010computational}.\footnote{We conjecture that our techniques may already suffice to obtain unconditional separations in the non-interactive model for circuits consisting of exactly three (total) interleaved $\QNC^0$ and parity layers, against their classical $\NC^0$ counterparts. More generally, possible extensions may obtain separations of the form $\QNC^0[\oplus] \subsetneq \NC^0[\oplus]$ \cite{Briet24decoding}.}

\subsection{Technical Overview}
\label{sec:tech-overview}
A key challenge in establishing a depth hierarchy within $\QNC^0$ is the scarcity of lower bound techniques for quantum circuits. For classical circuits, depth lower bounds in $\NC^0$ are obtained via locality (light cone) arguments, and similar ideas extend to the quantum setting for decision problems. Indeed, many problems that are hard for $\NC^0$ remain information-theoretically hard for $\QNC^0$, and thus cannot distinguish between quantum depths in a fine-grained manner. However, $\QNC^0$ circuits can leverage nonlocal correlations to solve multi-output (i.e. relation or sampling) problems that are hard for $\NC^0$ or even $\AC^0$ and stronger classical circuits. Consequently, separating depths within $\QNC^0$ requires techniques that quantify how such correlations scale with circuit depth.

\paragraph{\(\blacktriangleright\) Roadmap.}
To tackle this challenge, we build on the observation that correlations generated between spatially separated regions of a shallow quantum circuit can be analyzed as if these regions were non-communicating players in a nonlocal game.\footnote{ Slote~\cite{slote24} explores similar ideas in proving $\QNC^0\circ\AC^0$ lower bounds for decision problems, but does not provide the ingredients needed to derive a depth hierarchy from such correlations.} This viewpoint reduces the problem to designing suitable nonlocal games and analyzing their winning strategies. For this purpose, we leverage constraint systems and explore mappings from constraint systems to nonlocal games whose operator-valued solutions enforce computational restrictions on the players, such as bounds on circuit depth.

There are three phases to our construction. First, we define a rigid operator-valued constraint system whose perfect solutions enforce a specific global algebraic structure (see \Cref{sec:constraint_system}). We supplement this with a sharp depth barrier arising from that global structure, showing that shallow circuits below a minimal depth threshold cannot realize it (see \Cref{append:Toffoli}).

Second, we embed this constraint system via a nonlocal game into an interactive protocol involving a Clifford verifier and two provers, which enforces these constraints on the non-communicating provers (\Cref{subsec:Int_Clifford}). Furthermore, we translate this interactive protocol into a relation problem in the plain, non-interactive model with a quantum input provided to $\QNC^0$ circuits (\Cref{sub:quantum_plain}). 

Finally, we dequantize the verifier by replacing its quantum operations with a classical two-round interaction that enforces the necessary Clifford operations via a three-prover self-testing protocol (\Cref{subsec:delegated}). By composing our depth-enforcing constraint system with a constraint system used for self-testing the Clifford operations (\Cref{subsec:composing}), we obtain our main results (\Cref{subsec:final_depth}), yielding a fully classical protocol that preserves the quantum depth hierarchy while retaining quantum advantage over classical shallow depth circuits. We further show that these exact separations admit robust versions, using stability results for approximate representations and self-testing techniques (\Cref{App:Rigidity}).

\paragraph{Faithful realizations of constraint systems.}
The standard one-to-one correspondence between constraint systems and nonlocal games establishes that operator-valued solutions yield perfect strategies for a game whose questions test the defining constraints \cite{cleve2014characterization,cleve2017perfect,coladangelo2017robust,culf2024re}. When the constraint system encodes a group presentation, perfect strategies correspond to representations of the associated solution group. In standard examples, such as the Mermin-Peres magic square, these representations admit low-complexity realizations using Pauli observables and therefore do not distinguish quantum strategies by circuit depth. By contrast, representations whose implementation requires large depth typically arise from non-abelian groups with non-linear relations, which cannot be directly embedded into single-round nonlocal games because the corresponding observables are not jointly measurable.

Our construction circumvents this obstruction by exploiting a canonical faithful representation of the abelian group $\mathbb Z_2^m$ (the Boolean hypercube) in which each generator acts diagonally and nontrivially on exactly one computational basis state. In the associated nonlocal game, perfect strategies are therefore forced to realize observables equivalent to the multi-controlled phase (MCP) gate $\mathsf{C}^{m-1}(\mathsf{Z})$, whose implementation requires depth $\Omega(\log m)$ in bounded-fan-in circuit models over finite $2$-qubit gate sets (see \Cref{append:Toffoli}). Thus, varying $m$ yields a hierarchy of depth lower bounds.

To enforce this representation within a constraint system, we embed the abelian group into a larger algebraic structure that incorporates both Pauli operators and affine symmetries,
\begin{equation}
  G_{\mathbb Z_2^m}=  \bigl(\mathbb Z_2^m *_{\mathbb Z_2^m} \mathcal{P}_m \bigr)\rtimes_{(\tau,\alpha)} \mathrm{AGL}(m,2).
\end{equation}
Intuitively, the amalgamated product $\mathbb Z_2^m *_{\mathbb Z_2^m} \mathcal{P}_m$ constrains the representations of $\mathbb Z_2^m$ to be diagonal in a fixed computational basis via the inclusion of relators in the presentation that ties products of elements in $\mathbb Z_2^m$ with Pauli strings containing only $Z$ and $I$ factors.\footnote{See \Cref{subsec:irrep-to-faithful} for a formal definition of the amalgamated product.} The subsequent semidirect product with $\mathrm{AGL}(m,2)$ enforces faithfulness by requiring that distinct group elements act distinctly under affine relabelings, thereby yielding multi-controlled phase gates (see \Cref{lem:Z2_faithful_rep} for a formal statement).

\paragraph{From constraint systems to depth certification tasks.} The second step of our construction translates the faithful operator-valued constraint system from \Cref{def:OVCS_Z2} into an interactive protocol (\Cref{def:QI_LBCS}). The main obstacle is that the Pauli and affine symmetries appearing in the constraint system do not admit a direct realization as classical question labels in a single-round nonlocal game.

Our key idea is to absorb these symmetries into the verifier’s state preparation. The verifier distributes EPR pairs and, depending on the sampled relation, applies suitable Clifford operations ($\mathsf{CNOT}, \mathsf{Swap}$, and $\mathsf{X}$ gates) implementing the affine symmetries of $\mathrm{AGL}(m,2)$. For the parity–Pauli identification constraints, the verifier may additionally pre-measure a Pauli observable on one prover's side without revealing the outcome. Consequently, the provers are queried only about variables from the $\mathbb Z_2^m$ sector, while the distributed state encodes the remaining structure of the constraint system (see \Cref{def:QS_Z2OCS}). This realizes the constraint system as a natural semi-quantum game.

Although the verifier’s preparation depends on the sampled question, this dependence is hidden locally: each prover receives a maximally mixed reduced state corresponding to one half of an EPR pair in an unknown Clifford basis. Consequently, every perfect strategy induces a satisfying operator assignment of the underlying constraint system, extending the standard correspondence between operator-valued constraint systems and nonlocal games to this semi-quantum setting. Therefore, by the rigidity result of \Cref{lem:Z2_faithful_rep}, every perfect strategy realizes the unique faithful representation, forcing the implementation of multiple-controlled phase operators. The corresponding depth lower bound implies that provers of depth below $\lfloor \log m \rfloor$ cannot succeed perfectly (\Cref{lemma:cliff_verif}).

The completeness analysis shows that honest provers can implement a perfect strategy using shallow circuits. Each queried observable is measured via a compute–phase–uncompute procedure, where the question determines the basis change and the only non-Clifford operation is a generalized controlled-phase measurement (see \Cref{fig:full_framework} and \Cref{thm:upperToffoli}). A technical subtlety arises from the global and parity–Pauli identification constraints, which naively require linear depth in $m$. We address this using a constraint-compression procedure (see the completeness proof of \Cref{lemma:cliff_verif}).

\paragraph{Reduction to the non-interactive (plain) model.} We translate the previous interactive protocol into the plain (non-interactive) circuit model by encoding both the verifier's prepared state and the sampled questions into a single quantum input (\Cref{def:R_m}). A $\QNC^0$ circuit is then viewed as simulating both provers simultaneously. To recover the non-communicating structure in a circuit, we randomize the placement of the relevant registers in the input and apply a light cone argument: with non-negligible probability, the circuit decomposes into two regions whose forward light cones do not intersect, effectively behaving as independent provers. Soundness in the non-interactive model therefore reduces to the interactive case. Completeness follows by implementing, in parallel, the same perfect strategy as in the protocol on the embedded instance. This yields a quantum-input relation problem that inherits the depth hierarchy established in the interactive setting (\Cref{thm:QPlain}).

\paragraph{Dequantizing the Clifford verifier.}
To obtain our main theorem, we must enforce the depth-sensitive structure from \Cref{subsec:Quantum_depth} using only classical interaction, thereby removing the verifier’s quantum capabilities. The main difficulty is that the quantum-verifier protocol relies on two operations unavailable to a classical verifier: preparing Clifford-rotated EPR states and measuring Pauli observables.

We address this by replacing these operations with a two-round three-prover protocol for delegated state preparation and self-testing. Two provers (Alice and Bob) hold the systems on which the final measurements are performed, while a third prover (Charlie) implements the verifier’s hidden Clifford operations through gate teleportation (see \Cref{fig:delegated_state_preparation}). Crucially, only Charlie receives the description of the Clifford operation, while Alice and Bob remain ignorant of it when answering the measurement questions. This asymmetry removes a freedom present in standard two-prover rigidity tests, where local Clifford transformations can be absorbed into equivalent strategies. Charlie’s teleportation transcript then fixes the Pauli frame of the shared state, and a Clifford-rotated parallel Mermin–Peres test certifies that Alice and Bob hold the intended state and measure the corresponding Pauli observables (\Cref{lem:clifford_basis_rigidity}).

The second step is to transfer this protocol to a single prover setting (\Cref{def:Cliff_EPR_self}). We partition the registers of a single $\QNC^0$ circuit into groups corresponding to Alice, Bob, and Charlie. Standard light cone arguments then imply that, with high probability over the verifier’s choice of registers, these groups effectively behave as non-communicating parties. Operationally, the first round performs constant-depth entanglement swapping and gate teleportation to prepare Clifford-rotated Bell pairs between distant regions, while the verifier’s parity computation fixes the corresponding Pauli frame. This yields an unconditional state-commitment procedure for shallow quantum circuits (\Cref{lemma:cliff_test_qnc}).

\paragraph{Towards an interactive protocol with classical transcript.} Finally, we define a fully classical two-round interactive protocol (\Cref{def:Classical_Int}). Its underlying constraint system is obtained by composing the rotated parallel Mermin–Peres constraint system from \Cref{def:Cliff_EPR_self} with the faithful Boolean hypercube system from \Cref{sec:constraint_system}. Algebraically, this mirrors the amalgamated-product structure introduced earlier: the Pauli observables certified by the parallel Mermin–Peres test are identified with the Pauli strings appearing in the parity–Pauli constraints, while the Clifford operations required for the affine symmetry relations are enforced via delegated state preparation. The resulting system admits a unique operator-valued solution, simultaneously fixing the hypercube observables and the auxiliary Pauli sector (\Cref{lem:Extended_Z2_faithful_rep}). Accordingly, in the first round, the verifier classically commits the prover to the correct Clifford-rotated entangled resource, while in the second, it appends the Boolean-hypercube queries together with the Clifford-basis test, thereby enforcing the same faithful solution as in the quantum-verifier protocol.

Completeness and quantum soundness follow as before, while classical soundness follows from the absence of classical solutions to the underlying operator-valued constraint system. Quantitatively, in this setting, the bounds are governed by a parallel Mermin–Peres subtest.\footnote{Note that a separation could also be obtained from the subgame enforcing the multi-controlled phase (Toffoli-type) observables, although the resulting bounds are weaker.} Altogether, this yields a fully classical protocol that preserves the quantum depth hierarchy while establishing an unconditional advantage over shallow classical circuits.

\paragraph{Moving from a rigid to a robust setting.} So far, we established all ingredients required for rigidity of the winning strategy for the semi-quantum game introduced in \Cref{subsec:Quantum_depth}, as well as the one implicitly realized on Alice's and Bob’s side in the protocols of \Cref{sec:Classical_int}. Going further, for our final theorems, we extend this rigidity statement to the approximate regime, yielding a robust quantum depth hierarchy (up to a small depth-dependent error). The central technical challenge is to show that any near-perfect strategy remains close to a genuine operator-valued solution.

A key feature of our construction is that the game admits a natural group-theoretic formulation, allowing us to leverage a general technique based on stability results for approximate representations~\cite{gowers2017generalizations}. Roughly, these results show that operators approximately satisfying the defining relations of a group are close to an exact representation. We employ a state-dependent variant adapted to the nonlocal-game setting, where distances are measured relative to the shared state rather than in operator norm (\Cref{lemma:Vidick}). Consequently, if the game relations hold approximately on the verifier-prepared state, then the induced observables are close, up to isometry, to a genuine representation of the solution group.

Starting from an approximate strategy for \Cref{def:QI_LBCS}, induced by the operator-valued constraint system of \Cref{def:OVCS_Z2}, we show that Alice’s and Bob’s observables approximately satisfy the defining relations of the solution group (\Cref{lemma:set_of_dif}). A key difficulty is that our semi-quantum games yield state-dependent bounds for several different states (\Cref{lemma:prob_to_dist}), whereas the general stability framework requires a single, uniform notion of approximation. We overcome this by exploiting that the Clifford operations are implemented exactly and that all Pauli observables act on the same EPR pairs, allowing the bounds to be transferred and unified across the relevant reference states.

To use the stability result, however, we must extend the approximate operator assignment from generators and relators to arbitrary group elements. Since different decompositions of the same element need not agree exactly in the approximate setting, we fix a canonical form for each element (\Cref{lemma:can_form}), and define its observable as the corresponding ordered product of generator observables and relators (\Cref{lemma:can_form} and \Cref{def:apx_map}). Together with the consistency bounds from \Cref{lemma:set_of_dif}, this allows us to control state-dependent distances between arbitrary group elements, providing the key input to the stability argument and the proof of the main robustness lemma (\Cref{lem:RQP}).

Finally, the circuit depth lower bound for multi-controlled phase gate operators are robust to small approximation errors (\Cref{thm:robustdepth}). Combining the robust gate synthesis depth lower bound with the operator robustness of near-perfect strategies thus yields the robust version of depth hierarchy.

\paragraph{Robustness under approximate state preparation.} For the semi-quantum game in \Cref{sec:Classical_int}, we encounter an additional difficulty in extending rigidity to robustness because neither the state preparation nor the Clifford operations are implemented directly by the verifier. This means that the constraints tested by the game are evaluated on a collection of states, each carrying an additional error term. This obstructs the argument of \Cref{lem:RQP}, where the reduction to the stability lemma relied on exact Clifford relations to reduce arbitrary words in $\mathcal{P}_m \rtimes_{\alpha} \mathrm{GL}(m,2)$ to canonical form. 

To overcome this issue, we derive bounds relating operator distances evaluated on the approximate states to the corresponding distances evaluated on a fixed ideal reference state, together with trace-distance bounds between the approximate and ideal states (\Cref{lemma:sta_chang_D}). We then use the robustness of the Clifford-basis test (\Cref{def:Cliff_EPR_self}) to cover both the states and the observables. This allows us to control the distance between the state prepared through Charlie’s actions and the ideal Clifford-rotated EPR states appearing in \Cref{def:QI_LBCS}. Consequently, all state-dependent distances arising in \Cref{def:Classical_Int} can be transferred to the ideal reference states, at the cost of larger error bounds. This yields approximate relations for the generators and relators that are sufficient to repeat the stability-based argument from the semi-quantum setting, thereby proving the final robustness theorem for \Cref{sec:Classical_int}.
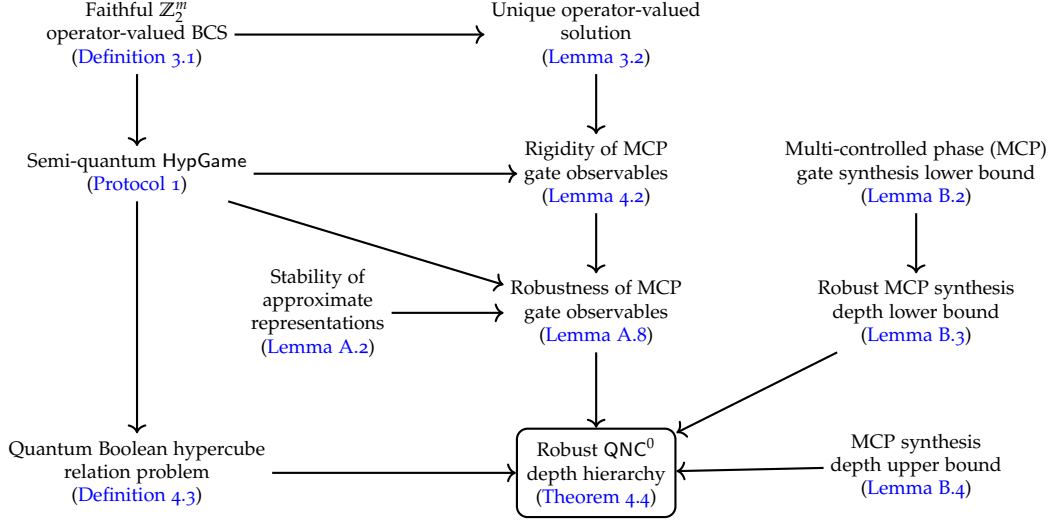
\begin{figure}[h]
\centering
\scriptsize
\begin{tikzpicture}[
    xscale=0.92,
    yscale=0.92,
    transform shape,
    item/.style={align=center, font=\scriptsize},
    arrow/.style={->, thick},
    panel/.style={font=\bfseries\small}
]

\node[panel] at (3.6,1.2) {Part I: Semi-quantum game and plain-model construction};

\node[item] (cs) at (-2,0) {
Faithful $\mathbb{Z}_2^m$\\
operator-valued BCS\\
(\Cref{def:OVCS_Z2})
};

\node[item] (unique) at (4.6,0) {
Unique operator-valued\\
solution\\
(\Cref{lem:Z2_faithful_rep})
};

\node[item] (game) at (-2,-2.0) {
Semi-quantum $\textsf{HypGame}$\\
(\Cref{def:QI_LBCS})
};

\node[item] (rigid) at (4.6,-2.0) {
Rigidity of MCP\\
gate observables\\
(\Cref{lemma:cliff_verif})
};

\node[item] (exactlb) at (9.2,-2.0) {
Multi-controlled phase (MCP)\\
gate synthesis lower bound\\
(\Cref{prop:Toffoli_lower})
};

\node[item] (robust) at (4.6,-4.0) {
Robustness of MCP\\
gate observables\\
(\Cref{lem:RQP})
};

\node[item] (robustlb) at (9.2,-4.0) {
Robust MCP synthesis\\
depth lower bound\\
(\Cref{thm:robustdepth})
};

\node[item] (exactup) at (9.2,-6.2) {
MCP synthesis\\
depth upper bound\\
(\Cref{thm:upperToffoli})
};

\node[item] (stability) at (0.6,-4.0) {
Stability of\\
approximate\\
representations\\
(\Cref{lemma:Vidick})
};

\node[item] (rel) at (-2,-6.3) {
Quantum Boolean hypercube\\
relation problem\\
(\Cref{def:R_m})
};

\node[
    item,
    draw,
    thick,
    rounded corners,
    inner sep=4pt
] (qplain) at (4.6,-6.3) {
Robust $\QNC^0$\\
depth hierarchy\\
(\Cref{thm:QPlain})
};

\draw[arrow] (cs) -- (unique);
\draw[arrow] (cs) -- (game);

\draw[arrow] (unique) -- (rigid);
\draw[arrow] (game) -- (rigid);

\draw[arrow] (rigid) -- (robust);
\draw[arrow] (game) -- (robust);

\draw[arrow] (stability) -- (robust);

\draw[arrow] (exactlb) -- (robustlb);
\draw[arrow] (robust) -- (qplain);

\draw[arrow] (game) -- (rel);

\draw[arrow] (exactup) -- (qplain);

\draw[arrow] (rel) -- (qplain);
\draw[arrow] (robustlb) -- (qplain);

\end{tikzpicture}
\caption{Logical dependencies among the main ingredients in our analysis of the quantum-input relation problem and the robust depth hierarchy of \Cref{thminf:qplain} (see \Cref{thm:QPlain} for its formal version).}
\label{fig:proof-dependency-architecture}
\end{figure}

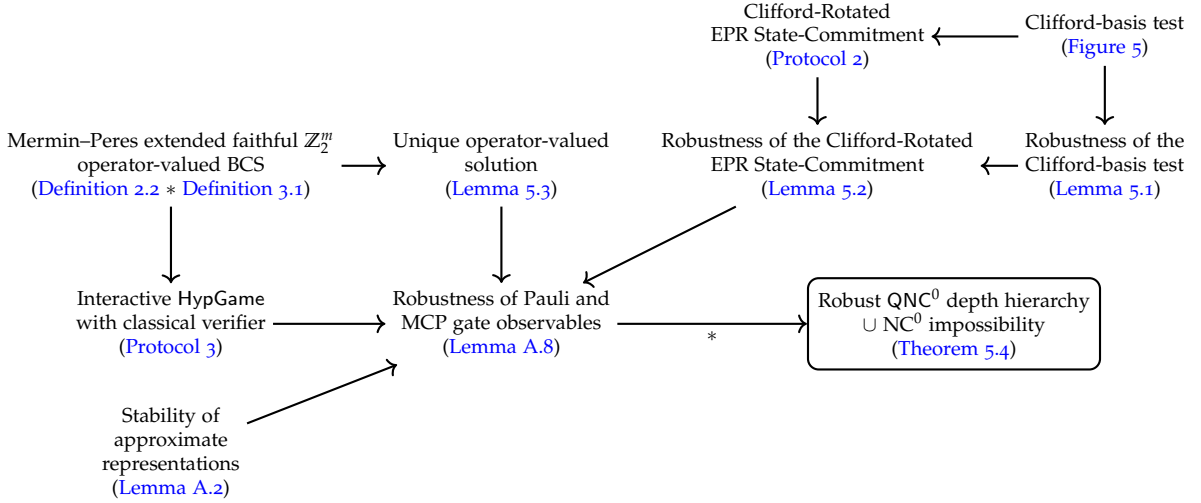
\begin{figure}[H]
\centering
\scriptsize
\begin{tikzpicture}[
    xscale=0.95,
    yscale=0.95,
    transform shape,
    item/.style={align=center, font=\scriptsize},
    arrow/.style={->, thick},
    panel/.style={font=\bfseries\small}
]

\node[panel] at (5.2,3.0)
{Part II: Dequantization / Classical-verifier construction};

\node[item] (mpcs) at (0,0) {
Mermin--Peres extended faithful $\mathbb Z_2^m $\\
operator-valued BCS\\
(\Cref{def:MP_parallel} $*$ \Cref{def:OVCS_Z2})
};

\node[item] (interactive) at (0,-2.2) {
Interactive $\textsf{HypGame}$\\ with classical verifier\\
(\Cref{def:Classical_Int})
};

\node[item] (unique) at (4.6,0) {
Unique operator-valued\\
solution\\
(\Cref{lem:Extended_Z2_faithful_rep})
};

\node[item] (forced) at (4.6,-2.2) {
Robustness of Pauli and\\ MCP gate observables \\
(\Cref{lem:RQP})
};

\node[
    item,
    draw,
    thick,
    rounded corners,
    inner sep=4pt
] (final) at (10.9,-2.2) {
Robust $\QNC^0$
depth hierarchy \\
$\cup$ NC$^0$ impossibility\\
(\Cref{thm:hierarchy-classical-verifier})
};

\node[item] (stability) at (0.0,-4.0) {
Stability of\\
approximate\\
representations\\
(\Cref{lemma:Vidick})
};

\node[item] (protocol2) at (9.0,1.8) {Clifford-Rotated\\ EPR State-Commitment\\
(\Cref{def:Cliff_EPR_self})
};

\node[item] (figure2) at (13.0,1.8) {
Clifford-basis test\\(\Cref{fig:delegated_state_preparation})
};

\node[item] (lem52) at (9.0,0) {Robustness of the Clifford-Rotated\\ EPR State-Commitment\\
(\Cref{lemma:cliff_test_qnc})
};

\node[item] (lem51) at (13.0,0) {Robustness of the\\ Clifford-basis test\\
(\Cref{lem:clifford_basis_rigidity})
};

\draw[arrow] (mpcs) -- (unique);
\draw[arrow] (mpcs) -- (interactive);

\draw[arrow] (interactive) -- (forced);
\draw[arrow] (unique) -- (forced);

\draw[arrow] (forced) -- node[midway, below] {$*$} (final);

\draw[arrow] (stability) -- (forced);

\draw[arrow] (figure2) -- (protocol2);

\draw[arrow] (protocol2) -- (lem52);
\draw[arrow] (figure2) -- (lem51);

\draw[arrow] (lem51) -- (lem52);
\draw[arrow] (lem52) -- (forced);

\end{tikzpicture}
\caption{Logical dependencies among the main ingredients in our analysis of the classical-verifier construction underlying \Cref{thm:main-informal} (see \Cref{thm:hierarchy-classical-verifier} for its formal version). Note that the dependency indicated by the arrow $\underset{*}{\rightarrow}$ is not direct. Its use additionally requires combining \Cref{lem:RQP} with \Cref{def:Classical_Int}, \Cref{thm:robustdepth}, and \Cref{thm:upperToffoli}, as in the proof of \Cref{thminf:qplain}. }
\label{fig:classical-proof-dependency}
\end{figure}

\medskip
In summary, our framework provides a way to construct depth-sensitive tasks exhibiting quantum advantage, while offering an algebraic route to proving such separations through constraint systems and their representations. Beyond the specific constructions presented here, the mapping from operator-valued constraint systems to semi-quantum games and the associated protocols extends naturally to other amalgamations of $\mathbb Z_2^n$ with the Pauli group, to other abelian groups in place of $\mathbb Z_2^n$, and potentially to more general classes of constraint systems. We expect these ideas to provide new tools for studying the fine-grained structure of shallow quantum computation and its separation from shallow classical models.

\subsection{Related Work}
\label{sec:related-work}

We place our results in the context of four closely related directions: classical circuit depth hierarchies, separations between shallow quantum and classical circuits, the broader study of depth as a resource within quantum computation, and connections to self-testing protocols.

Depth is a fundamental resource in classical circuit complexity. It is well known that $\AC^0$ circuits cannot compute PARITY, with superpolynomial lower bounds established by Ajtai and by Furst, Saxe, and Sipser \autocite{Ajtai83Sigma11,FSS84ParityPH}, and sharpened via H{\aa}stad’s switching lemma \autocite{Hastad86STOC,Hastad87Book}. These techniques also yield explicit depth hierarchy theorems, showing that increasing depth strictly increases computational power \autocite{Sipser83BorelCircuitComplexity}. Stronger results in the average-case setting further show that shallow circuits cannot even approximate deeper ones without exponential size blowup \autocite{HRST17JACM,Hastad16FOCS,Hastad16ECCC,Hoza24}, building on a large body of work on correlation bounds, Fourier concentration, and influence \autocite{LMN93JACM,Boppana97AvgSensitivity,Tal17FourierTails,ODonnellWimmer07ICALP,Rossman2015}. Our results establish a depth hierarchy for $\QNC^0$, the natural quantum analogue of $\NC^0$, and identify it as the first shallow quantum circuit class for which such hierarchy theorems become genuinely nontrivial.

\paragraph{Shallow quantum circuits.}
A breakthrough result of Bravyi, Gosset, and K\"{o}nig exhibited an explicit relation problem solvable by constant-depth quantum circuits, but requiring logarithmic depth for classical bounded fan-in circuits \autocite{BGK18ShallowAdvantage}. This first unconditional separation between $\QNC^0$ and $\NC^0$ has since been extended in several directions, including separations against stronger classes such as $\AC^0$ and average-case hardness \cite{le2019average,Watts19,coudron2021,deOliveira2025}, interactive separations from $\NC^1$ \cite{Grier20}, robustness to noise and fault-tolerant settings \cite{bravyi2020quantum,grier2021interactivenoisy,hasegawa21,caha2023colossal}, and sampling-based separations \cite{Watts23,Grier26sampling}.

These works establish that shallow quantum circuits can outperform shallow classical circuits, but they do not address the internal structure of $\QNC^0$ itself. In particular, they identify problems solvable in constant depth without distinguishing between different constant-depth regimes or characterizing how computational power scales with depth. However, this question is practically relevant, as circuit depth controls both expressivity and noise in near-term variational algorithms, including QAOA, quantum neural networks, and other quantum machine learning models \cite{Niu2019qaoa,Hastings2019,PellowJarman2024,Ostrowski2020qaoa,Wu2021}. From this perspective, our results suggest that increasing circuit depth enables the realization of increasingly complex nonlocal correlations and computational advantages that remain inaccessible to shallower quantum circuits, thereby refining and extending known quantum advantages over classical shallow-depth models.

\paragraph{Quantum depth hierarchies.}
Understanding the role of circuit depth in quantum computation remains a central open question. On the one hand, key subroutines such as those underlying Shor’s algorithm can be implemented in polylogarithmic depth \cite{cleve2000fast,gossett1998quantum}, suggesting that surprisingly shallow quantum circuits may already capture much of the power of quantum computation. On the other hand, it is widely believed that increasing circuit depth enables strictly broader classes of computations to be performed. Evidence for the first viewpoint appears in conjectures such as that of Jozsa \cite{Jozsa2005}, and subsequent discussions by Aaronson \cite{aaronson2005qchallenge}, suggesting that interleaving polylogarithmic-depth quantum circuits with classical computation may suffice to capture $\mathsf{BQP}$.

Evidence for depth-based separations between polylogarithmic-depth quantum circuits and more powerful circuit classes has been obtained under computational assumptions \cite{chia2022classical} and in oracle models \cite{arora2023quantum}, but unconditional results remain scarce due to the lack of robust lower-bound techniques for quantum circuits. This mirrors longstanding open problems in classical complexity, such as separating $\mathsf{NC}$ from $\mathsf{P}$. The difficulty is further highlighted by recent works showing that certain circuit classes, previously considered promising candidates for separating $\mathsf{QNC}$ from $\mathsf{EQP}$ \cite{bernstein1993quantum}, admit unexpectedly strong parallelization properties \cite{watts2025quantum}. 

More recently, new approaches to lower bounds for shallow quantum circuits have emerged, including information-theoretic \cite{parham2025quantum} and Fourier-analytic \cite{slote24} techniques. While these results begin to illuminate the limitations of shallow quantum computation, a systematic understanding of how computational power scales with depth remains open.

Our results contribute to this direction by establishing an explicit, unconditional depth hierarchy, while introducing new fine-grained lower-bound techniques that enforce depth-sensitive quantum correlations.

\paragraph{Connections to self-testing.} Although our final protocol with classical inputs and classical outputs is not formulated explicitly as a robust self-test of nonlocal games, it admits a natural interpretation in this manner: approximate success enforces closeness to the unique operator assignment satisfying the constraint system, while simultaneously certifying the presence of maximally entangled resource states. Specifically, the protocol certifies the subgroup $\mathcal{P}_m \rtimes_{\alpha} \mathrm{GL}(m,2)$ via a self-test of EPR pairs, Pauli observables, and certain Clifford operators.

We remark that previous work has developed significantly more efficient and robust protocols than ours for self-testing many EPR pairs and Pauli observables \cite{Anand_Vidick18}, including extensions to single-qubit Clifford observables \cite{Coladangelo2024}. While these results are highly relevant, it is not immediately clear how to integrate them into our framework. In particular, our construction achieves weaker self-testing guarantees for the underlying resource states, but in exchange it enables the certification of genuinely non-Clifford resources through observables associated with $\mathsf{C}^{m-1}(\mathsf{Z})$ operations, thereby probing increasing levels of quantum magic. 

From a technical perspective, our robustness analysis follows the standard group-theoretic approach to self-testing \cite{cleve2017perfect}, leveraging stability results for approximate representations \cite{gowers2017generalizations,slofstra2019set,coladangelo2017robust}. While \cite{coladangelo2017robust} suggested that this framework should extend to more general settings including amenable groups, our construction already operates over a finite group. The main additional challenge in our case arises from the semi-quantum nature of the games, where the provers receive question-dependent states. In the quantum information community, self-testing-type robust certification of the quantum strategy in a semi-quantum game \cite{buscemi2012all} has been developed \cite{supic2020self}. Nevertheless, integrating these tools with the group-theoretic rigidity analyses for standard nonlocal games in a fully state-dependent setting would require substantial extensions.

\subsection{Outlook and Future Work}
\label{sec:discussion}

Our results reveal a stratification phenomenon within shallow quantum computation, formalizing the intuition that circuit depth is a genuine resource even in the constant-depth regime. They complement both classical bounded-depth hierarchy theorems and separations between shallow quantum and classical circuits by showing that, even within $\QNC^0$, increasing depth can strictly enhance expressivity and computational power. In particular, our work goes beyond demonstrating quantum advantage over classical circuits and instead reveals a nontrivial fine structure internal to $\QNC^0$ itself.

A useful perspective is that the hierarchy isolates structural barriers to general depth reduction methods. A natural question in circuit complexity is whether a circuit of depth $d$ can be simulated by a similar circuit of depth $d'<d$, with a potential increase in circuit size. Our hierarchy theorem for $\QNC^0$ indicates that general depth reduction may not be possible for this class of circuits. In contrast, depth reduction theorems are known for other classes such as Clifford circuits or when the connectivity of the underlying architecture is limited, where depth-width (space or number of ancillae) tradeoffs have been established \cite{Jiang2020,Yuan2024,watts2025quantum}.

The quantum strategies arising in our protocol with a classical verifier admit natural realizations within standard models of near-term quantum computation, including shallow circuits with adaptive mid-circuit measurements and classical feedforward. In particular, the first round of the protocol implements a qubit routing primitive, akin to those used in architectures with limited connectivity \cite{devulapalli2024quantum}, while the second round involves multi-qubit controlled phase gates, which are ubiquitous in quantum algorithms and circuit design. This correspondence suggests that the depth requirements identified in our work may also arise in practical implementations of quantum computation.

\medskip
While our construction introduces such a hierarchy, it also leaves several aspects only partially understood. We discuss several broader directions for future investigation below, extending beyond refinements of the current parameters. 

\paragraph{Robustness.}
Our construction enforces exact solutions to the constraints under perfect success in the corresponding nonlocal game. Moreover, approximate success still enforces an approximate realization of the underlying algebraic structure, yielding corresponding depth lower bounds  up to suitable error thresholds. However, our hierarchy relies on multi-controlled phase operations whose behavior depends on global input correlations, much like an AND function. If the approximation error is too large, a dishonest strategy could replace the intended operation with a trivial one while still achieving a high success probability, collapsing the hierarchy. Thus, the robustness threshold must be strong enough to exclude such degeneracies. At the same time, approximate implementations within these thresholds may require substantially fewer non-stabilizer resources \cite{gosset2025multi}.

One possible direction to obtain more robust hierarchies is to instantiate other quantum Boolean functions within our framework, particularly ones that are structurally distinct from Toffoli-type gates, such as parity (fan-out) operations. Such extensions may also lead to alternative magic hierarchies or refined separations under better approximate notions of success.

\paragraph{Average-case hierarchies.} Our construction identifies explicit problems witnessing depth separation in the worst-case setting, where shallow circuits cannot succeed perfectly on all valid inputs. A natural direction is to investigate analogous \emph{average-case} hierarchies via correlation bounds, asking whether depth-($d-1$) strategies can achieve non-trivial correlation with perfect quantum strategies over natural input distributions without substantial circuit-size overhead. Such results would bring the hierarchy closer to practical settings, where one might seek computational advantage on typical inputs rather than every instance, and may also relate to quantum pseudorandomness.

\paragraph{Stronger circuit classes.} A natural question is whether similar unconditional depth hierarchies can be established in richer quantum circuit classes, such as $\QAC^0$ and $\QNC^0[\oplus]$. These models admit more powerful parallel primitives, including certain nonlocal operations in constant depth, and remain poorly understood from a lower-bound perspective. Progress in this direction would clarify the role of depth in shallow quantum computation, as well as the power of extensions such as fan-out and adaptive mid-circuit measurements. It is also natural to ask whether the resulting quantum–classical separations can be strengthened against more powerful classical models, such as $\AC^0$.

\paragraph{Depth certification.}
Our work also provides a concrete route toward unconditional, classically verifiable delegated certification of the quantum circuit depth achievable on an untrusted device. In our setting, a classical verifier interacts with entangled provers to certify that any successful strategy must implement a circuit of at least a prescribed depth. However, the current protocols fall short of optimal iteration–depth tradeoffs, where the number of iterations measures the complexity of sampling inputs (questions). In the absence of a separation between $\mathsf{QNC}$ and $\mathsf{EQP}$, the best one might hope for is to certify depths up to $O(\log n)$ using $\poly n$ iterations, whereas our construction only reaches depth $O(\log\log n)$. Closing this gap for certification protocols remains an important problem.

\medskip

More broadly, our results point to a general methodology. Combining linear constraint systems with symmetry-enforcing non-linear relations enables us to enforce rigid global structure of observables, while preserving the joint measurability needed for embedding them into nonlocal games and interactive protocols. This approach suggests a new route toward fine-grained lower bounds in shallow quantum models, extending beyond depth to other resources such as non-stabilizerness or communication. We expect the further development of these ideas to lead to new hierarchy theorems, new certification protocols, and further connections between interactive proofs and quantum circuit complexity.

\subsection*{Organization} We introduce the necessary preliminaries in \Cref{sec:prelims}. The rest of the paper is organized as follows. In Section~\ref{sec:constraint_system}, we introduce our group-embedding technique and construct the constraint system for which we establish rigidity and uniqueness of solutions, properties that form the algebraic backbone of our construction. In Section~\ref{subsec:Quantum_depth}, we show how to embed this constraint system into a quantum interactive protocol with a Clifford verifier, and derive a quantum-input relation problem that exhibits a quantum depth hierarchy. In Section~\ref{sec:Classical_int}, we remove the quantum capabilities of the verifier by introducing a classical two-round interactive protocol, and prove our main depth hierarchy theorem together with quantum advantage over classical circuits.

The appendices contain additional technical developments and supporting proofs. In particular, they extend the rigidity statements to robust variants, relate these robustness guarantees to circuit lower bounds through multi-controlled phase gates, and provide the auxiliary constructions associated with these gates.

\section{Preliminaries}
\label{sec:prelims}
We first fix notation and recall the basic concepts we use:
(i) constant-depth quantum circuits, (ii) operator-valued (quantum) constraint satisfaction systems and (iii) two-prover nonlocal (interactive) games.

\paragraph{Notation.}
For $n\in\mathbb{N}$, let $[n]:=\{0,\dots,n-1\}$. We identify $\{0,1\}^n$ with the additive group $\mathbb{Z}_2^n$, with group operation bitwise XOR, denoted $\oplus$. For $a\in[n]$, let $\mathbf{e}_a\in\{0,1\}^n$ be the standard basis vector with a $1$ in coordinate $a$. We index coordinates from $0$ to $n-1$. We use $i\in[M]$ for integer indices and $\veci\in\mathbb{Z}_2^m$ for their binary representations. More generally, bold symbols $\veci,\vecj$ etc.\ will denote vectors, and the corresponding non-bold symbols will denote integer indices.

For $\veci,\vecs\in\{0,1\}^n$, let $|\veci|$ denote the Hamming weight of $\veci$, let $\veci[j]$ be the $j$th bit of $\veci$, and define $\veci\cdot\vecs \;:=\; \sum_{j=0}^{n-1} \veci[j]\vecs[j] \pmod 2$. We index families of operators by bit strings, e.g.\ $\{z_{\veci}\}_{\veci\in\{0,1\}^n}$. For $a,b\in[n]$, let $\sigma_{a,b}:\{0,1\}^n\to\{0,1\}^n$ be the permutation that swaps coordinates $a$ and $b$. 

We focus on relation problems, where outputs $y \in Y$ on inputs $x \in X$ must satisfy $(x,y) \in R$ for some relation $R \subseteq X \times Y$.

\paragraph{Hilbert space and operator notation.} Let $\mathcal{H}^m=\left(\mathbb C^2\right)^{\otimes m}$ denote the $m$-qubit Hilbert space of dimension $M=2^m$, and let $U(\mathcal{H}^m)$ denote the set of unitaries on $\mathcal{H}^m$. For an operator $A$, let $A^\dagger$ denote its adjoint, and $\|A\|_2$ its operator (spectral) norm. An observable is a Hermitian operator with $A=A^\dagger$. A binary observable, or $\pm1$-observable, satisfies $A^2=I$ (i.e., $\spec(A)\subseteq\{\pm1\}$).
A set of observables $\{A_i\}$ is simultaneously diagonalizable if its elements commute pairwise. 

We use the computational basis $\{\ket{\mathbf{x}}:\mathbf{x}\in\{0,1\}^m\}$ for $(\mathbb{C}^2)^{\otimes m}$. Let $X_i$ denote the single-qubit Pauli-$X$ acting on qubit $i$, and similarly for other standard gates. For $\vecs\in\{0,1\}^m$, define 
\[
Z(\vecs) := \bigotimes_{j=0}^{m-1} Z^{\vecs[j]}, \qquad
X(\vecs) := \bigotimes_{j=0}^{m-1} X^{\vecs[j]},
\]
so that $Z(\vecs)\ket{\mathbf{x}}=(-1)^{\vecs\cdot \mathbf{x}}\ket{\mathbf{x}}$ and $X(\vecs)\ket{\mathbf{x}}=\ket{\mathbf{x}\oplus \vecs}$. For $\vect= (\veca,\vecb) \in \mathbb{Z}_2^{2m}$, define the $m$-qubit Pauli operator
\[
P(\vect) = i^{\veca\cdot \vecb} X(\veca) Z(\vecb).
\] 

\paragraph{Constant-depth circuits.}
A \emph{quantum circuit} on $n$ qubits is a sequence of unitary gates from a fixed finite gate set, each acting on at most $K=O(1)$ qubits, i.e.\ having bounded fan-in. The \emph{depth} $d$ is the minimum number of layers obtained by parallelizing gates with disjoint supports. We write $\mathcal{C}_d$ for the set of circuits in a class $\mathcal{C}$ with depth at most $d$. Classical outputs are obtained from the circuit by measuring a designated subset of $k$ qubits in the computational basis, yielding a bit string $\mathbf{y} \in \{0,1\}^k$.
Circuits may use \emph{ancillary qubits} initialized to a fixed state, typically $\ket{0}$ (a ``clean'' ancilla). 

We use $\QNC^0$ to denote families of polynomial-size, constant-depth quantum circuits with bounded fan-in gates, and $\mathsf{NC}^0$ for the classical analogue.
A key structural constraint of such circuits is locality: in a depth-$d$ circuit with fan-in $K$, each output depends only on a backward light cone of inputs, of size at most $K^d$. For constant $d$, this is $O(1)$, independent of the number $n$ of inputs. Unless stated otherwise, we assume $K=2$.

We also consider parity extensions of these circuits. The class $\mathsf{NC}^0[\oplus]$ augments $\mathsf{NC}^0$ with unbounded fan-in parity gates, allowing a single gate to compute the $\mathsf{XOR}$ of an arbitrary number of input bits, bypassing the light cone restrictions of standard $\mathsf{NC}^0$. Its quantum analogue $\QNC^0[\oplus]$ allows parity-type operations on an unbounded number of qubits (e.g.\ via $\mathsf{CNOT}$ circuits), enabling global dependencies within constant depth.

\subsection{Operator-Valued Constraint Systems}
\label{sec:prelims-ovcs}
We will encode the intended prover behavior via operator-valued constraints.
We consider constraints to be binary-valued. An operator-valued constraint system $S$ of size $|S|$ consists of:
\begin{itemize}
  \item a set of formal variables $V=\{v_i\}_{i\in\mathcal{I}}$;
  \item a collection $\{\mathcal{S}_1,\ldots,\mathcal{S}_{|S|}~\mid~\mathcal{S}_{j}\subseteq\mathcal{I}\}$ of contexts, defining constraints $\{C_j=(\mathcal S_j,c_j)\}_j$ (word equations) 
 \[
\prod_{l\in\mathcal S_j} v_l =c_j I,
\qquad c_j\in\{\pm1\}.
\]
  \item (optional) conjugation constraints $U\, v_i\, U^\dagger \;=\; v_{f(i)}$ for fixed unitaries $U$ and $f:\mathcal{I}\to\mathcal{I}$.
\end{itemize}
Such a constraint system is also called a quantum binary constraint system (BCS). A \emph{quantum satisfying assignment} (or \emph{operator-valued solution}) assigns to each variable $v_l$ a $\pm1$-observable on some finite-dimensional Hilbert space $\mathcal{H}$ such that all constraints hold exactly.

In our results, we use the following two standard constraint systems.
The first, the Mermin--Peres system, is the basic parity obstruction behind the magic-square game: multiplying all six equations would force $1=-1$ for scalar $\{\pm1\}$ assignments, so it has no classical satisfying assignment.
Nevertheless, it admits the Pauli operator solution below, and its parallel version will later serve as a Pauli-certification/self-testing gadget in our construction.

\begin{definition}[Mermin--Peres constraint system]\label{def:MP}
    The Mermin--Peres (a.k.a. magic square) linear constraint system is given by the following constraints:
    \begin{equation}\label{eq:MerminPeres}
\begin{aligned}
& v_1v_2v_3=1, \quad  v_4v_5v_6=1, \quad v_7v_8v_9=1, \\
& v_1v_4v_7=1, \quad v_2v_5v_8=1, \quad v_3v_6v_9=-1, \\
\end{aligned}
\end{equation}
which is satisfied (up to a global isometry) by the two-qubit Pauli strings:
\begin{equation}\label{eq:MPqubits}
\begin{aligned}
& v_1=Z\otimes I,\quad v_2=I\otimes Z, \quad v_3=Z\otimes Z, \\
& v_4=I\otimes X, \quad v_5=X\otimes I, \quad v_6=X\otimes X, \\
& v_7=Z\otimes X, \quad  v_8=X\otimes Z, \quad v_9=Y\otimes Y.
\end{aligned}
\end{equation}
\end{definition}

\begin{definition}[Parallel Mermin--Peres constraint system]
\label{def:MP_parallel}
For $m \in \mathbb{N}$, the $m$-fold parallel Mermin--Peres constraint system consists of $m$ independent copies of the Mermin--Peres constraint system. For each copy $j \in [m]$ with variables $\{v_{j,1},\dots,v_{j,9}\}$, define derived variables
\begin{align}
a_{j,0} &= v_{j,1}v_{j,2}v_{j,3}, &
a_{j,1} &= -\,v_{j,1}v_{j,4}v_{j,7}, \\
a_{j,2} &= v_{j,4}v_{j,5}v_{j,6}, &
a_{j,3} &= -\,v_{j,2}v_{j,5}v_{j,8}, \\
a_{j,4} &= v_{j,7}v_{j,8}v_{j,9}, &
a_{j,5} &= v_{j,3}v_{j,6}v_{j,9}.
\end{align}

The parallel system enforces the following global constraints:
\begin{equation}\label{eq:parallelanti}
\prod_{j=0}^{m-1} a_{j,\veci[j]} = (-1)^{\sum_{j=1}^{m} \veci[j]},
\qquad \forall \veci \in \{0,1,\dots,5\}^{m}.
\end{equation}

In addition, variables corresponding to different copies commute. This constraint system is satisfied by $m$-fold parallel repetitions of the solutions to the Mermin--Peres system (\Cref{def:MP}).
\end{definition}

\subsection{Standard Constraint System Games}\label{sec:BCSgameDef}

\paragraph{Two-prover nonlocal games.}
We formulate our separations in terms of interactive protocols between a verifier and two provers. In the single-round setting, such protocols are equivalent to two-prover nonlocal games. A single-round game is specified by finite question sets $\mathcal X,\mathcal Y$, answer sets $\mathcal A,\mathcal B$, a distribution $\mu$ over pairs of questions, and a verification predicate $V:\mathcal X\times\mathcal Y\times\mathcal A\times\mathcal B\to\{0,1\}$. The verifier samples $(x,y)\sim\mu$, sends $x$ to Alice and $y$ to Bob, receives answers $a$ and $b$, and accepts iff $V(x,y,a,b)=1$. Multi-round protocols extend this model by allowing several rounds of interaction, where future messages may depend on the transcript accumulated during previous rounds.

For a single-round game, a \emph{quantum strategy} is a shared bipartite state $\ket{\psi}\in\mathcal{H}_A\otimes\mathcal{H}_B$ with question-dependent POVM elements $(M^{x}_a, N^{y}_b)$ (projective measurements when applicable).
The \emph{success probability} is
\begin{equation*}
\Pr[\text{win}]=\E_{(x,y)\sim\mu}\sum_{a,b} V(x,y,a,b) \bra{\psi}(M^{x}_a\otimes N^{y}_b)\ket{\psi}.
\end{equation*}

\emph{Constant-depth strategies} restrict each prover to a $\QNC^0$ circuit acting on their local workspace and ancillae, followed by computational-basis measurements.

\paragraph{From a constraint system to a game.}
To any constraint system, we can associate a one-round game with two nonlocal players, Alice (A) and Bob (B), and a verifier (V), in either of two ways:
\begin{enumerate}
    \item \textbf{Constraint-variable game}: The verifier randomly samples a constraint $C_j$ and variable $v_i \in \mathcal{S}_j$. Alice is asked for a satisfying assignment to $C_j$, and Bob for an assignment to $v_i$. They win iff Alice’s assignment satisfies the constraint and agrees with Bob’s assignment on $v_i$.

    \item \textbf{Constraint-constraint game}: The verifier randomly samples two constraints $C_i,C_j$ with contexts $\mathcal{S}_i,\mathcal{S}_j$. Alice and Bob are asked for satisfying assignments to $C_i$ and $C_j$, respectively. They win iff the assignments agree on all overlapping variables in $\mathcal{S}_i \cap \mathcal{S}_j$.
\end{enumerate}

For constraint systems in which every context contains exactly two variables, we also define the \textbf{2-CS game}. The verifier samples a constraint and sends one of its variables to each player. Alice and Bob win iff their assignments satisfy the constraint.

Importantly, a line of work \cite{cleve2014characterization,coladangelo2017robust,culf2024re} has established a close correspondence between perfect strategies for nonlocal games and operator-valued solutions to the underlying constraint systems. In the classical setting, perfect strategies correspond to scalar satisfying assignments (e.g.\ $\{\pm 1\}$).

In the quantum setting, variables are assigned Hermitian operators, introducing a nontrivial compatibility requirement: observables appearing in a common constraint must be jointly measurable, and hence commute. This is commonly formalized using \emph{linear constraint systems} (LCS), where constraints take the form
\[
c_j = \prod_{i\in\mathcal{S}_j} V_i,
\]
with all observables in a constraint required to commute. Under these conditions, operator-valued solutions are equivalent to perfect strategies for the associated nonlocal game.

Concretely, if an LCS admits a quantum solution $\{A_i\}_i$ on a $d$-dimensional Hilbert space, the provers may share the maximally entangled state
\[
\ket{\Phi^+}=\frac{1}{\sqrt d}\sum_{i=0}^{d-1}\ket{ii},
\]
with Alice measuring the observables $\{A_i\}_i$ and Bob measuring $\{A_i^{\mathrm T}\}_i$. The resulting outcomes are assigned to the corresponding variables $V_i$.

\section{Constraint Systems from Group Presentations}\label{sec:constraint_system}
As outlined in \Cref{sec:BCSgameDef}, binary constraint systems with operator-valued solutions are a standard tool for constructing nonlocal games that exhibit separations in winning probability between classical and quantum strategies.

Constraint systems and their associated nonlocal games admit a natural algebraic description in terms of group presentations and representations \cite{cleve2014characterization,cleve2017perfect,coladangelo2017robust}. A constraint system may be viewed as a group presentation, where the variables serve as generators, while each constraint determines a relation, equivalently a relator after moving all terms to one side. Operator-valued solutions then correspond to (projective) representations of the presented group. From this perspective, the verifier’s constraints can be designed to force the provers’ observables to realize representations of a prescribed group, often arising from presentations with linear relators whose associated solution group is non-abelian. Such games are typically constructed to exclude classical perfect strategies while admitting quantum ones.

Understanding the structure of perfect quantum strategies is closely related to self-testing~\cite{vsupic2020self}, where one seeks to infer the underlying observables and shared state from the observed correlations. In concrete examples, such as the Mermin--Peres magic square game, this is achieved by identifying the relevant solution group and analyzing its representations, thereby characterizing the unique admissible operator-valued solutions.

\paragraph{Beyond Quantum--Classical Separations.}
The existence of a perfect strategy arising from a non-trivial representation of a non-abelian solution group may rule out only trivial or one-dimensional (i.e.\ classical) strategies, while still permitting representations implementable by simple quantum circuits. In particular, linear constraint systems based on linearly presented groups may admit structurally simple realizations compatible with low-depth unitary computation \cite{Zhang2024,cleve2014characterization}.\footnote{Zhang, Pan, and Liu~\cite{Zhang2024} rule out strategies restricted to Clifford-type provers, but their methods do not appear to extend to general circuit-depth lower bounds.} Our goal is more demanding than simply separating quantum from classical strategies: we seek a constraint system whose perfect operator-valued solutions are ``rigid'' enough to reveal fine-grained computational features of quantum provers playing the corresponding nonlocal game.  

For a quantum depth-sensitive separation, one must force strategies to realize representations with provable complexity lower bounds. A natural idea is to use groups whose non-trivial representations inherently require substantial resources to implement. Classical examples of this include switching networks and permutation groups, where realizing arbitrary permutations via local generators requires logarithmic depth \cite{waksman1968permutation}. However, enforcing such structure typically requires nonlinear constraints, and there is currently no general reduction from such systems to single-round nonlocal games. Existing reductions rely crucially on joint measurability of the observables appearing in each question, which is difficult to preserve while imposing nonlinear constraints.

Our approach resolves this tension by starting from an abelian core, which preserves the joint measurability of the relevant observables, and embedding it into a larger group that enables us to carefully control (``rigidify'') its admissible representations and the corresponding perfect operator-valued solutions of the constraint system. While this does not eliminate the need for non-linear constraints, a careful choice of embedding group allows them to be enforced and tested within a natural extension of the linear BCS setting.

\subsection{From Irreducible to Faithful Representations}
\label{subsec:irrep-to-faithful}
We begin with the abelian group $\mathbb Z_2^m$. Over $\mathbb C$, its irreducible representations are all one-dimensional. Any family of commuting involutions satisfying the group law would constitute a non-trivial representation, but this only imposes weak structural constraints on the winning strategy.

Instead, we will force the strategy to implement a specific \emph{faithful} (i.e.\ injective or 1-to-1) representation of dimension $|\mathbb Z_2^m|=2^m$. Faithfulness prevents the representation from factorising into a product over a non-trivial kernel and the corresponding quotient group. The canonical faithful unitary representation over $(\mathbb{C}^2)^{\otimes m}$ is diagonal in the computational basis indexed by $\mathbb Z_2^m$, mapping each group element to an involution (i.e.,\ a Hermitian unitary) that flips the phase of exactly one basis vector (i.e., a diagonal matrix with $\pm 1$ entries). These matrices commute, and in quantum terms they correspond to highly nonlocal diagonal phase gates, including multi-controlled phase gates (denoted $\mathsf{C}^{m-1}(\mathsf{Z})$). Conjugating by Hadamard gates yields the corresponding multi-controlled Toffoli gates $\mathsf{C}^{m-1}(\mathsf{X})$, which are diagonal in the $\mathsf{X}$ basis. Such gates implement AND-type Boolean functions and require logarithmic depth in standard circuit models \cite{Dutta2025,dutta2025optimaltdepth}.

We force the implementation of this faithful representation of $\mathbb Z_2^m$ by constructing an embedding of $\mathbb Z_2^m$ into a larger group with additional compatibility constraints. The resulting ``rigidity'' of the representation makes our construction depth-sensitive, by tying perfect success in the nonlocal game to implementing observables with genuinely global structure (equivalent to multi-controlled phase operators). Consequently, any perfect strategy must implement these operations rather than an equivalent but computationally simpler realization. As $m$ (and hence the size of the constraint system, and corresponding nonlocal game instance) grows, the enforced operations become more global, and our lower bound exploits this to obtain a fine-grained separation between constant-depth provers.

Before proceeding further, let us first define the algebraic objects used in our construction.

\paragraph{Presentations, free products, amalgams, and semidirect products.}
A \emph{finite group presentation} $\langle S\mid R\rangle$ specifies a group $G$ by generators $S=\{s_1,\dots,s_k\}$ and relations $R=\{r_1,\dots,r_\ell\}$, where each $r_j$ is a word over $S\cup S^{-1}$ (i.e., words in the free group $F(S)$). It denotes the quotient group
\[
\langle S \mid R\rangle \;:=\; F(S)\big/\langle\!\langle R\rangle\!\rangle,
\]
where $\langle\!\langle R\rangle\!\rangle$ is the normal closure of $R$ in $F(S)$. A solution (or realization) of the presentation in a group $H$ is a map $\phi:S\to H$ that extends to a homomorphism from $F(S)\to H$ which maps $r_j$ to the identity element in $H$ for all $r_j\in R$.

If $G=\langle S\mid R\rangle$ and $H=\langle T\mid U\rangle$ with $S\cap T=\emptyset$, then the \emph{free product} $G*H$ has the presentation $\langle S\cup T \mid R\cup U\rangle$. For a group $K$, given homomorphisms (usually embeddings) $\iota_G:K\to G$ and $\iota_H:K\to H$, the \emph{amalgamated product} $G*_K H$ is the quotient of $G*H$ by the additional relations $\iota_G(k)\iota_H(k)^{-1}$ for all $k\in K$. We will refer to these as hybrid relations, and they serve to identify the two images of $K$ inside the free product.

Let $N$ and $H$ be groups and let $\varphi:H\to \mathrm{Aut}(N)$ be a homomorphism describing an action of $H$ on $N$ by automorphisms. The \emph{semi-direct product} $N\rtimes_{\varphi} H$ is the group whose underlying set is $N\times H$ with multiplication $(n,h)\cdot(n',h') \;:=\; \bigl(n\,\varphi(h)(n'),\, hh'\bigr)$.
Equivalently, $N\rtimes_{\varphi} H$ is generated by copies of $N$ and $H$ subject to the relations of $N$ and $H$ together with the conjugation rule $hnh^{-1} \;=\; \varphi(h)(n) ~\forall~n\in N,\; h\in H$, so that $N$ is a normal subgroup and $H$ acts on $N$ by conjugation.

\paragraph{Symmetry-determined embedding.} To eliminate the freedom to choose arbitrary representations of the abelian core, we embed $\mathbb Z_2^m$ into a larger group built from three components: $\mathbb Z_2^m$, the Pauli group, and the general affine group. The Pauli sector identifies the abstract observables with a canonical Pauli structure by identifying certain products of the generators of $\mathbb Z_2^m$ with Pauli operators. This still leaves freedom in how the individual observables are realized. We remove that freedom by imposing affine symmetries, which require the observables to transform consistently under relabelings of the Boolean hypercube. Together, these constraints force a unique faithful realization of the observable family while preserving the commutativity of the original $\mathbb Z_2^m$ sector, as we shall see below.

Let $\mathcal{P}_m$ denote the $m$-qubit Pauli group modulo global phase. Starting from the free product $\mathbb Z_2^m * \mathcal{P}_m$, which allows mixed words over the generators of $\mathbb Z_2^m$ and $\mathcal{P}_m$ while preserving the internal relations of each sector independently, we impose the hybrid relations
\begin{equation*}
    \prod_{\substack{\veci\in\{0,1\}^m:\ \veci\cdot \vecs = 1}} z_{\veci} = Z(\vecs)
\end{equation*}
for all $\vecs \in \{0,1\}^m$.
The elements $\chi_{\vecs} := \prod_{\veci:\veci\cdot \vecs=1} z_{\veci}$ form a subgroup of $\mathbb Z_2^m * \mathcal{P}_m$ isomorphic to $\mathbb Z_2^m$. The hybrid relations identify this subgroup with the Pauli subgroup generated by $\{Z(\vecs)\}_{\vecs\in\mathbb Z_2^m}$. The resulting group is the amalgamated product
\begin{equation*}
    \mathbb Z_2^m *_{\mathbb Z_2^m} \mathcal{P}_m,
\end{equation*}
where $\mathbb Z_2^m$ is embedded into itself via $\vecs \mapsto \chi_{\vecs}$ and into $\mathcal{P}_m$ via $\vecs \mapsto Z(\vecs)$, and the amalgamation identifies the images of the two embeddings.

We then extend this structure by affine symmetry. Let
$\mathrm{AGL}(m,2)=\mathbb Z_2^m\rtimes \mathrm{GL}(m,2)$
denote the general affine group, whose elements are pairs $g=(A,\vecb)$ with $A\in\mathrm{GL}(m,2)$ and $\vecb\in\mathbb Z_2^m$. We use $\mathrm{AGL}(m,2)$ because its action preserves the affine structure of $\mathbb Z_2^m$ underlying the hybrid relations introduced above. The affine group acts on the hypercube sector $\mathbb Z_2^m$ by
\[
g z_{\veci} g^{-1}=z_{\tau(g,\veci)},
\qquad
\tau(g,\veci)=A\veci\oplus\vecb.
\]

It also acts on the Pauli sector through the corresponding Clifford action. Writing $\alpha(g,\vecs)$ for the induced action on Pauli labels, we have 
\[
g\,P(\vecs)\,g^{-1} = P(\alpha(g,\vecs)),
\]
for any $\vecs \in \mathbb{Z}_2^{m}$ indexing Pauli operators up to phase. This results in a semi-direct product group: 
\begin{equation}
  G_{\mathbb Z_2^m}=  \bigl(\mathbb Z_2^m *_{\mathbb Z_2^m} \mathcal{P}_m \bigr)\rtimes_{(\tau,\alpha)} \mathrm{AGL}(m,2).
\end{equation}

$G_{\mathbb Z_2^m}$ is non-abelian, due both to the intrinsic commutation relations of the Pauli group and the affine conjugation action. But the $\mathbb Z_2^m$ subgroup itself remains abelian. Consequently, the distinguished observable family $\{z_{\veci}\}_{\veci\in\mathbb Z_2^m}$ is jointly measurable, as required in the standard LCS-based nonlocal games. Moreover, the Pauli and Clifford symmetry constraints enforcing rigidity are efficiently implementable, enabling their direct incorporation into an interactive protocol. 

Next, we will show that this enlarged group structure precisely rigidifies the representations of $\mathbb Z_2^m$ to the canonical faithful representation described earlier.

\subsection{Uniqueness of Operator-Valued Solutions}
We now describe a finite presentation of $G_{\mathbb Z_2^m}$, which will subsequently be compiled into a nonlocal game. Rather than presenting the group abstractly, we encode its defining relations directly as an operator-valued constraint system. The Pauli sector is represented by Pauli operators, and the affine symmetries are represented by the corresponding Clifford operators in the computational basis. The resulting constraints encode their action on the generators of the $\mathbb Z_2^m$.

\begin{definition}
[Faithful $\mathbb Z_2^m$ operator-valued BCS]\label{def:OVCS_Z2}
Let $m\in\mathbb N,m\geq2$ and $M=2^m$. The constraint system is specified by operator-valued variables $\{z_{\veci}\}_{\veci\in\{0,1\}^m}$ acting on a Hilbert space of dimension $M$, subject to four families of relations: 
\begin{equation*}
\begin{aligned}
\textbf{(Involution relations)} \qquad
& z_{\veci}^2 = I, && \forall \veci\in\{0,1\}^m,\\
& z_{\veci} z_{\vecj} = z_{\vecj} z_{\veci}, && \forall \veci,\vecj\in\{0,1\}^m\\[0.4em]
\textbf{(Parity–Pauli identification)} \qquad
& \prod_{\veci: \veci\cdot\vecs = 1} z_{\veci} = Z(\vecs), 
&& \forall \vecs \in \{0,1\}^m\setminus 0^m,\\[0.6em]
\textbf{(Affine symmetry relations)} \qquad
& \mathsf{Swap}_{a,b}\, z_{\veci}\, \mathsf{Swap}_{a,b} 
= z_{\sigma_{a,b}(\veci)}, 
&& \forall \veci\in\{0,1\}^m,\; a,b\in[m],\; a\leq b,\\
& X_a\, z_{\veci}\, X_a = z_{\veci\oplus\mathbf{e}_a}, 
&& \forall \veci\in\{0,1\}^m,\; a\in[m],\\
& \mathsf{CNOT}_{a,b}\, z_{\veci}\, \mathsf{CNOT}_{a,b} 
= z_{\veci\oplus(\veci[a]\cdot\mathbf{e}_b)}, 
&& \forall a,b\in[m],\; a\neq b,\\[0.4em]
\textbf{(Global constraint)} \qquad
& \prod_{\veci \in \{0,1\}^m} z_{\veci} = -I.
\end{aligned}
\end{equation*}
\end{definition}

The four families of relations play complementary roles and together leave only the canonical faithful realization of $\mathbb Z_2^m$ as a valid operator-valued solution. The involution and commutation relations define an abelian reflection group on the generators $\{z_{\veci}\}_{\veci}$ and imply that the observables are simultaneously diagonalizable. The parity–Pauli identification ties this abstract structure to the Paulis diagonal in the computational basis and restricts the admissible diagonal $\pm 1$-spectra. The affine symmetry relations enforce consistency under the action of $\mathrm{AGL}(m,2)$ on $\{z_{\veci}\}_{\veci}$, ultimately forcing these spectra to match the canonical basis labelling. Finally, the global constraint excludes the trivial representation. The proof of rigidity follows precisely this sequence of steps. A concrete illustration of the rigidity mechanism for $m=3$ is given in \Cref{app:example}.

\begin{lemma}[Rigidity of the faithful $\mathbb Z_2^m$ constraint system]
\label{lem:Z2_faithful_rep}
Any operator-valued solution to the constraint system in \Cref{def:OVCS_Z2} is unitarily equivalent to the canonical faithful representation of $\mathbb Z_2^m$. In particular, there exists a unique solution in which
\begin{equation}
z_{\veci} = \operatorname{diag}(1,\dots,1,\underbrace{-1}_{i\text{-th position}},1,\dots,1),
\qquad \forall i \in [M].
\end{equation}
\end{lemma}

\begin{proof}
The $z_{\veci}$ are unitaries over the Hilbert space $\mathbb C ^M$ with $M=2^m$. Since $z_{\veci}^2 = I$ for all $i\in[M]$, each $z_{\veci}$ is Hermitian and has eigenvalues $\pm1$. The relations $z_{\veci} z_{\vecj} = z_{\vecj} z_{\veci}$ further imply that these operators are simultaneously diagonalizable by a unitary $U$, such that $D_{\veci} := U z_{\veci} U^\dagger$ is diagonal for every $i$.

Conjugating the constraint $\prod_{\veci:\, \veci\cdot \vecs = 1} z_{\veci} = Z(\vecs)$
by $U$ yields for all $s\in\{0,1\}^m$ that
\[
\prod_{\veci:\, \veci\cdot \vecs = 1} D_{\veci} = U Z(\vecs) U^\dagger.
\]
The left-hand side is diagonal, and hence $U Z(\vecs) U^\dagger$ is diagonal for
all $\vecs$. By linearity, $U$ also conjugates every diagonal matrix $\sum_{{\vecs}\in \{0,1\}^m} c_{\vecs} Z(\vecs)$ with $c_{\vecs}\in\mathbb C$ to another diagonal matrix.

In particular, any rank-$1$ computational basis projector $P_{\mathbf{k}}=\ket{\mathbf{k}}\bra{\mathbf{k}}$ maps to a diagonal matrix.
Since $(UP_{\mathbf{k}}U^\dagger)^2=UP_{\mathbf{k}}U^\dagger$ and $\Tr{UP_{\mathbf{k}}U^\dagger}=\Tr{P_{\mathbf{k}}}=1$, we have that $UP_{\mathbf{k}}U^\dagger=\ket{\vecj}\bra{\vecj}$ also projects onto a computational basis vector. Hence, we have that
\begin{equation}
    U \ket{\mathbf{k}}=e^{i\theta_{\mathbf{k}}}\ket{\sigma(\mathbf{k})}, \text{ for some permutation } \sigma\in S_M, \theta_{\mathbf k}\in[0,2\pi)\;.
\end{equation}
Hence $U$ has only one non-zero entry per row and column, since $\bra{\vecj}U\ket{\veck}= \bra{\vecj}e^{i\theta_\veck}\ket{\sigma(\veck)}= \delta_{(\sigma(\veck),\vecj)}\cdot e^{i\theta_\veck}$. Furthermore, we may write $U=\mathsf{D}_U\Sigma$ with $\Sigma$ a permutation matrix, and $\mathsf{D}_U$ a diagonal unitary matrix.
As $z_{\veci}= U^\dagger D_{\veci}U= \Sigma^{-1}\mathsf{D}_U^\dag D_{\veci}\mathsf{D}_U\Sigma$, each $z_{\veci}$ is also diagonal. Combined with the knowledge of its spectrum, we have $z_{\veci}=\mathsf{diag}(\pm 1, \pm 1, ..., \pm 1)$.

To further determine the form of $z_{\veci}$, consider the restriction from the $\mathsf{Swap}$ operators:
\begin{align}
    \mathsf{Swap}_{a,b}\, z_{\mathbf{0}}\, \mathsf{Swap}_{a,b} &= z_{\mathbf{0}}, & a\leq b\leq m.
\end{align}
This implies that the operator $z_{\mathbf{0}}$ remains invariant when the qubits are permuted around. Let $z_{\mathbf{0}}=\mathsf{diag}(y_{0..0}, ...,y_{\veck},...,y_{1..1})$ with $\veck\in \{0,1\}^m$ and $y_{\veck}\in \{\pm1\}$. Then $y_{\vecj}=y_{\veck}$ if $|\vecj|=|\veck|$. Therefore, there are at most $(m+1)$ distinct diagonal matrices $z_{\mathbf{0}}$, with entries $f^{0},...,f^{m}\in \{\pm1\}$ corresponding to the $m+1$ Hamming weight slices, $y_{\veck}= f^{i}$ for $|\veck|=i\in[m]$.

By the relators $\mathsf{CNOT}_{a,b} z_{\veci} \mathsf{CNOT}_{a,b} = z_{\tau_{a,b}(\veci)}$, the $\mathsf{CNOT}$ operators can map any index $i$ with Hamming weight $|\veci|\geq 1$ to any other index $j$ with $|\vecj|\geq 1$. But since $z_{\mathbf{0}}$ remains invariant, while its diagonal entries other than $y_{00\ldots 0}$ are permuted, we see that
\begin{equation}
    f^1=f^2=...=f^{m}\;.
\end{equation}

All the operators $z_{\veci}$ have the same properties as $z_{\mathbf{0}}$, up to a permutation of the diagonal entries, as
\begin{equation}
 X_a\, z_{\veci}\, X_a = z_{\veci\oplus \mathbf{e}_a},
\end{equation}
where $a\in [m]$ allows us to map to any other $z_{\veci}$ starting from $z_{\mathbf{0}}$. 

Finally, since $\prod_{\veci \in \{0,1\}^m} z_{\veci} = -I$, we have that $f^1=f^2=...=f^{m}=-f^0$. This follows from the fact that the product of the first entry of the diagonal of all $z_{\veci}$ is equal to $-1$. Furthermore, because the $z_{\veci}$ are the same as $z_{\mathbf{0}}$ with diagonal entries permuted, all their first entries are just all the values along the diagonal of $z_{\mathbf{0}}$. The set of representation matrices for the $z_{\veci}$ thus corresponds exactly to the canonical faithful representation of $\mathbb Z_2^m$, up to permutation and a global sign. 
\end{proof}

\section{Quantum Interactive Protocol with Clifford Verifier}\label{subsec:Quantum_depth}

We begin with a quantum verifier that can implement constant-depth Clifford circuits and is allowed quantum communication with two non-communicating provers. By granting the verifier the ability to prepare quantum states, we construct a single-round quantum interactive protocol that enforces increasing depth (or coherence-time) requirements on the provers. We then translate this interactive protocol into a relation problem with quantum input and classical output in the plain non-interactive model, thereby establishing a quantum depth hierarchy for the class of $\QNC^0$ circuits.

\subsection{Single-round Quantum Interactive Protocol}\label{subsec:Int_Clifford}
To derive our protocol, we adapt the standard correspondence between linear operator-valued constraint systems and two-prover nonlocal games (see \Cref{sec:BCSgameDef} and, e.g., \cite{cleve2014characterization, coladangelo2017robust, culf2024re}) to a setting with a quantum verifier. A direct translation of \Cref{def:OVCS_Z2} would require the verifier to query both the observables ${z_{\veci}}$ and the Clifford symmetry operators $\{\mathsf{CNOT},\mathsf{Swap},X\}$. However, rather than treating all generators in \Cref{def:OVCS_Z2} as possible question labels, we let the verifier absorb all $\{\mathsf{CNOT}, \mathsf{Swap}, X\}$ operators into its own state preparation procedure. 

The protocol consists of a state-preparation phase and a constraint-testing phase. Note that the Clifford operators only specify the local measurement bases for the Bell pairs shared between the provers (i.e. players in the nonlocal game), and do not introduce additional algebraic degrees of freedom into the provers' answers. Since the verifier is allowed to prepare (and hence fix) the entangled resource state shared by the provers, these basis changes can be incorporated directly into its state preparation step. Consequently, the verifier may restrict the question set to the generators ${z_{\veci}}$ alone.

\begin{definition}[Question sampling for the $\mathbb{Z}_2^m$-operator-valued BCS]
\label{def:QS_Z2OCS}

Let $\{z_{\veci}\}_{\veci\in\mathbb{Z}_2^m}$ be the generators of the constraint system from \Cref{def:OVCS_Z2}. The question set $\mathcal{Q}$ consists of ordered tuples $q=(r,q_A,q_B,U_A,U_B)$,
where $q_A,q_B$ are generator or constraint labels and $U_A,U_B$ specify the local Clifford frame used during state preparation. The verifier's procedure for sampling $q$ is as follows.

\begin{enumerate}
\item Sample a prover index $r\in\{A,B\}$ uniformly at random and let $\bar r=\{A,B\}\setminus{r}$.

\item Sample $q_r$ uniformly at random from one of the following three types:
\begin{itemize}
\item \emph{Type 1: Label of the global constraint}
\[
\prod_{\veci\in\{0,1\}^m} z_{\veci}=-I.
\]

\item \emph{Type 2: label $\vecs\in\mathbb Z_2^m$ of a parity--Pauli identification constraint}
\[
\prod_{\veci:\,\veci\cdot\vecs=1} z_{\veci}=Z(\vecs).
\]

\item \emph{Type 3: label ${\veci}$ of a generator $z_{\veci}$}.
\end{itemize}

\item Set $U_r=I$ and generate $(q_{\bar r},U_{\bar r})$ according to the type of $q_r$:
\begin{itemize}
\item If $q_r$ is a constraint query (type 1 or 2), sample $q_{\bar r}$ uniformly at random from the generators appearing in that constraint and set $U_{\bar r}=I$.

\item If $q_r$ is a generator query (type 3), sample $U\in\{I,\mathsf{Swap}_{a,b},X_a,\mathsf{CNOT}_{a,b}\}_{a,b\in[m]}$ uniformly at random, set $U_{\bar r}=U$, and let $q_{\bar r}$ be the label of the generator $Uz_{\veci}U^\dag$.
\end{itemize}
\end{enumerate}
\end{definition}

\paragraph{Protocol description.} Our interactive protocol is described fully in \Cref{def:QI_LBCS}. In each round, the verifier first prepares an $m$-fold tensor product of Bell pairs, $\ket{\Phi}_{AB}:=\bigotimes_{i=1}^m \ket{\Phi^+}_{A_i B_i}$, with
\begin{equation*}
    \ket{\Phi^+}_{A_i B_i}
    = \frac{\ket{0}_{A_i}\ket{0}_{B_i}+\ket{1}_{A_i}\ket{1}_{B_i}}{\sqrt{2}}.
\end{equation*}
Here $\ket{\Phi}_{AB}\in \mathcal{H}_A\otimes\mathcal{H}_B$ across the Alice-Bob bipartition, with $A := A_1\cdots A_m$ and $B := B_1\cdots B_m$. The verifier then applies unitaries $U_A$ and $U_B$ from the question tuple $q=(r,q_A,q_B,U_A,U_B)$ to subsystems $A$ and $B$, respectively. Let $\Phi_{AB}$ denote the joint density operator following the verifier's processing of the sampled questions, and let $\leftarrow$ denote the assignment operation. We have
\begin{equation}\label{eq:stateprepare}
    \Phi_{AB}\leftarrow (U_A\otimes U_B)\left(\bigotimes_{i=1}^m \ket{\Phi^+}\bra{\Phi^+}_{A_i B_i}\right)(U_A\otimes U_B)^{\dagger}\;.
\end{equation}
For a type-2 question, the verifier further measures $Z(\vecs)$ on subsystem $r$:
\begin{equation}\label{eq:stateprepare2}
    \Phi_{AB}\leftarrow  (\hat{Z}(s)\otimes I_{\bar{r}})\Phi_{AB}(\hat{Z}(s)\otimes I_{\bar{r}})/\tr((\hat{Z}(s)\otimes I_{\bar{r}})\Phi_{AB}),
\end{equation}
and records the outcome $\hat{Z}(s)$. Note that the verifier does not inform the provers whether a measurement has been performed. The verifier then sends the quantum system and question label $(A, q_A)$ to Alice, and $(B, q_B)$ to Bob. 

The verifier checks output correctness as follows.
\begin{itemize}
    \item If $q_r$ is of types 1 or 2 in \Cref{def:QS_Z2OCS}, following the structure of a constraint-variable game, the verifier checks that prover $r$ outputs an assignment satisfying the constraint indexed by $q_r$. Simultaneously, the verifier checks whether the complementary prover $\bar{r}$ provides a value for the generator $q_{\bar{r}}$ that is identical to the assignment of this generator from $r$.

    \item If $q_r$ is of type 3, following the structure of a 2-CS game, the verifier checks if Alice and Bob return the same value to the generators they were asked for.
\end{itemize}

\begin{algorithm}[htb]
\floatname{algorithm}{Protocol}
\caption{Interactive Boolean Hypercube Game with Clifford-Verifier (\textsf{``Interactive HypGame''})}\label{def:QI_LBCS}
\begin{algorithmic}[1]
\Require Question set $\mathcal{Q}$, parameters $m,k\in \mathbb{N}$. 

\vspace{0.2mm}

\Statex {\bfseries Stage 1: Randomly choose a valid set of questions.}

\vspace{0.2mm}

\State Sample $k$ question tuples $\{q^i=(r^i,q_A^i,q_B^i,U_A^i,U_B^i)\in\mathcal{Q}\}_{i=1}^{k}$ as per \Cref{def:QS_Z2OCS}.\\
Let $t\gets 0$.

\vspace{2mm}

\For{$i=1:k$}

\vspace{0.2mm}
\Statex {\bfseries Stage 2: Prompt provers to answer all questions in the set.}

\vspace{0.2mm}
    
    \State The verifier prepares $\Phi_{AB}$ as per \Cref{eq:stateprepare,eq:stateprepare2}.
    
    \vspace{0.2mm}
    
    \If{$q_{r}^i$ belongs to Type 2}
      \State The verifier measures $Z(\vecs)$ on subsystem $r^i$.    
    \Else\ Continue
   \EndIf
   \State The verifier sends $(A,q_A^i)$ to Alice and $(B,q_B^i)$ to Bob.

\vspace{0.2mm}

\Statex {\bfseries Stage 3: Check the answers returned by the provers.}

\vspace{0.2mm}
  
  \State The provers send back their classical answers $a_A^i$ and $a_B^i$.
  \State The verifier checks the correctness: $t\gets t+1$ if $(a_A^i,a_B^i)$ is a valid answer. 
\EndFor
\State \Return the fraction of accepted rounds $t/k$.
\end{algorithmic}
\end{algorithm}

\smallskip
The interactive protocol above closely follows the mapping from operator-valued constraint systems to nonlocal games introduced in \Cref{sec:BCSgameDef}. This formulation provides a uniform presentation and allows us to directly leverage the standard analysis of constraint–variable and 2-CS games. In particular, the protocol inherits the usual correspondence between linear operator-valued constraint systems and nonlocal games: any perfect winning strategy induces an operator assignment satisfying all constraints in \Cref{def:OVCS_Z2}, and conversely, any satisfying assignment yields a perfect strategy. Consequently, by the rigidity result of \Cref{lem:Z2_faithful_rep}, every perfect strategy realizes the unique faithful representation of $\mathbb{Z}_2^m$ characterized therein.

We now state the completeness and soundness properties of the protocol.

\begin{lemma}\label{lemma:cliff_verif}
For any $m\in\mathbb{N}$, consider the associated interactive protocol in \Cref{def:QI_LBCS}. Then there exists $k$ large enough such that the following holds:

\begin{itemize}
\item \textbf{Perfect completeness:}
There exists a strategy for two provers in $\mathsf{QNC}^0$ with depth $d=O(\log m)$ over a finite universal gate set and a constant number of ancillary qubits, that causes the verifier to accept with probability~$1$.

\item \textbf{Soundness against shallow provers:}
For any pair of provers with circuit depths $d_1$ and $d_2$ composed of bounded fan-in gates, respectively, if $d'=\min(d_1, d_2) < \lfloor \log m \rfloor$, then any such strategy achieves success probability at most  $1-2^{-\Omega(m)}$, regardless of circuit size or the number of ancillary qubits.
\end{itemize}
\end{lemma}

\paragraph{Proof of completeness.} 
For the completeness analysis, the honest provers answer the queried generators $z_{\veci}$ according to the outcomes of measurements of the Hermitian operators from \Cref{lem:Z2_faithful_rep} applied to the quantum state received from the verifier.

\smallskip
\noindent\textit{Constraint satisfaction.} We show that the provers shall reply correctly to all questions in all cases. Without loss of generality, we assume that the first step samples $r$ to be Alice.

\begin{enumerate}
\item Global constraint: 
    
The verifier will distribute $m$-fold tensor products of Bell pairs $\bigotimes_{i=1}^{m}\ket{\Phi^+}_{A_iB_i}$ to the provers. By construction, Alice will output a satisfying assignment to the constraint. Note that the corresponding observables of the set of generators, $z_{\veci}$, are compatible and can be measured simultaneously. Moreover, the provers will assign consistent values to the same generator as
    \begin{equation}\label{eq:EPRtrick}
        _{AB}\bra{\Phi} z_{\veci} \otimes z_{\veci}\ket{\Phi}_{AB} =\tr(z_{\veci}z_{\veci}^{\mathrm{T}})/M=1,
    \end{equation}
    with $z_{\veci}=z_{\veci}^{\mathrm{T}}$.

\item Parity-Pauli identification: 
    
In the case where $q_A$ refers to the constraint $\prod_{\veci: \veci\cdot \vecs = 1} z_{\veci}=Z(\vecs)$ for some $\vecs$, the verifier will have measured $Z(\vecs)$ on system $A$ of the Bell pairs before sending them to the provers. Then,
\begin{align}
  _{AB}\bra{\Phi} Z(\vecs)  \left (\prod_{\veci: \veci\cdot \vecs = 1} z_{\veci}  \right )\otimes I \ket{\Phi}_{AB} =_{AB}\bra{\Phi} I  \otimes I \ket{\Phi}_{AB}=1.
\end{align}
Therefore, the parity of the measurement outcomes of $z_{\veci}$ for $\veci\cdot \vecs=1$ will be equal to the measurement outcome of $Z(\vecs)$. The values assigned by the provers to the same generator will again be equal as in Eq.~\eqref{eq:EPRtrick}.

\item Involution or affine symmetry relations: 

The verifier will apply the sample unitary $U\in\{I, \mathsf{Swap}_{a,b}, \mathsf{CNOT}_{c,d}, X_a\}$ to system $B$. Then, 
\begin{equation}
  _{AB}\bra{\Phi} z_{\veci} \otimes U^\dagger z_{\vecj} U \ket{\Phi}_{AB} = \tr(U^\dagger z_{\vecj} U z_{\veci}^{\mathrm{T}})/M=1,
\end{equation}
as the generator sent to Bob, $z_{\vecj}$, is selected based on \Cref{def:OVCS_Z2} such that $U z_{\veci} U^\dagger=z_{\vecj}$. Note that one can equally check the validity of the previous assertions by applying the Heisenberg picture to the observables $\{z_{\veci}\}$. 
\end{enumerate}

Thus, if the provers measure the observables obtained from 
\Cref{lem:Z2_faithful_rep} when asked the corresponding generators $z_{\veci}$, they will generate the correct outcomes and succeed in every round.

\medskip
\noindent\textit{Circuit depth analysis.}
Having established the correctness of the honest strategy, we now analyze the circuit depth required to implement it. By \Cref{lem:Z2_faithful_rep}, each question label $q_A$ specifies an observable $M_{q_A}$ acting on the received subsystem $\Phi_A$. Thus, the complexity of the honest strategy reduces to implementing these measurements efficiently.

The observables $M_{q_A}$ can be expressed in terms of generalized controlled-phase gates $\mathsf{C}^{m-1}(\mathsf{Z})$ acting on subsets of qubits associated with the generators $z_{\veci}$. Consequently, the dominant contribution to the circuit depth arises from realizing these multi-qubit operations.

To measure such an observable, we employ a phase-estimation procedure (specifically, the Hadamard test), illustrated in \Cref{fig:top_circuit}. The procedure uses an additional ancillary qubit initialized in $\ket{0}$ and applies a controlled implementation of $\mathsf{C}^{m-1}(\mathsf{Z})$, whose phase information is subsequently extracted by a computational-basis measurement on the ancilla. The controlled operation is equivalent to a generalized controlled-phase gate $\mathsf{C}^{m}(\mathsf{Z})$ and can, in turn, be realized by a generalized Toffoli gate $\mathsf{C}^{m}(\mathsf{X})$. Consequently, the circuit depth is determined by the implementation of this generalized Toffoli gate. Using the construction of \cite{nie2024quantum}, we give a precise upper bound in \Cref{eq:Toffolidepth}, which scales as $O(\log m)$. Furthermore, the constants in the scaling overhead are small, although they depend on the specific underlying gate set.

\medskip
\noindent\textit{Loading question labels.}
To implement the measurements within a uniform circuit family, we encode the question label $q_A$ as a computational-basis state $\ket{q_A}$ and use it to control a fixed circuit architecture. Rather than constructing a different circuit for each possible question, the same circuit is applied uniformly, with the input label determining the corresponding observable through a standard ``compute–phase–uncompute'' pattern.

\begin{enumerate}
\item \emph{Compute:} using the bits of $q_A$ as classical controls, apply a single layer of controlled-$X$ gates in parallel to the register $\Phi_A$.

\item \emph{Essential non-Clifford operation:} apply the phase-estimation circuit to measure the $\mathsf{C}^{m-1}(\mathsf{Z})$ observable, utilizing an ancillary qubit as described previously.

\item \emph{Uncompute:} apply the same layer of controlled-$X$ gates used in the compute step to invert the initial preprocessing.

\item \emph{Readout:} measure the designated output qubit(s) in the computational basis to obtain the answer $a_A$.
\end{enumerate}

Notably, in the second step, we apply a fixed circuit in which $\mathsf{C}^{m-1}(\mathsf{Z})$ utilizes the first $(m-1)$ qubits as controls and the final qubit as the target. The basis transformation from this fixed operator to other $z_{\veci}$ generators is mediated by the \emph{compute} operation. \Cref{fig:top_circuit,fig:bottom_circuit} illustrate the overall circuit and the circuit implementing the \emph{compute} operation respectively. Hence, the total depth of the implementation is the depth of $\mathsf{C}^{m-1}(\mathsf{Z})$ (given in \Cref{eq:Toffolidepth}) plus a constant overhead of five layers. Moreover, the construction requires only a single ancillary qubit.

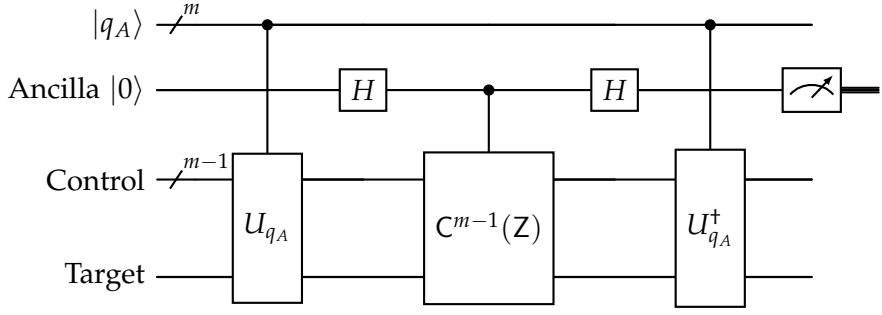
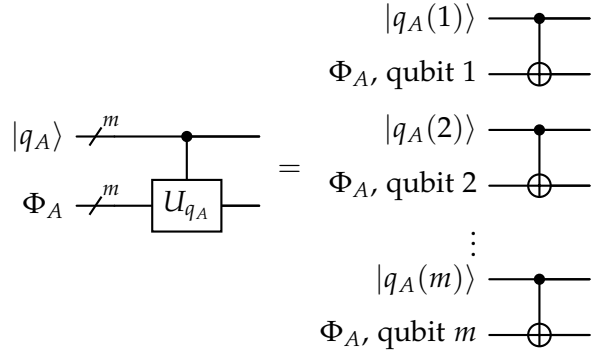
\begin{figure}[hbt!]
    \centering
    
    \begin{subfigure}{\textwidth}
        \centering
        \begin{quantikz}[baseline=(current bounding box.center)]
    \lstick{$\ket{q_A}$} & \qwbundle{m} & \ctrl{2} & \qw & \qw & \qw & \ctrl{2} & \qw \\
    \lstick{Ancilla $\ket{0}$} & \qw & \qw & \gate{H} & \ctrl{1} & \gate{H} & \qw & \meter{} & \cw \\
    \lstick{Control} & \qwbundle{m-1} & \gate[wires=2]{U_{q_A}} & \qw & \gate[wires=2]{\mathsf{C}^{m-1}(\mathsf{Z})} & \qw & \gate[wires=2]{U_{q_A}^\dagger} & \qw \\
    \lstick{Target} & \qw & \qw & \qw & \qw & \qw & \qw & \qw 
\end{quantikz}
        \caption{Overall circuit architecture. The $m-1$ control qubits and the target qubit compose $\Phi_A$. The question label state $\ket{q_A}$ is read and applied to $\Phi_A$ via a compute--phase--uncompute architecture; the \emph{compute} circuit of the controlled-$U_{q_A}$ is shown in \Cref{fig:bottom_circuit}, and \emph{uncompute} is its reverse operation. Measurement of $z_{\veci}$ is realized with a phase estimation protocol, using one ancillary qubit initialized to $\ket{0}$.}
        \label{fig:top_circuit}
    \end{subfigure}
    
    \vspace{1cm} 
    
    \begin{subfigure}{\textwidth}
        \centering
        \begin{quantikz}[baseline=(current bounding box.center)]
        \lstick{$\ket{q_A}$} & \qwbundle{m} & \ctrl{1} & \qw \\
        \lstick{$\Phi_A$} & \qwbundle{m} & \gate{U_{q_A}} & \qw
    \end{quantikz}
$=$
\begin{quantikz}[baseline=(current bounding box.center)]
    \lstick{$\ket{q_A(1)}$}   & \ctrl{1} & \qw \\
    \lstick{$\Phi_A$, qubit $1$}   & \targ{} & \qw \\
    \lstick{$\ket{q_A(2)}$}   & \ctrl{1} & \qw \\
    \lstick{$\Phi_A$, qubit $2$}   & \targ{} & \qw \\
    \lstick{\vdots} \\
    \lstick{$\ket{q_A(m)}$}   & \ctrl{1} & \qw \\
    \lstick{$\Phi_A$, qubit $m$}   & \targ{} & \qw
\end{quantikz}
        \caption{Circuit to realize the \emph{compute} operation. With the question label $q_A$ encoded as a bit string on the computational basis, the $i$'th qubit in $\ket{q_A}$ controls the $i$'th qubit in $\Phi_A$ via a $\mathsf{CNOT}$ gate.}
        \label{fig:bottom_circuit}
    \end{subfigure}
    
    \caption{The circuit to measure some $z_{\veci}$ upon input question label $q_A$.}
    \label{fig:full_framework}
\end{figure}

\medskip
\noindent
\textit{Depth-efficient enforcement of global constraints.}  We must enforce all constraints of the system in order to apply our rigidity lemma, which guarantees the uniqueness of the winning strategy and, in particular, forces the implementation of the $\mathsf{C}^{m-1}(\mathsf{Z})$ gate. The previous construction establishes how the honest provers can implement individual observables $z_{\veci}$. For most constraints, namely the involution and affine symmetry relations, this already suffices, since they require only a single observable measurement.

The global constraint and the parity–Pauli identification relations, on the other hand, comprise multiple $z_{\veci}$ generators. When $k$ $z_{\veci}$ observables are measured simultaneously, notice that they are all diagonal on the computational basis. Hence, one can introduce $k$ ancillary qubits and perform phase estimation in parallel while sharing a single generalized phase gate $\mathsf{C}^{m-1+k}(\mathsf{Z})$. However, what is problematic here is that these two types of constraints involve products over generators that are exponentially many in $m$. If the provers were required to explicitly compute and output all corresponding eigenvalues, this would increase the depth of the honest strategy to $O(m)$.

To establish a $\Theta(\log m)$ depth bound, we first observe that in this protocol, the prover's response to either a global constraint or a parity-Pauli identification constraint can be modeled as the outcome of measuring a fixed observable. This measurement simultaneously determines all relevant generator outcomes, inducing a fixed ordering of the generators and a corresponding output string that encodes the eigenvalues of the $z_{\veci}$ operators. Crucially, as long as the verifier's subsequent classical post-processing is reversible, it does not alter the validity of the prover's strategy. We are therefore permitted to compose the prover's measurement outcomes with an arbitrary reversible map selected by the verifier, provided the map is invertible.

We exploit this degree of freedom by selecting a reversible transformation that maps the target observable outcomes to those obtained from computational-basis measurements. Such a transformation exists because computational-basis measurements are sufficient to reconstruct the eigenvalues of any set of commuting diagonal observables. Equivalently, the verifier can be viewed as applying the inverse of this transformation to the interpretation of the provers' responses. Under this construction, the honest provers are only required to perform computational-basis measurements, while the verifier's classical post-processing reconstructs the outcomes corresponding to the original observables. Consequently, the exponential-size product constraints are enforced in the same $O(\log m)$-depth circuits of the provers without requiring additional quantum circuit depth from them.

Therefore, without modifying the underlying operator-valued solution or any essential property of the protocol, the implementation of these additional observables can be absorbed into the verifier's classical post-processing.\hfill\qedsymbol

\paragraph{Proof of soundness.} For soundness, we show that if the non-communicating provers succeed perfectly in the interactive protocol, then they must realize the observables of \Cref{lem:Z2_faithful_rep}, up to local isometries. We further establish a robust version of this statement, showing that any near-perfect strategy must implement observables that remain close to the ideal realization. Consequently, both perfect and near-perfect success require the implementation of observables whose exact or approximate realization demands a minimum circuit depth.

\medskip
\noindent\textit{Non-contextual strategies.} 
We assume by default that the verifier follows the protocol honestly during both question sampling and the prompting of the provers. The soundness of our protocol then follows from the arguments established in \cite{cleve2014characterization, coladangelo2017robust, culf2024re}. As first proven in \cite[Theorem 1]{cleve2014characterization}, for a one-round constraint-variable nonlocal game defined by a non-trivial linear constraint system, a perfect winning strategy implies that up to a local isometry, the provers must apply a tensor product of EPR states ($\ket{\Phi^+}$); furthermore, the observables they measure must be non-contextual\footnote{In \cite{cleve2014characterization}, non-contextuality refers to the requirement that if different questions (or ``contexts'') overlap on a specific variable, the prover must measure the same observable to assign a value to that variable, regardless of the context. Note that this definition of non-contextuality in \cite{cleve2014characterization} is the opposite of that found in works such as \cite{budroni2022kochen}; here, we adhere to the terminology of the former.} and must correspond to a satisfying assignment to the linear constraint system. Recent work in \cite{culf2024re} further establishes this correspondence for perfect strategies in constraint-constraint nonlocal games and 2-CS games, mapping them to satisfying assignments of the underlying system.

Essentially, the soundness of \Cref{def:QI_LBCS} is established by applying the results of \cite{cleve2014characterization, culf2024re} to the distributed states and question labels at the commencement of Stage 3. This setting deviates from prior work in that the verifier may apply non-trivial unitaries or measurements, meaning the distributed state is not necessarily a tensor product of EPR states ($\ket{\Phi^+}$). Nevertheless, we argue that the provers gain no additional information regarding the questions assigned to them by receiving such states.

\begin{enumerate}
    \item Affine symmetry relations:

    When these questions are sampled, the verifier applies a specific unitary to one of the provers' shares of the Bell states. In such cases, the distributed quantum state remains maximally entangled, subject only to a local basis change. Crucially, the question labels $q_A$ and $q_B$ refer to randomly sampled generators $z_{\veci}$, and their marginal distribution is identical to the case where an involution relation $z_{\veci}^2=I$ is sampled. Consequently, the provers cannot infer the basis change from their own received states, which remain maximally mixed, nor from the generator labels $q_A$ and $q_B$\footnote{Notably, the provers cannot gain an advantage here even if they pre-share entanglement, although this scenario is not considered in our circuit model, where entanglement preparation itself necessitates a specific circuit depth.}. Up to the local unitary applied by the verifier (as specified by the sampled affine symmetry relation), the question is equivalent to an involution relation. Thus, the soundness for these question types follows from the 2-CS game results established in \cite{culf2024re}.
    
    \item Parity-Pauli identification relations:
    
    When these questions are sampled, the verifier performs a Pauli measurement of the type $Z(\mathbf{s})$ on either subsystem $A$ or $B$ of the shared Bell states. As in the preceding argument, each prover receives a local state that is maximally mixed. Consequently, from either their local state or the question label, a prover asked a label corresponding to a generator $z_{\veci}$ cannot determine which specific constraint the generator is associated with. Therefore, as before, their strategy must be non-contextual. By considering the provers' measurements in conjunction with the verifier's prior measurement, the soundness for these question types follows from the constraint-variable game results established in \cite{cleve2014characterization}.
\end{enumerate}

In summary, the provers can cause the verifier to accept with probability one only if they implement a non-contextual strategy, thereby restoring the usual correspondence between perfect strategies for the nonlocal game and operator-valued solutions of the underlying constraint system. By \Cref{lem:Z2_faithful_rep}, the operator-valued assignments satisfying the system of \Cref{def:QI_LBCS} are uniquely determined up to local isometries. Consequently, the interactive protocol itself admits a unique rigid solution: any provers that fail to implement observables equivalent to the prescribed measurement operators cannot achieve a perfect acceptance probability.

\medskip
\noindent\textit{On the soundness of depth-efficient enforcement of global constraints.}
We further clarify that the protocol modification designed to enforce an $O(\log m)$ depth bound in the complete protocol implementation does not compromise soundness. In the case of parity–Pauli identification constraints, the verifier performs a measurement of the observable $Z(\textbf{s})$ prior to distributing the state. Consequently, the prover's output string must be consistent with this latent outcome; the correct string cannot be generated without performing a computational-basis measurement on the received quantum system. Simultaneously, the second prover is queried on a single generator within the same constraint, which facilitates a consistency check between the two responses. This prover, who is ignorant of the question asked to the other prover, must be measuring the specified $z_{\veci}$ generator as characterized in \Cref{lem:Z2_faithful_rep} without any modification. This mechanism ensures that the constraint is satisfied in a non-contextual manner, preserving the integrity of the preceding analysis.

\medskip
\noindent\textit{Circuit depth lower bound on the measurements.} 
To realize the rigid perfect strategy and hence to achieve perfect success, the provers must implement the full family of multi-controlled phase operators, or equivalently generalized Toffoli gates, up to a local isometry. For bounded fan-in circuits, the exact implementation of these transformations admits depth lower bounds that are invariant under such isometries. We establish this formally in \Cref{prop:Toffoli_lower}, showing that any realization requires circuit depth at least $d \ge \lfloor \log m \rfloor$. Consequently, no quantum circuit of depth $o(\log m)$ can implement the required observables, completing the soundness argument in the rigid setting.

Moreover, \Cref{lem:RQP} establishes a robust version of the protocol; see \Cref{subsec:rob_Q} for the full analysis. Specifically, we show that any strategy achieving success probability at least $1-\varepsilon$ must implement observables that are $2^{O(m)}\mathsf{poly}(\varepsilon)$-close to the rigid solution of the underlying constraint system, and hence equally close, up to local isometries, to the full family of multi-controlled phase gates certified by the protocol. We then combine this robustness statement with approximate unitary-synthesis lower bounds. In \Cref{thm:robustdepth}, we show that any bounded fan-in circuit realizing these transformations within error $\delta< 2\sqrt{2}\cdot2^{-m/2}$ must still have depth at least $d \ge \lfloor \log m \rfloor$. Consequently, no quantum circuit of depth $o(\log m)$ can achieve success rate above $1-\varepsilon(m)$ for $\varepsilon(m)=2^{-O(m)}$ in the protocol, completing the soundness argument in the robust setting.
\hfill \qedsymbol

\subsection{Towards Removing Interactivity: a Quantum Relation Problem}\label{sub:quantum_plain}

Next, we translate the previous interactive protocol into a computational problem with a quantum input and a classical output in the non-interactive model. Specifically, we consider an input quantum state of $n$ qubits. The input state essentially encodes the $2m$-qubit quantum state prepared by the verifier with respect to the sampled question in the interactive protocol. However, the position of this state will be randomized over the qubit indices of the input state. Other than the qubits prepared in the question state, the remaining qubits among the $n$ qubits will be prepared as the computational basis state $\ket{0}$. 

\begin{definition}[Quantum Boolean Hypercube Relation Problem $\mathcal{R}_m$]
\label{def:R_m}
Fix $m,n \in \mathbb{N}$. The quantum-input relation problem $\mathcal{R}_m \subseteq \mathcal{H}^n \times \{0,1\}^n$ is a relation over inputs of size $n>4m$ defined as follows:

Each input consists of a quantum state of the form
\[
\ket{\psi}_{i,j;\,q_A,q_B}
= \ket{0}^{\otimes s_1} \otimes \underbrace{ \ket{q_A}}_{i} \otimes\ \Phi_A 
\otimes \ket{0}^{\otimes s_2} \otimes \underbrace{\ket{q_B} }_{j} \otimes\ \Phi_B 
\otimes \ket{0}^{\otimes s_3},
\]
where:
\begin{itemize}
    \item $(q_A,q_B)$ are classical question labels sampled as in \Cref{def:QS_Z2OCS}, encoded in the computational basis on designated registers $i$ and $j$ (we denote the registers with the indices of the first qubits in $\ket{q_A}$ and $\ket{q_B}$), respectively;
    \item $\Phi$ denotes the $2m$-qubit density operator prepared as in Eq.~\eqref{eq:stateprepare} and \eqref{eq:stateprepare2} with the $A/B$ partition as in \Cref{def:QI_LBCS} with $\Phi_A=\tr_B(\Phi_{AB}),\Phi_B=\tr_A(\Phi_{AB})$;
    \item the remaining qubits are initialized to $\ket{0}$.
\end{itemize}

An output string $\mathbf{z} \in \{0,1\}^n$ satisfies the relation if it is of the form
\[
\mathbf{z} = 0^{s_1}\, a_A \, 0^{s_2}\, a_B \, 0^{s_3},
\]
where $(a_A,a_B)$ is a valid accepting pair of answers for the nonlocal game instance defined by $(q_A,q_B)$.
\end{definition}

Following the results in \Cref{subsec:Int_Clifford}, it is natural to expect that a sufficiently expressive quantum computational model can efficiently solve all instances of this family of problems. Indeed, solving a given instance is equivalent to implementing the honest provers' strategy from \Cref{lemma:cliff_verif} for the corresponding interactive protocol. Note that as the input size $n$ grows, the relevant hardness parameter $m = O(1)$ remains constant; this parameter determines both the underlying question set and the minimal circuit depth required to implement the associated strategy. Consequently, any quantum circuit architecture with sufficient depth to implement the multi-controlled phase gates $\mathsf{C}^{m-1}(\mathsf{Z})$ can solve the problem exactly.

In contrast, we demonstrate that shallower devices cannot solve the relation problem perfectly. Specifically, $\mathsf{QNC}^0$ circuits of depth $o(\log m)$ fail to solve the problem exactly due to their depth restrictions and the resulting constraints on connectivity. Intuitively, a $\mathsf{QNC}^0$ circuit with depth constant in the input size $n$ cannot propagate information across all relevant qubit indices as $n$ increases; it therefore cannot coordinate the nonlocal operations required to circumvent the depth requirements of a winning strategy. Consequently, the circuit must resort to the unique accepting strategy of the interactive protocol which, as analyzed in \Cref{subsec:Int_Clifford}, cannot be realized with depth $o(\log m)$.

More formally, our proof relies on a light cone argument. We show that, with high probability over the qubit embeddings, the questions and Bell pairs are processed by two disjoint regions of the circuit whose light cones do not intersect. The circuit thus behaves as two non-communicating quantum devices acting locally on a shared entangled state. Each region is therefore restricted to the computational power of a local prover, precisely the setting characterized in \Cref{lemma:cliff_verif}.

\begin{theorem}
[Depth separation via quantum inputs]\label{thm:QPlain}
Fix a constant $m \in \mathbb{N}$, and 
consider the quantum Boolean hypercube relation problem $\mathcal{R}_m$ in \Cref{def:R_m}. The relation problem has the following properties:

\begin{itemize}
\item \textbf{Perfect completeness:}
There exists a family of bounded fan-in circuits $\{C_n\}$ with depth $O(\log m)$ and a constant number of ancillary qubits that solves $\mathcal{R}_m$ with success probability $1$ on every valid input.

\item \textbf{Soundness against shallow circuits:}
For every family of bounded fan-in circuits $\{C_n\}$ of depth $d<\lfloor \log m \rfloor$, regardless of the number of ancillary qubits or access to quantum advice, there exists an input size $n$ such that $C_n$ solves $\mathcal R_m$ with success probability at most $1-2^{-\Omega(m)}$.
\end{itemize}
\end{theorem}

\begin{remark}
Note that the robustness gap of our interactive problem, and of the underlying non-local game, decreases rapidly with the parameter $m$. This is a direct consequence of the multi-controlled phase gates that characterize the unique winning strategy, where $m$ determines the number of control qubits. Indeed, multi-controlled phase gates differ from the identity only on an exponentially small subspace of the Hilbert space. Consequently, if the approximation error is not sufficiently restricted, such gates can be approximated trivially by shallow circuits, including the identity operator itself. Therefore, any meaningful depth lower bound necessarily requires an error threshold that decreases with $m$. This phenomenon is not specific to our protocol: even in a non-delegated setting, certifying the implementation of such gates via process tomography requires distinguishing them from the identity on an exponentially vanishing fraction of the state space, leading to comparable quantitative bounds \cite{mohseni2008quantum,rodionov2015compressed}.
\end{remark}

\paragraph{Completeness.}
Completeness of the relation problem is a direct corollary of the completeness of \Cref{def:QI_LBCS}.
Let $\ket{\psi}_{i,j;\,q_A,q_B}$ be any promised input to $\mathcal{R}_m$ as in \Cref{def:R_m} (and recall that $q_A,q_B$ are computational-basis encodings of the classical questions). We describe a bounded fan-in circuit $C_n$ that outputs an accepting string with probability~$1$.

The circuit acts independently on the two disjoint blocks
\[
\ket{q_A}\otimes \Phi_{A}
\qquad\text{and}\qquad
\ket{q_B}\otimes \Phi_{B},
\]
and leaves the remaining $\ket{0}$-padding registers unchanged. We call the block of qubits $\ket{q_A}\otimes \Phi_{A}$ the $A$-block and the block of qubits $\ket{q_B}\otimes \Phi_{B}$ the $B$-block. 
Below we describe the operations on the $A$-block, and the complete operations on the $B$-block can be constructed similarly.

\medskip
\noindent\emph{Circuit construction.}
Using the quantum circuit in \Cref{fig:top_circuit}, one can load the question encoded in $\ket{q_A}$ into the system and apply a corresponding measurement using the phase estimation circuit. Because this circuit implements a projective measurement of $M_{q_A}$ on $\Phi_{A}$, the resulting answer $a_A$ is an accepting one prescribed by \Cref{lem:Z2_faithful_rep}. A similar construction on the $B$-block produces $a_B$. As the padding registers are untouched and remain $0$, the final output string is exactly of the form
\[
\vecz \;=\; 0^{s_1}\, a_A\, 0^{s_2}\, a_B\, 0^{s_3}
\]
and therefore satisfies $(\ket{\psi}_{i,j;q_A,q_B},\vecz)\in \mathcal{R}_m$. As analyzed in the completeness analysis of \Cref{def:QI_LBCS}, the circuit can be realized in depth $O(\log m)$ with bounded fan-in gates and one ancillary qubit. 

Finally, we note that the circuit family $\{C_n\}_n$ described above is, as stated, non-uniform in $n$. This can be remedied with a simple modification. In the underlying interactive protocol of \Cref{def:QI_LBCS}, the provers receive explicit classical descriptions of their questions, $q_A$ and $q_B$, and can therefore implement the corresponding operations directly, without needing to identify where in the system the questions apply. Their behavior is thus naturally question-controlled rather than position-dependent.

We exploit this in the circuit setting by broadcasting the classical questions to all registers and implementing, in parallel, the corresponding local operations across all candidate subsets. Since no routing or index computation is required, this yields a uniform $\QNC^0$ circuit family. This establishes perfect completeness.

\paragraph{Soundness.} 
Consider a constant-depth quantum circuit with all-to-all connectivity with a bounded fan-in of $K=2$. For this architecture, let $i$ and $j$ denote randomly selected blocks of qubits carrying the question labels $q_A$ and $q_B$ in the input state $\ket{\psi}_{i,j;q_A,q_B}$. We show that there is a non-negligible probability that the output $a_A$ is independent of the input question label $q_B$, while simultaneously the output $a_B$ is independent of the input question label $q_A$.

Specifically, in a $\QNC^0$ circuit, due to the bounded depth and bounded fan-in restrictions, the backward light cone of any output register covers at most $2^d$ input qubits: $|L^\leftarrow(i)|\leq 2^d$, where $i$ denotes the index of the output register, $L^\leftarrow(\cdot)$ represents the backward light cone, and $|\cdot|$ denotes the size, i.e., the number of qubits. It can be easily seen that for a given set of output registers, $O$, for a randomly sampled set of input registers, $I$, we have
\[
\Pr_I[O\cap L_C^{\rightarrow}(I)\neq\emptyset]\leq\sum_{P\subseteq O,P\neq\emptyset}\Pr_{I}[I\cap L_C^{\leftarrow}(P)\neq\emptyset]\leq 2^{|O|}|O|2^d \Pr_I[v\in I],
\]
where $L_C^{\rightarrow}(\cdot)$ represents the forward light cone, and $v\in I$ is any fixed input qubit index. With this transformation between the forward and backward light cones, we can show that the qubits outputting the answers lie outside the forward light cones of the input questions with high probability, effectively isolating the local computations. This can be described precisely by the following lemma:

\begin{lemma}[{\cite[Supplementary Information Lemma 4]{Zhang2024}}]\label{lemma:lightcone} 
Consider a $\QNC_0$ circuit $C$ with depth $D$ composed of gates with bounded fan-in. Define event $E_C\subset S$ as the set of input states satisfying
\begin{equation}
    \mathsf{supp}(a_A)\cap L^{\rightarrow}(\mathsf{supp}(q_B))=\emptyset \quad \text{and} \quad  \mathsf{supp}(a_B)\cap L^{\rightarrow}(\mathsf{supp}(q_A))=\emptyset, 
\end{equation}
where $S$ represents the set of all possible inputs, and we denote the qubits carrying the answer $a_A$ as $\mathsf{supp}(a_A)$, and similarly for $\mathsf{supp}(a_B)$.

Then, under a uniform distribution of inputs $S$, the event $E_C$ occurs with probability $1-O(1/n)$.
\end{lemma}

This lemma follows from a standard light cone argument, including a transformation between output-backward and input-forward light cones.

Applying the aforementioned lemma, we observe that for sufficiently large $n$, the random choice of inputs in the promise problem $\mathcal{R}_m$ ensures that, with constant probability, the response to question label $q_A$ (respectively $q_B$) is independent of the input associated with $q_B$ (respectively $q_A$); see \Cref{def:QS_Z2OCS}. In this event, the protocol decomposes into two space-like separated $\mathsf{QNC}^0$ circuits that respond independently to $q_A$ and $q_B$.

This is precisely the setting of \Cref{lemma:cliff_verif}, which shows that any prover strategy implemented by circuits of depth $d < \log m$ succeeds with probability at most $1-2^{-O(m)}$. It follows that, with constant probability, any such shallow-depth strategy fails to produce a valid solution to $\mathcal{R}_m$. Consequently, no family of $\mathsf{QNC}^0$ circuits of depth $d < \log m$ can solve $\mathcal{R}_m$ with success probability exceeding $1-2^{-\Omega(m)}$.

Finally, we rule out the use of quantum advice states. Any winning strategy requires access to near-perfect EPR pairs across two randomly selected registers. By the monogamy of entanglement for EPR pairs, together with the large number of possible register choices, no fixed polynomial-size advice state can simultaneously provide the required correlations for more than a negligible fraction of instances arising from $\mathcal{R}_m$. Consequently, quantum advice does not substantially enlarge the set of instances that can be solved correctly, and therefore does not asymptotically increase the success probability of insufficient-depth $\QNC^0$ circuits.
\hfill\qedsymbol

This theorem yields a natural quantum depth hierarchy: harder instances of the quantum-input relation problem require increasingly deep quantum circuits to achieve near-perfect success.

\paragraph{Comparison with Semi-Quantum Games.} The setting in which quantum inputs are provided and classical outputs are required is naturally related to the framework of semi-quantum games introduced by Buscemi \cite{buscemi2012all}. In these games, non-communicating provers receive randomly sampled quantum states as questions and must return classical answers, which the verifier checks against a relation defined by the labels of the quantum inputs. Semi-quantum games play an important role in quantum information theory because they generalize Bell nonlocality: for every entangled state, including states that do not violate any standard Bell inequality, there exists a semi-quantum game in which that state achieves a higher success probability than any classical strategy.

The interactive protocol of \Cref{subsec:Int_Clifford} naturally fits within this framework, while the present section extends such games into computational problems. However, our use of the model differs from its conventional role in quantum information theory. Rather than using semi-quantum games to certify entanglement or nonlocality, we employ a related game-based construction to distinguish between quantum provers with different computational capabilities. In particular, the resulting relation problem characterizes the quantum circuit depth required to implement specific observables and correlations.

\paragraph{Comparison with state synthesis problems.} One might argue that quantum depth hierarchies can already be obtained from quantum state-synthesis problems, with multipartite states such as GHZ states serving as natural examples. However, our construction differs from such approaches in several important respects.

First, state-synthesis problems typically involve classical inputs and quantum outputs, whereas our setting reverses this structure by considering quantum inputs and classical outputs. More importantly, state-synthesis approaches require certifying the preparation of the target quantum state, which generally necessitates either trusted quantum measurements together with tomography-like procedures or more elaborate verification methods such as swap tests. In contrast, our construction only assumes that the verifier can realize constant-depth Clifford circuits.

Another distinction is that GHZ state-synthesis tasks typically certify the ability to prepare entanglement across a fixed collection of qubit registers determined by the circuit architecture. In a uniform circuit family, the locations participating in the GHZ preparation are therefore essentially hardwired into the circuit. By contrast, in our construction, the relevant subsystems are selected dynamically through the input and may appear at many different locations within the circuit. The task, therefore, certifies the ability to maintain and manipulate coherent correlations between multiple possible regions of the device, rather than only along a predetermined entangling structure.

\section{Classical Interactive Protocol}\label{sec:Classical_int}

In this section, our main objective is to replace the Clifford-capable verifier in \Cref{def:QI_LBCS} of \Cref{subsec:Int_Clifford} with a fully classical verifier. To achieve this, we construct a classical verification procedure that preserves the two essential functionalities on which our earlier construction relied:

\begin{itemize}
\item the preparation of specific Clifford-rotated EPR states, and
\item the measurement of Pauli observables.
\end{itemize}

Since a classical verifier cannot directly perform quantum state preparation or Pauli measurements, both tasks must be delegated to quantum provers and enforced through interaction. We first show how a two-round interactive protocol with three provers enables verifiable delegated state preparation of the required states over two of these provers. Moreover, the same mechanism that certifies the state also certifies the Pauli observables needed for our construction. Finally, we show how this setting can be naturally adapted to the shallow-depth regime, where it behaves similarly to a state-commitment protocol for space-like separated regions of a quantum circuit.

Our second step is to embed the questions corresponding to the constraint system introduced in \Cref{def:OVCS_Z2} into this state-commitment framework. This embedding leverages the certification achieved in the first step to enforce an equivalent uniqueness property for the operator-valued solution as established in \Cref{lem:Z2_faithful_rep}. The construction proceeds through a careful composition of games, mirroring at the algebraic level the amalgamated-product structure introduced in \Cref{subsec:irrep-to-faithful}.

In the end, this yields a family of fully classical-verifier interactive tasks that separate shallow classical and quantum computation. Every classical $\textsf{NC}^0$ circuit has success probability bounded away from 1, whereas each task admits a perfect strategy implemented by a shallow quantum circuit of sufficient depth. Moreover, quantum circuits below the required depth threshold cannot achieve near-perfect success. Thus, the interactive formulation preserves the robust quantum circuit depth hierarchy established in the previous section.

\subsection{Delegated State Preparation with Three Provers}\label{subsec:delegated}

Delegated quantum state preparation, and more broadly delegated quantum computation, have been extensively studied in the literature \cite{Reichardt2013,Coladangelo2024}. Our setting, however, requires one of the strongest forms of certification: the verifier chooses a local Clifford frame for a Bell state, and the protocol must enforce that this specific Clifford operation is physically applied while remaining hidden from the parties performing the later measurements. This setting has received comparatively less attention because Clifford operations are often regarded as computationally simple \cite{Clifford_sim} and can frequently be absorbed into classical post-processing.

In particular, existing rigidity and delegation protocols certify EPR pairs and measurements only up to local isometries. In particular, in two-prover settings, local Clifford transformations can be absorbed into a change of measurement basis or classical post-processing without affecting the observed correlations. Consequently, prior techniques do not suffice to certify that a verifier-chosen Clifford operation has actually been implemented on remote subsystems, or otherwise require additional assumptions or resources that are undesirable in our setting. For instance, delegated blind quantum computation based on measurement-based quantum computation \cite{Broadbent09} could in principle achieve this functionality, but such protocols typically require either initial quantum communication between the verifier and prover or, to remove it, computational assumptions \cite{mahadev2018classical}.

To address this issue, we introduce a two-round interactive protocol involving one classical verifier and three quantum provers; see \Cref{fig:delegated_state_preparation}. The third prover, Charlie, is introduced specifically to implement verifier-chosen Clifford operations, which are applied to the systems of Alice and Bob via gate teleportation. The key idea is to decouple the application of the Clifford transformation from the subsystems on which the final measurements are performed. Charlie alone receives the verifier’s instruction specifying the Clifford operation, while Alice and Bob remain unaware of this choice. This asymmetry is essential: by withholding the Clifford description from the measuring provers, we prevent them from absorbing its effect into local basis changes or classical post-processing. Since Alice and Bob do not know which Clifford operation was applied, they cannot adapt their measurement strategies to compensate for deviations later in the protocol. Conversely, any failure by Charlie to implement the prescribed operation produces inconsistencies in the correlations observed between Alice and Bob, which the verifier detects using the teleportation transcript together with the subsequent measurement outcomes.

\begin{figure}[htb]
\centering
\begin{tikzpicture}[
    x=1cm,y=1cm,
    lab/.style={font=\small},
    party/.style={font=\small},
    cmsg/.style={->,thick,dashed},
    roundsep/.style={densely dashed,gray}
]

\def\xRes{0.2}
\def\xStart{2}
\def\xRoneL{4.2}
\def\xRoneR{6.4}
\def\xSep{6.2}
\def\xRtwoL{7.0}
\def\xRtwoR{10.8}
\def\xEnd{10.8}

\def\yA{2.2}
\def\yC{1.0}
\def\yB{-0.2}
\def\yV{-1.6}

\node[lab] at (\xRes-1,\yA-0.4) {$\ket{\Phi^+}^{\otimes m}$};
\node[lab] at (\xRes-1,\yB+0.5) {$\ket{\Phi^+}^{\otimes m}$};

\draw[thick] (\xStart-1.95,1.55) -- (\xStart-1.95,1.95);   
\draw[thick] (\xStart-1.95,1.95) -- (\xStart-1.45,\yA+0.05); 
\draw[thick] (\xStart-1.95,1.55) -- (\xStart-1.45,\yC+0.05); 

\draw[thick] (\xStart-1.95,0.10) -- (\xStart-1.95,0.50);   
\draw[thick] (\xStart-1.95,0.5) -- (\xStart-1.45,\yC-0.05); 
\draw[thick] (\xStart-1.95,0.10) -- (\xStart-1.45,\yB+0.00); 

\node[party] at (\xStart-0.8,\yA+0) {Alice};
\node[party] at (\xStart-0.8,\yC+0) {Charlie};
\node[party] at (\xStart-0.8,\yB+0) {Bob};
\node[party] at (\xStart-0.8,\yV+0) {\textbf{Verifier}};

\draw (\xStart,\yA) -- (\xEnd,\yA);
\draw (\xStart,\yC) -- (\xEnd,\yC);
\draw (\xStart,\yB) -- (\xEnd,\yB);
\draw (\xStart,\yV) -- (\xEnd,\yV);

\draw[roundsep] (\xSep,\yV-0.7) -- (\xSep,\yA+0.7);

\node[lab] at (\xSep-2,3.2) {\textbf{Round 1}};
\node[lab] at (\xSep+2,3.2) {\textbf{Round 2}};

\draw[cmsg] (\xSep-3.6,\yV) -- (\xSep-2.6,\yC);
\node[lab, above] at (\xSep-3.55,-0.8) {$c$};

\draw[cmsg] (\xSep-2,\yC) -- (\xSep-1,\yV);
\node[lab, right] at (\xSep-1.,-0.8) {$y$};

\draw[cmsg] (\xRtwoL,\yV) -- (\xSep+2.2,\yA);
\draw[cmsg] (\xRtwoL,\yV) -- (\xSep+2.2,\yB);

\node[lab, above] at (8.25,0.15) {$q_A(y)$};
\node[lab, below] at (8.00,-0.9) {$q_B(y)$};

\draw[cmsg] (9,\yA) -- (10.2,\yV);
\draw[cmsg] (9,\yB) -- (10.2,\yV);

\node[lab, above] at (9.9,0.2) {$a_A$};
\node[lab, below] at (9.45,-0.9) {$a_B$};

\end{tikzpicture}
\caption{Representation of the Clifford-basis test protocol. \textbf{Round 1.} The verifier selects a random Clifford operation $\mathcal{C}_c \in \{\mathsf{CNOT}, \mathsf{Swap}, \mathsf{X}, I\}$ together with the qubit registers on which it is to be applied. The verifier then instructs Charlie, via a classical label $c$, to perform the corresponding gate-teleportation procedure and return the measurement outcomes $y$. \textbf{Round 2.} The verifier computes the Pauli correction $r_{c,y}$ determined by Charlie’s measurement outcomes and sends questions corresponding to the Clifford-rotated parallel Mermin--Peres relations described in \Cref{Fig:CliffordtoMermin}, with the questions updated according to the correction $r_{c,y}$.}
\label{fig:delegated_state_preparation}
\end{figure}
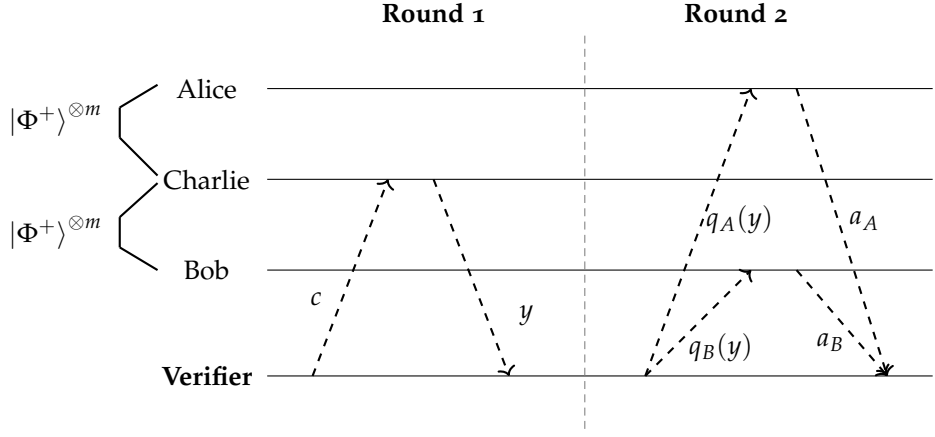

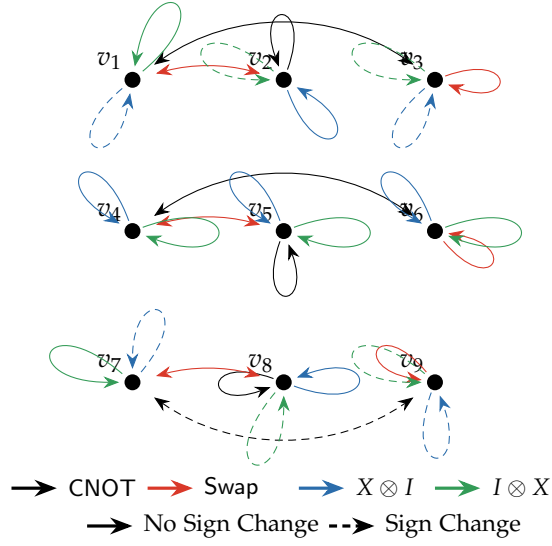
\begin{figure}[htb]
\centering
\begin{tikzpicture}[>=Stealth, scale=2]
\foreach \i in {1,2,3} {
    \foreach \j in {1,2,3} {
        \pgfmathsetmacro{\xcoord}{\j}
        \pgfmathsetmacro{\ycoord}{4-\i}
        \pgfmathsetmacro{\index}{int((\i-1)*3 + \j)}
        \fill ( \xcoord, \ycoord ) circle (1.5pt) node (v\index) {};
        \node[anchor=south east, font=\small] at (v\index) {$v_{\index}$};
    }
}

\tikzset{
    ->, >=Stealth, very thick,
    no_sign/.style={solid},
    sign_change/.style={dash pattern=on 3pt off 2pt},
    cnot/.style={black},
    swap/.style={cmykRed},
    xi/.style={cmykBlue},
    ix/.style={cmykGreen},
    bi/.style={<->, shorten >=4pt, shorten <=4pt}
}

\draw [bi, cnot, no_sign, bend left=35] (v1) to (v3);
\draw [cnot, no_sign, loop, out=70, in=110, looseness=8, distance=0.5cm] (v2) to (v2);
\draw [bi, cnot, no_sign, bend left=35] (v4) to (v6);
\draw [cnot, no_sign, loop, out=250, in=290, looseness=8, distance=0.5cm] (v5) to (v5);
\draw [bi, cnot, sign_change, bend right=35] (v7) to (v9);
\draw [cnot, no_sign, loop, out=160, in=200, looseness=8, distance=0.5cm] (v8) to (v8);

\draw [bi, swap, no_sign, bend left=15] (v1) to (v2);
\draw [bi, swap, no_sign, bend left=15] (v4) to (v5);
\draw [bi, swap, no_sign, bend left=15] (v7) to (v8);
\draw [swap, no_sign, loop, out=20, in=340, looseness=8, distance=0.5cm] (v3) to (v3);
\draw [swap, no_sign, loop, out=310, in=350, looseness=8, distance=0.5cm] (v6) to (v6);
\draw [swap, no_sign, loop, out=130, in=170, looseness=8, distance=0.5cm] (v9) to (v9);

\draw [xi, sign_change, loop, out=220, in=260, looseness=10, distance=0.6cm] (v1) to (v1);
\draw [xi, no_sign, loop, out=290, in=330, looseness=10, distance=0.6cm] (v2) to (v2);
\draw [xi, sign_change, loop, out=220, in=260, looseness=10, distance=0.6cm] (v3) to (v3);
\draw [xi, no_sign, loop, out=110, in=150, looseness=10, distance=0.6cm] (v4) to (v4);
\draw [xi, no_sign, loop, out=110, in=150, looseness=10, distance=0.6cm] (v5) to (v5);
\draw [xi, no_sign, loop, out=110, in=150, looseness=10, distance=0.6cm] (v6) to (v6);
\draw [xi, sign_change, loop, out=50, in=90, looseness=10, distance=0.6cm] (v7) to (v7);
\draw [xi, no_sign, loop, out=340, in=20, looseness=10, distance=0.6cm] (v8) to (v8);
\draw [xi, sign_change, loop, out=250, in=290, looseness=10, distance=0.6cm] (v9) to (v9);

\draw [ix, no_sign, loop, out=40, in=80, looseness=12, distance=0.7cm] (v1) to (v1);
\draw [ix, sign_change, loop, out=140, in=180, looseness=12, distance=0.7cm] (v2) to (v2);
\draw [ix, sign_change, loop, out=140, in=180, looseness=12, distance=0.7cm] (v3) to (v3);
\draw [ix, no_sign, loop, out=20, in=340, looseness=12, distance=0.7cm] (v4) to (v4);
\draw [ix, no_sign, loop, out=20, in=340, looseness=12, distance=0.7cm] (v5) to (v5);
\draw [ix, no_sign, loop, out=20, in=340, looseness=12, distance=0.7cm] (v6) to (v6);
\draw [ix, no_sign, loop, out=140, in=180, looseness=12, distance=0.7cm] (v7) to (v7);
\draw [ix, sign_change, loop, out=230, in=270, looseness=12, distance=0.7cm] (v8) to (v8);
\draw [ix, sign_change, loop, out=140, in=180, looseness=12, distance=0.7cm] (v9) to (v9);

\begin{scope}[shift={(2, 0.3)}, font=\footnotesize]
    \begin{scope}[xshift=-1.8cm]
        \draw[black, very thick] (0,0) -- (0.3,0) node[right, black] {$\mathsf{CNOT}$};
        \draw[cmykRed, very thick] (0.9,0) -- (1.2,0) node[right, black] {$\mathsf{Swap}$};
        \draw[cmykBlue, very thick] (1.9,0) -- (2.2,0) node[right, black] {$X\otimes I$};
        \draw[cmykGreen, very thick] (2.8,0) -- (3.1,0) node[right, black] {$I\otimes X$};
    \end{scope}
    
    \begin{scope}[xshift=-1.3cm, yshift=-0.25cm]
        \draw[black, very thick] (0,0) -- (0.3,0) node[right, black] {No Sign Change};
        \draw[black, very thick, dash pattern=on 3pt off 2pt] (1.6,0) -- (1.9,0) node[right, black] {Sign Change};
    \end{scope}
\end{scope}

\end{tikzpicture}
\caption{Action of Clifford operations on the variables of the Mermin--Peres constraint system. Each vertex represents a variable from \Cref{eq:MerminPeres}, and each directed edge represents the action of a Clifford operation mapping one variable to another under conjugation. When extended to the parallel Mermin--Peres constraint system (\Cref{def:MP_parallel}), these transformations encode the relators implementing the action $\alpha$ in the semidirect product structure $\mathcal{P}_m \rtimes_{\alpha} \mathrm{AGL}(m,2)$, defining the underlying operator-valued constraint system of the Clifford-basis test protocol.}
\label{Fig:CliffordtoMermin}
\end{figure}

We now state the self-testing lemma for the two-round three-prover protocol, certifying the preparation of EPR pairs in verifier-chosen local Clifford bases together with the corresponding Pauli observables.

\begin{lemma}[Robustness of the Clifford-basis test]\label{lem:clifford_basis_rigidity}
Consider the two-round three-prover protocol of \Cref{fig:delegated_state_preparation}, and suppose that Alice, Bob, and Charlie follow a strategy that succeeds with probability at least $1-\varepsilon$.  Then there exists a local isometry $V = V_A\otimes V_B$ acting on the joint Hilbert space of Alice and Bob such that, for every choice of Clifford label $c$ and corresponding Pauli strings $\vecs,\vect,\vecu,\vecv\in\{0,1\}^{m}$, there exist observables $A_{\vecs,\vect}$ and $B_{\vecu,\vecv}$ on Alice's and Bob's local Hilbert spaces satisfying
\begin{equation}\label{eq:clifford_rigidity_exp}
\Bigl|\langle \phi_c|X(\vecs)Z(\vect)\otimes X(\vecu)Z(\vecv)|\phi_c\rangle-\langle \psi_c| A_{\vecs,\vect}\otimes B_{\vecu,\vecv}|\psi_c\rangle\Bigr|\leq O(m^2\sqrt{\varepsilon})
\end{equation}
where $|\phi_c\rangle := V(|\psi_c\rangle)$.

Moreover, 
\begin{equation}
\bra{\phi_c} (I^{k} \otimes \mathcal{C}_c \otimes I^{m'})(\ket{\Phi}  \bra{\Phi}_{AB}\otimes \rho_{junk})(I^{k}\otimes \mathcal{C}_c \otimes I^{m'})^{\dagger} \ket{\phi_c} \geq 1-O(m^2\sqrt{\varepsilon}),
\end{equation}
\noindent where $\ket{\Phi}_{AB}$ is the $2m$-fold tensor product of Bell pairs $\bigotimes_{i=1}^{2m}\ket{\Phi^+}_{A_iB_i}$, $k\in[2m]$ and $m'=2m-k-|c|$.
\end{lemma}

The proof follows in \Cref{subsec:robust5}.

\paragraph{Unconditional state-commitment.} Subsequently, we translate the above self-testing statement, originally formulated for space-like separated provers sharing an entangled state and performing local measurements, into an equivalent setting for circuits in $\QNC^0$.

In this formulation, the collection of provers from \Cref{fig:delegated_state_preparation} are replaced by a single prover implemented as a shallow quantum circuit, where abstractly each prover corresponds to a designated subset of qubit registers. The protocol remains two-round with a classical verifier, and the first round is interpreted as a state-commitment phase: the circuit prepares the required entangled state with a Pauli frame fixed up to measurement outcomes, while the verifier obtains a classical transcript determining this frame and uses it to update the second-round questions. Since these questions are chosen independently of the interaction and depend only on the Pauli frame for consistency, which the circuit cannot determine due to the computational limitations of $\QNC^0$, the circuit cannot modify the committed state (up to local isometries) and must respond consistently with it, as do the space-like separated provers in the Clifford basis test.

\begin{algorithm}[htb]
\floatname{algorithm}{Protocol}
\caption{Clifford-Rotated EPR State-Commitment Protocol}\label{def:Cliff_EPR_self}
\begin{algorithmic}[1]
\Require Parameters $n,m,k\in \mathbb{N}$. 

\State Verifier randomly samples $k$ question tuples $\{q^i=(q_A^i,q_B^i)\}_{i=1}^{k}$ from the $m$-fold parallel Mermin--Peres game. Let $t\gets 0$.  

\For{$i=1:k$}

\vspace{0.2mm}
\Statex {\bfseries Round 1: State commitment.}
\vspace{0.2mm}

\State Verifier randomly selects labels $l<h<j\in [n]$, $\mathcal{C}_c \in \{\mathsf{CNOT},\mathsf{Swap},X,I\}$ and sends 
\Statex \hspace{\algorithmicindent}$s_c=0^{l-1} \times \underbrace{1}_{l} \times 0^{h-l-1} \times \underbrace{c}_h \times 0^{j-h-1}\times \underbrace{1}_{j}\times 0^{n-j}.$
\State Prover sends back a classical string $y\in \{0,1\}^{(n-2m)}$.

\vspace{0.2mm}
\Statex {\bfseries Round 2: Clifford-basis test.}
\vspace{0.2mm}

\State Verifier updates the question pair $(q_A^{i},q_B^i)$, based on $\mathcal{C}_c$ and $y$, and sends to prover $s_{(q_A,q_B)}=0^{l-1} \times q_A^{i}(y) \times 0^{j-l-1}  \times q_B^i(y) \times 0^{n-j}$. 

\State Prover sends back its answer $s_{(a_A,a_B)}=\ \{0,1\}^{l-1} \times a_A^i \times \{0,1\}^{j-l-1}\times a_B^i \times \{0,1\}^{n-j}$ .

\State Verifier checks correctness: $t\gets t+1$ if $s_{(a_A,a_B)}$ is a valid answer. 

\EndFor
\State \Return the fraction of accepted rounds $t/k$.
\end{algorithmic}
\end{algorithm}

\begin{lemma}[Robustness of the Clifford-basis test for $\QNC^0$ circuits]
\label{lemma:cliff_test_qnc}
For any $m \in \mathbb{N}$, set $k = m$. For all sufficiently large $n$, in the interactive protocol of \Cref{def:Cliff_EPR_self}, the following holds:
\begin{itemize}
\item \textbf{Perfect completeness.}
There exists a $\mathsf{QNC}^0$ circuit of depth $O(1)$ that succeeds in the protocol with probability $1$.

\item \textbf{Rigidity for shallow quantum circuits.}
Let $C$ be any $\mathsf{QNC}^0$ circuit that succeeds with probability at least $1-\varepsilon$. Then there exist two disjoint sets of output registers, denoted $A$ and $B$, such that the induced strategy on these registers satisfies the conclusions of \Cref{lem:clifford_basis_rigidity}.
\end{itemize}
Moreover, the verifier transcript can be computed in $\NC^0[\oplus]$.
\end{lemma}

\paragraph{Proof of completeness.} For completeness, we show that an honest prover restricted to $\QNC^0$ computational power can succeed in the Clifford-rotated EPR state-commitment protocol with probability $1$. 

It suffices to show that, in the first round, the circuit can prepare, up to Pauli corrections determined by the teleportation outcomes, the state
\begin{equation}
\left(I^{\otimes l} \otimes \mathcal{C}_c^{A,B} \otimes I^{\otimes n-l-1}\right)
 \ket{\Phi} \bra{\Phi}_{AB} \left(I^{\otimes l} \otimes \mathcal{C}_c^{A,B} \otimes I^{\otimes n-l-1}\right )^\dagger
\end{equation}
across two output registers indexed by $j$ and $l$. Here, the $A$-registers are located at position $j$, the $B$-registers at position $l$, and the Clifford operator $\mathcal{C}_c^{A,B}$ is specified by the label $c$, while the qubit register index $h$ defines where the previous label is provided in the string. The resulting state is then used to generate the correlations required for the nonlocal game questions $q_A$ and $q_B$.

The main difficulty is that the indices $l,h,j$, together with the location of the Clifford operation, depend on the classical input label, while the circuit family $\{C_n\}_n$ must remain fixed for each input size $n$ (see the label sent in Step 3 of \Cref{def:Cliff_EPR_self}). In particular, a $\QNC^0$ circuit cannot compute routing predicates from binary-encoded indices.

\medskip
\noindent\textit{Input-dependent routing in constant-depth circuits.} To address this, we define a general $\QNC^0$ circuit architecture consisting of two components: a local computational region, whose operations are restricted to Clifford gates, and a fixed one-dimensional nearest-neighbour teleportation chain connecting all qubits. The teleportation chain serves as a transport layer. States located at arbitrary indices can be transferred to auxiliary qubits in the computational region, processed locally, and then returned to designated indices by activating contiguous segments of the chain. The classical input specifies a unary activation pattern that determines which segment of the chain is used. Performing parallel Bell measurements along the activated segment implements teleportation across that segment. Because this activation pattern is supplied explicitly as part of the input, the circuit performs no internal search or index computation.

\medskip
Subsequently, we describe how the first round of the interactive task can be implemented by a circuit composed of the following three stages.

\medskip
\noindent\textbf{(1) Bell pair preparation.}
The first layer prepares a collection of nearest-neighbour EPR pairs across the circuit. Concretely, for each layer $r \in [m]$ and each position $t \in [n]$, the circuit prepares a Bell pair  $\ket{\Phi^+}$ between qubits at positions $(2t-1,r)$ and $(2t,r)$. This results in an $m$-fold tensor product of nearest-neighbour EPR pairs arranged along the one-dimensional layout of $n$ elements. In particular, this includes the qubits located at Charlie’s register $h$, which will later be used as the source system for teleportation. This preparation can be implemented using constant-depth Clifford circuits (of depth $2$).

\medskip
\noindent\textbf{(2) Clifford application at Charlie.}
In the second layer, the circuit applies the Clifford operation $\mathcal{C}_c$ to the appropriate subset of qubits within Charlie’s registers, indexed by $h$. To obtain a uniform circuit description, we employ the input-dependent routing architecture described earlier to implement these Clifford operations as constant-depth circuits.

Specifically, the relevant qubits are routed from Charlie’s registers to a dedicated computational region using entanglement-swapping layers and auxiliary qubits, and are subsequently returned to their original locations. This routing requires one layer of controlled swaps to select the participating data qubits, followed by three layers for entanglement swapping.

The input label $c$ determines the choice of Clifford operation. Since $\mathcal{C}_c$ belongs to a fixed finite set (e.g., $\{\mathsf{CNOT}, \mathsf{Swap}, X, I\}$), it can be implemented using at most two layers of classically controlled Clifford gates. Finally, an additional 3+1 layers are used to route the qubits back to their original positions. Overall, the number of ancilla qubits required is linear in the system size, i.e., $O(n)$.

\medskip
\noindent\textbf{(3) Entanglement swapping (Bell measurements).}
In the third layer, the circuit performs Bell-basis measurements to teleport Charlie’s subsystems to the output registers indexed by $j$ and $l$. 
Specifically, for each relevant segment determined by the input string $s_c$, Bell measurements are performed in parallel between neighbouring qubits, effectively implementing entanglement swapping along the path from $h$ to $j$ and from $h$ to $l$. 
Since the input string contains exactly two active positions, this procedure selects the appropriate intervals and establishes an $m$-fold tensor product of EPR pairs between the endpoints $j$ and $l$. 
This is a standard entanglement swapping procedure and can be realized with three layers.

\medskip
The entire circuit can be composed to implement the prover’s operations in the first round. In particular, using the previous circuit architecture, the circuit depth can be bounded by 12. The outcomes of the Bell measurements define the classical string $y$ returned by the prover in this round, and determine the Pauli corrections induced by the teleportation.

In the second round, the verifier uses $y$ to compute the corresponding Pauli corrections and combines them with the Clifford label $c$ to determine the measurement bases for the parallel Mermin--Peres instances. Both the Pauli frame updates and the Clifford conjugations reduce to local relabelings on $m$-bit strings together with $n$-bit parity computations, and hence can be implemented in $\NC^0[\oplus]$.

The prover’s strategy in this round is implemented by a constant-depth $\QNC^0$ circuit consisting only of single- and two-qubit gates that perform the required Pauli measurements on the registers associated with the corresponding observable labels. 

Overall, the three layers of the first round consist solely of constant-depth Clifford operations and measurements and can therefore be executed in parallel, while the second round involves only classical queries and local measurements on the designated registers. Consequently, the honest strategy admits a classically controlled $\QNC^0$ implementation that succeeds with probability~$1$.\hfill\qedsymbol

\paragraph{Proof of soundness.} For soundness, we reduce the analysis to the three-prover setting of \Cref{lem:clifford_basis_rigidity} by showing that the relevant registers of the $\QNC^0$ circuit behave as non-communicating parties with high probability.

The verifier selects three registers indexed by $l$, $h$, and $j$ (corresponding to Alice, Bob, and Charlie) uniformly at random.  The backward light cone of any output register in a depth-$d$ circuit has size at most $O(2^d)$. Since the circuit has constant depth, this implies that the probability that any two of the three selected registers have intersecting backward light cones is negligible. Applying a union bound over the three pairs and applying a standard light cone argument as in \Cref{lemma:lightcone}, we conclude that, with high probability, all three registers have disjoint backward light cones. Conditioned on this event, the corresponding subsystems behave as non-communicating provers in both rounds of the protocol.

Under this condition, the induced strategy of the circuit is operationally equivalent to a three-prover strategy in the sense of \Cref{fig:delegated_state_preparation}. The only difference is that, in the circuit setting, the entanglement shared between the registers is generated internally by the circuit rather than provided \emph{a priori}. However, this distinction is immaterial: any successful strategy for the protocol must still produce correlations consistent with those of the parallel repeated Mermin--Peres game in the second round.

Crucially, the verifier’s modifications to the second-round questions, accounting for the Clifford label $c$, and the teleportation outcomes $y$, only relabel the Pauli observables appearing in the test and are independent of the choice of the questions themselves. Thus, they do not alter the underlying constraint system of the parallel Mermin--Peres game. Therefore, the rigidity guarantees of \Cref{lem:clifford_basis_rigidity} continue to apply.

It follows that any $\QNC^0$ circuit that succeeds in the protocol with probability at least $1-\varepsilon$ must induce, on the selected registers, a strategy satisfying the same conclusions as in \Cref{lem:clifford_basis_rigidity}. In particular, up to local isometries, the circuit prepares the appropriate Clifford-rotated EPR state and implements the corresponding Pauli observables on the extracted registers.\hfill\qedsymbol

\subsection{Composing Constraint Systems}\label{subsec:composing}

Here, we describe the composition of the game and constraint systems underlying \Cref{def:QI_LBCS} and \Cref{def:Cliff_EPR_self}. In this construction, the Clifford-Rotated EPR State-Commitment Protocol certifies the states and observables required by \Cref{def:QI_LBCS}, replacing the functionalities previously carried out by the Clifford-capable verifier. This composition is therefore the key step in replacing the quantum verifier in \Cref{def:QI_LBCS} with a fully classical verifier while still enforcing the optimal strategies for the faithful $\mathbb Z_2^m$ operator-valued constraint system introduced in \Cref{sec:constraint_system}. The resulting protocol preserves the essential properties of both components and reformulates all observables as abstract variables within a single composed constraint system.

Overall, we consider two sub-tests: (1) the nonlocal game described in \Cref{def:QI_LBCS}, derived from the faithful $\mathbb Z_2^m$ operator-valued constraint system in \Cref{def:OVCS_Z2}, in which the verifier alternates between consistency and constraint-satisfaction checks; (2) the Clifford-basis test described in \Cref{fig:delegated_state_preparation}, defined from the parallel Mermin--Peres constraint system in \Cref{def:MP_parallel} subjected to a set of classical commands of local unitary transformations. Note that the parallel Mermin--Peres game itself is already obtained as a parallel repetition of the single-instance Mermin--Peres game. Moreover, both sub-tests admit rigid optimal strategies for self-testing the shared state $\bigotimes_{i=1}^{m}\ket{\Phi^+}$.

The essence of the game composition is to notice the overlapping variables that show up in different constraint systems, which are performed over the same underlying system. Specifically,

\begin{itemize}
    \item \textbf{Pauli observables:} The explicit Pauli observables $Z(\mathbf{s})$ appearing in (1), originating from \Cref{def:OVCS_Z2}, are identified with the subset of abstract variables $a_{j,\veci[j]}$ satisfying $\veci[j]=0$ in (2).

    \item \textbf{Clifford unitary transformations:} The Clifford operations $\mathsf{CNOT}$, $X$, and $\mathsf{Swap}$ are explicitly specified in (1), where they act by mapping variables $z_{\veci}$ to variables $z_{\veci'}$. In the Clifford-basis test of \Cref{fig:delegated_state_preparation}, these same transformations are enforced abstractly through their action on the variables $v_{j,k}$, with $k\in\{1,\dots,9\}$, inducing relabellings between variables corresponding to $a_{j,\veci[j]}$ and $a_{j,\veci'[j]}$ for $\veci[j],\veci'[j]\in\{0,\dots,5\}$ (see \Cref{Fig:CliffordtoMermin}).
\end{itemize}

To compose the constraint systems, we present all variables abstractly. In addition to the generators $z_{\veci}$ of $\mathbb{Z}_2^m$, the generators $a_{j,\mathbf{i}[j]}$ of the parallel Mermin--Peres constraint system, and its composing variables $v_{j,k}$ in single-instance Mermin--Peres constraint systems for some values of $k$ according to \Cref{eq:MerminPeres}, we introduce the following abstract generators:
\begin{itemize}
    \item $s_{a,b}, c_{a,b}, x_a$, corresponding respectively to $\mathsf{Swap}_{a,b}$, $\mathsf{CNOT}_{a,b}$, and $X_a$.
\end{itemize}
We denote the set of these generators as $\mathrm{CL}$ (referring to Clifford operations). 

Based on these generators, we define the set of merge relators $R_{\mathrm{merge}}$ by
\begin{equation}
\begin{aligned}
& \prod_{\veci: \veci\cdot \mathsf{Bin}(\mathbf{k}[1],...,\mathbf{k}[m/2]) = 1} z_{\veci} = v_{(1,\mathbf{k}[1])}
...v_{(m/2,\mathbf{k}[m/2])}, && \forall \mathbf{k} \in \{1,2,3\}^{m/2},\\
& s_{a,b}\, z_{\veci}\, s_{a,b} 
= z_{\sigma_{a,b}(\veci)}, 
&& \forall \veci\in\{0,1\}^m,\; a,b\in[m],\; a\leq b,\\
& x_a\, z_{\veci}\, x_a = z_{\veci\oplus\veca}, 
&& \forall \veci\in\{0,1\}^m,\; a\in[m],\\
& c_{a,b}\, z_{\veci}\, c_{a,b} 
= z_{\veci\oplus(\veci[a]\mathbf{e}_b)}, 
&& \forall a,b\in[m],\; a\neq b, \\ 
&g \left(v_{(1,k[1])}...v_{(m/2,k[m/2])}\right) g  =v_{(1,\alpha_g(k[1]),))}...v_{(m/2,\alpha_g(k[m/2]))}   &&\forall  g \in \mathrm{CL},\ \forall \mathbf{k} \in \{1,2,3\}^{m/2}.
\end{aligned}
\end{equation}
Here, $\mathsf{Bin}:\mathbb{N}^{m/2}\rightarrow \mathbb{Z}_2^m$ denotes the map sending each integer to its associated two-bit binary string, applied componentwise. In particular, $\mathsf{Bin}(1)=01,\ \mathsf{Bin}(2)=10,\ \mathsf{Bin}(3)=11$. The map $\alpha$ denotes the standard relabeling of Pauli observables induced by Clifford conjugation on the variables of the $m/2$-instance parallel Mermin--Peres constraint system, corresponding to the action $\alpha$ in the semidirect product structure $\mathcal{P}_m \rtimes_{\alpha} \mathrm{AGL}(m,2)$. The full relabelling of generators induced by the Clifford generators in $\mathrm{CL}$ within the composed constraint system is illustrated in \Cref{Fig:ConstraintsComposition}.

For convenience, we express the constraint systems using group-presentation notation. We write the parallel Mermin--Peres constraint system as
\begin{equation*}
\langle \{a_{j,\veci[j]}\}_{j\in[m],\veci\in\{0,1,\cdots,5\}^m}\cup\{v_{j,k}\}_{j\in[m],k\in[9]} \mid R_{\mathrm{MP}}\rangle.
\end{equation*}
where $R_{\mathrm{MP}}$ is the set of relators defined in \Cref{def:MP_parallel}. Similarly, let $R_{\mathbb{Z}_2^m}$ denote the relators defining the $\mathbb Z_2^m$ group.

Finally, combining these relator sets with the merge relations $R_{\mathrm{merge}}$ and the global constraint relator from \Cref{def:OVCS_Z2}, which we denote by $R_{\mathrm{global}}$, we obtain the composed constraint system
\begin{equation*}
\langle \{z_{\veci}\}_{\veci}\cup\{a_{j,\veci[j]}\}_{j\in[m],\veci\in\{0,1,\cdots,5\}^m}\cup\{v_{j,k}\}_{j\in[m],k\in[9]}\cup \mathrm{CL} \mid R_{\mathrm{MP}}\cup R_{\mathbb{Z}_2^m} \cup R_{\mathrm{merge}}\cup R_{global}\rangle.
\end{equation*}
We refer to this as the \emph{Mermin--Peres extended faithful $\mathbb{Z}_2^m$ operator-valued constraint system}. The composed structure realizes an amalgamated product between the faithful $\mathbb{Z}_2^m$ constraint system and the parallel Mermin--Peres system, together with the induced Clifford action, yielding a presentation isomorphic to $\bigl(\mathbb Z_2^m *_{\mathbb Z_2^m} \mathcal{P}_m \bigr)\rtimes_{(\tau,\alpha)} \mathrm{AGL}(m,2)$.

\begin{figure}[htb]
\centering
\begin{tikzpicture}[>=Stealth, scale=2]
\tikzset{
    dot/.style={circle, fill=black, inner sep=1.2pt},
    path_style/.style={->, thick},
    label_style/.style={font=\small, inner sep=2pt, fill=white, fill opacity=0.8, text opacity=1, rounded corners=1pt},
    divider/.style={ultra thick, gray, dash pattern=on 8pt off 4pt},
    region_label/.style={font=\large\bfseries, gray!70}
}

\draw[divider] (-0.4, 1.7) -- (-0.4, -0.5);

\draw[divider] (-0.4, 1.2) -- (4.9, 1.2);

\node[dot] (zi)      at (0, 1.6)    {}; \node[above=2pt, label_style] at (zi) {$z_{\mathbf{i}}$};
\node[dot] (zsigma)  at (1.5, 1.6)  {}; \node[above=2pt, label_style] at (zsigma) {$z_{\sigma_{a,b(\mathbf{i})}}$};
\node[dot] (ziplus)  at (3, 1.6)    {}; \node[above=2pt, label_style] at (ziplus) {$z_{\mathbf{i}\oplus\mathbf{e}_a}$};
\node[dot] (zieb)    at (4.5, 1.6)  {}; \node[above=2pt, label_style] at (zieb) {$z_{\mathbf{i}\oplus \mathbf{i}[a]\mathbf{e}_b}$};

\node[dot] (cnot)    at (-0.8, 1.2) {}; \node[left=2pt, label_style]  at (cnot) {$\mathsf{c}_{a,b}$};

\node[dot] (v1)      at (0.75, 0.8) {}; \node[below=2pt, label_style] at (v1) {$v_{j,1}$};
\node[dot] (v2)      at (2.25, 0.8) {}; \node[below=2pt, label_style] at (v2) {$v_{j,2}$};
\node[dot] (v3)      at (3.75, 0.8) {}; \node[below=2pt, label_style] at (v3) {$v_{j,3}$};

\node[dot] (X)       at (-0.8, 0.4) {}; \node[left=2pt, label_style]  at (X) {$x_a$};

\node[dot] (aj0)     at (2.25, 0.0) {}; \node[right=2pt, label_style] at (aj0) {$a_{j,0}$};

\node[dot] (swap)    at (-0.8, -0.4) {}; \node[left=2pt, label_style]  at (swap) {$\mathsf{s}_{a,b}$};

\begin{scope}
    \foreach \rone in {zi, zsigma, ziplus, zieb} {
        \foreach \rtwo in {v1, v2, v3} {
            \draw [gray, thin] (\rone) -- (\rtwo);
        }
    }
    \foreach \rtwo in {v1, v2, v3} {
        \draw [gray, thin] (\rtwo) -- (aj0);
    }
\end{scope}


\draw [path_style, black] (zi) .. controls ($(zi)+(-0.4, -0.1)$) and ($(cnot)+(0.2, 0.2)$) .. (cnot);
\draw [path_style, black] (cnot) .. controls ($(cnot)+(0.6, 0.6)$) and ($(zieb)+(-0.6, 0.2)$) .. (zieb);

\draw [path_style, cmykBlue] (zi) .. controls ($(zi)+(-0.6, -0.6)$ ) and ($(X)+(0.3, 0.3)$) .. (X);
\draw [path_style, cmykBlue] (X) .. controls ($(X)+(1.2, -0.2)$) and ($(ziplus)+(-0.3, -0.8)$) .. (ziplus);

\draw [path_style, cmykRed] (zi) .. controls ($(zi)+(-1.0, -1.2)$) and ($(swap)+(0.3, 0.3)$) .. (swap);
\draw [path_style, cmykRed] (swap) .. controls ($(swap)+(1.2, -0.3)$) and ($(zsigma)+(-0.3, -1.2)$) .. (zsigma);
\end{tikzpicture}
\caption{The composition of the tests involves overlapping variables.  Vertices represent variables, and the vertical dashed line separates unitary-transformation variables (i.e.\ $\mathsf{c}_{a,b},x_a,\mathsf{s}_{a,b}$) from observable variables that are measured. Undirected edges connect compatible variables that appear together in a constraint and are therefore jointly measurable. All the variables on the right side are abstractly presented, and we separate the ones native to the $\mathbb{Z}_2^m$ operator-valued constraint system ($z_{\mathbf{i}},z_{\sigma_{a,b(\mathbf{i})}},z_{\mathbf{i}\oplus\mathsf{e}_a},z_{\mathbf{i}\oplus\veci[a]\mathbf{e}_b}$) from those native to the (parallel) Mermin--Peres constraint system ($v_{j,1},v_{j,2},v_{j,3},a_{j,0}$). Directed paths indicate how $z_{\veci}$-type variables are mapped to each other: the intermediate vertex represents the unitary acting on the observable associated with the initial vertex, and the final vertex the transformed observable. We omit the unitary transformation over $v_{j,i}$'s (and hence $a_{j,l}$'s) for clarity, and present it separately in \Cref{Fig:CliffordtoMermin}. The notation is as used in \Cref{def:MP_parallel,def:OVCS_Z2}.
}
\label{Fig:ConstraintsComposition}
\end{figure}
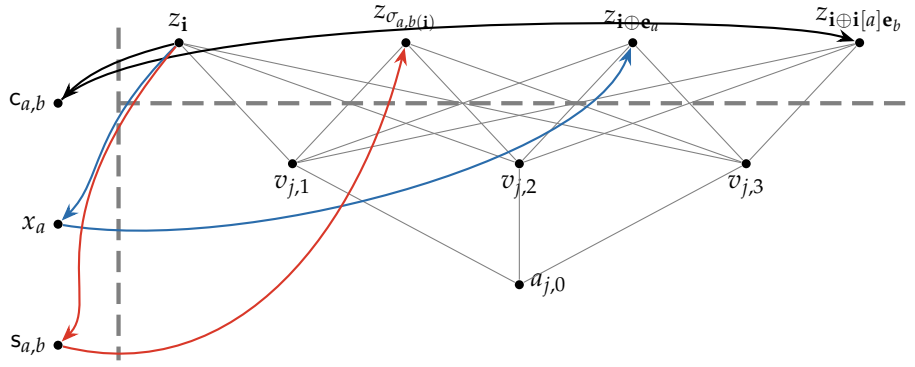

\begin{lemma}
\label{lem:Extended_Z2_faithful_rep}
Any quantum operator-valued solution to the Mermin--Peres extended faithful $\mathbb Z_2^m$ operator-valued constraint system is equivalent to the solution in which the variables $z_{\veci}$ are represented by the diagonal matrices from \Cref{lem:Z2_faithful_rep}, the variables $v_{j,k}$ are represented by the two-qubit Pauli observables from Eq.~\eqref{eq:MPqubits}, and the variables $a_{j,\veci[j]}$ are determined from the $v_{j,k}$ according to \Cref{def:MP_parallel}.
\end{lemma}
\begin{proof}
First, consider the parallel Mermin--Peres constraint system defined by the parallel anti-commutation relations in Eq.~\eqref{eq:parallelanti}. By the rigidity results for the parallel repeated Mermin--Peres game \cite{coudron2016parallel}, the tensor product of the optimal single-instance solutions is the unique quantum operator-valued solution to this system. In particular, the variables $a_{j,\veci[j]}$ are realized by the corresponding Pauli observables.

We now extend the system by adding the relators of the composed constraint system. Since adding relators can only restrict the set of operator-valued solutions, every solution of the extended system must satisfy the original relations as well. We use this monotonicity property later in our analysis of the composed system.

Identifying the corresponding variables $v_{j,k}$ from \Cref{def:MP_parallel} with a subset of the Pauli strings $Z(\vecs)$, the parity-Pauli identification constraints from \Cref{def:OVCS_Z2} can be written as
\begin{equation}\label{eq:MPAbstract}
\prod_{\veci: \veci\cdot \vecs = 1} z_{\veci} = Z(\vecs),\quad\quad  \forall \vecs\in \{01,10,11\}^{m/2}.
\end{equation}

These relations do not yet recover the full family of parity--Pauli identification constraints, since \Cref{def:OVCS_Z2} requires one such relation for every nonzero $\vecs\in\{0,1\}^m$. However, once the Clifford relators are adjoined, whose action on the parallel Mermin--Peres variables is fixed through the conjugation relations, the remaining Pauli strings are generated by Clifford conjugation. For this reason, we henceforth identify the Clifford generators with their corresponding Clifford unitaries.  

In particular, the conjugation identity
\begin{equation*}
    \mathsf{CNOT}(Z\otimes Z)\mathsf{CNOT} = I\otimes Z,
\end{equation*}
allows existing Pauli-$Z$ strings to be transformed into new ones. By composing such Clifford conjugations, every observable $Z(\vecs)$ appearing in \Cref{def:OVCS_Z2} can be generated, thereby recovering the full family of parity--Pauli identification constraints.

Consequently, after fixing the $v$-variables, the remaining relations on the generators $\{z_{\veci}\}_{\veci}$ reduce exactly to the constraint system of \Cref{def:OVCS_Z2}. Therefore, following the proof of \Cref{lem:Z2_faithful_rep}, the unique solution is
$z_{\veci} = \operatorname{diag}(1,\dots,1,\underbrace{-1}_{i\text{-th position}},1,\dots,1)$ for all $\veci \in \{0,1\}^m$. This uniquely determines the entire solution for the composed constraint system.
\end{proof}

\subsection{A Quantum Depth Hierarchy}\label{subsec:final_depth}
Finally, we show that the composed constraint system obtained from the parallel Mermin--Peres game together with the $\textsf{HypGame}$ described in \Cref{subsec:composing} gives rise to a two-round interactive protocol (\Cref{def:Classical_Int}) with a computationally efficient classical verifier. Moreover, the protocol robustly enforces, up to the natural equivalences, the unique operator-valued solution of the composed constraint system.

Therefore, this protocol yields both a separation between classical and quantum computation and a depth-sensitive separation within quantum computation itself. Classical bounded fan-in circuits cannot succeed perfectly, as they are unable to reproduce the required quantum correlations, while quantum circuits below the required depth threshold cannot realize the enforced strategy exactly, since they cannot prepare the necessary observables. In contrast, there exists a quantum strategy of sufficient depth over a finite universal gate set that succeeds with probability~$1$.

\begin{algorithm}[htb]
\floatname{algorithm}{Protocol}
\caption{\textsf{Interactive HypGame} with a Classical Verifier}\label{def:Classical_Int}
\begin{algorithmic}[1]
\Require Parameters $m=2^d,k\in \mathbb{N}$. 

\State Let $t\gets 0$.
\For{$i=1:k$}

\vspace{0.2mm}
\Statex {\bfseries Round 1: State commitment.}
\vspace{0.2mm}

\State Verifier performs the first round of the Clifford rotated parallel EPR self tests (\Cref{def:Cliff_EPR_self}).

\vspace{0.2mm}
\Statex {\bfseries Round 2: Clifford-basis+Boolean Hypercube test.}
\vspace{0.2mm}

\State Verifier updates the question pair $(q_A^{i},q_B^i)$, based on $\mathcal{C}_c$ and $y$, and sends to prover 
\Statex \hspace{\algorithmicindent}$s_{(q_A,q_B)}=0^{l-1} \times \underbrace{q_A^{i}(y)}_{l} \times 0^{j-l-1}  \times \underbrace{q_B^i(y)}_{j}\times 0^{n-j}$.

\State Prover sends back its answer $s_{(a_A,a_B)}=\ \{0,1\}^{l-1} \times a_A^i \times \{0,1\}^{j-l-1}\times a_B^i \times \{0,1\}^{n-j}$ .
\State Verifier checks correctness: $t\gets t+1$ if $s_{(a_A,a_B)}$ is a valid answer. 

\EndFor

\State \Return the fraction of accepted rounds $t/k$.
\end{algorithmic}
\end{algorithm}

\begin{theorem}[Interactive problem with classical messages]
\label{thm:hierarchy-classical-verifier}
Fix an integer $d>d_0$ and consider circuits of fan-in $2$. For the interactive 2-round problem $\mathcal{IR}_{m}^n$ defined in \Cref{def:Classical_Int}, we have:

\begin{itemize}
\item \textbf{Perfect completeness (upper bound).}
There exists a circuit family $\{C_{n}\}_{n\in\mathbb{N}}$ of depth $c\cdot d$ and bounded fan-in that solves $\mathcal{IR}_{m}^n$ with success probability $1$ on every promised input.

\item \textbf{Soundness against shallower circuits (lower bound).}
Every quantum circuit family of depth $d'\leq d-1$ can succeed with probability at most $1-\varepsilon(d)$, irrespective of the size of the circuit, the underlying gate set, and the number of ancillary qubits or access to quantum advice. 

\item \textbf{Soundness against classical shallow circuits.} 
  No classical bounded fan-in circuit family of constant-depth $(\mathsf{NC}^0)$ can succeed with probability exceeding $17/18$. 
\end{itemize}
Moreover, the verifier transcript can be computed in $\NC^0[\oplus]$. The constant $c = 2 + 2C_4 + C_3$, where $C_4$ and $C_3$ denote the depths for the realization of 4- and 3-qubit controlled Toffoli constructions.  Simultaneously, the soundness gap satisfies $\varepsilon(d)\geq 2^{-O(2^d)}$. Finally, $d_0\leq 12$ defines the minimal depth of the quantum circuits in the first interaction. 
\end{theorem}

\paragraph{Proof of perfect completeness.} For each input size $n$, let $C_n$ denote the honest classically controlled quantum circuit achieving perfect completeness. The circuit receives, in each round of the protocol, the verifier’s classical messages specifying the corresponding labels.

In the first round, the circuit uses the same construction as given in \Cref{lemma:cliff_test_qnc}. In particular, it performs entanglement swapping and gate teleportation in constant depth (modulo the quantum error correction in teleportation) in order to prepare an $m$-fold tensor of EPR pairs between the two designated registers, in the Clifford-rotated basis specified by the verifier.

In the second round, the verifier provides queries derived from the Mermin--Peres extended faithful $\mathbb Z_2^n$ operator-valued constraint system defined in \Cref{subsec:composing}. By construction, these queries locally match the structure of the questions appearing in both \Cref{def:Cliff_EPR_self,def:QI_LBCS}. This follows from the fact that the Mermin--Peres extended faithful $\mathbb Z_2^n$ operator-valued constraint system is obtained by composing the two sub-constraint systems.

Consequently, the circuit can implement the same optimal strategy used in \Cref{thm:QPlain} and \Cref{lemma:cliff_test_qnc}. Concretely, for the $v$-type Mermin--Peres variables, the circuit performs the corresponding Pauli measurements, while for $z$-type $\mathbb{Z}_2^m$ variables, it implements the corresponding multi-controlled $\mathsf{C}^{m-1}(\mathsf{Z})$ observables, as described in \Cref{fig:full_framework}. All these operations can be realized in constant depth. As these observables determine the operator-valued solution to the composed constraint system, as shown in \Cref{lem:Extended_Z2_faithful_rep}, the uniform family of circuits $C_n$ shall output valid answers to the nonlocal game questions with probability $1$. This establishes the perfect completeness of the protocol.

Moreover, the update of the questions, and the computation of the Pauli correction parities, act on bit strings of length $m$. These operations correspond to affine transformations over $\mathbb{F}_2^m$, i.e., elements of $\mathrm{AGL}(m,2)$. Such transformations can be implemented by Boolean circuits of depth $O(\log m)$. Since $m$ is a fixed constant in our setting, this yields an overall constant depth with respect to the input size $n$ and the verifier transcript can be computed in $\NC^0[\oplus]$.
\hfill\qedsymbol

Operationally, the first round prepares a distributed entangled resource state through entanglement swapping between non-communicating regions of the circuit. The computational content underlying the depth hierarchy is then concentrated in the second round, where the verifier certifies the prepared state while enforcing consistency with the intended global algebraic structure. In this way, the interaction is not used to hide computation in a complex verifier, but rather to isolate quantum depth as the source of the required nonlocal correlations.

\paragraph{Proof of soundness (against quantum provers).} Soundness against insufficient-depth $\QNC^0$ circuits follows by first observing that the interactive protocol implements the composed constraint system defined in \Cref{subsec:composing}, and therefore rigidly enforces its unique operator-valued solution.

We start by considering that the interactive problem of \Cref{def:Classical_Int} contains, as a subprotocol, the Clifford-rotated EPR self-test of \Cref{def:Cliff_EPR_self}. Hence, any circuit that succeeds with probability $1$ in the full protocol must, in particular, succeed perfectly in this self-testing component. Although not every question in the full protocol contributes directly to the certification of the first-round state preparation, this does not weaken the argument: the prover only learns which type of second-round test is being performed after the first-round state preparation has already been completed. Therefore, the prover cannot condition its first-round behavior on whether the subsequent test is “only” an EPR self-test or part of the larger constraint-system game. In particular, no contextual cheating strategy is available at this stage, and the rigidity statement of \Cref{lemma:cliff_test_qnc} continues to apply.

Moreover, the resulting state is equivalent to the resource state required in \Cref{thm:QPlain}. Hence, once this state is certified (e.g.\ by sufficiently many repetitions of the protocol), the soundness argument reduces to that of \Cref{lemma:cliff_verif} for the remaining constraints of the Mermin--Peres extended faithful $\mathbb Z_2^m$ operator-valued constraint system. Indeed, on the $z$-variables, the enforced operator-valued solution coincides with the faithful solution of the constraint system in \Cref{def:OVCS_Z2} underlying \Cref{def:QI_LBCS}. Consequently, \Cref{def:Classical_Int} enforces the target composed constraint system and imposes the unique operator-valued solution characterized in \Cref{lem:Extended_Z2_faithful_rep}.

We furthermore establish a robust version of \Cref{def:Classical_Int} in \Cref{lemma:rob_sec5}; see \Cref{subsec:robust5} for the full analysis. Specifically, we show that any strategy achieving success probability at least $1-\varepsilon$ must implement observables that are $2^{O(m)}\mathsf{poly}(\varepsilon)$-close, up to local isometries, to the unique operator-valued solution characterized in \Cref{lem:Extended_Z2_faithful_rep}. Thus, near-perfect success certifies an approximate realization of the same faithful representation enforced in the perfect setting.

Combining this robust rigidity statement with the approximate unitary-synthesis lower bounds established earlier, we obtain the same depth lower-bound mechanism as in \Cref{lemma:cliff_verif}. In particular, any strategy achieving sufficiently high success probability must approximately implement the family of multi-controlled phase (equivalently, generalized Toffoli) observables enforced by the protocol. Since realizing such observables within the required precision demands circuit depth at least $d \geq \lfloor \log m \rfloor$, it follows that no family of bounded fan-in $\QNC^0$ circuits of depth $o(\log m)$ can achieve success larger than $1-2^{-\Omega(m)}$. This completes the soundness proof in the robust setting.
\hfill\qedsymbol

\paragraph{Proof of soundness (against classical provers).} For soundness against classical provers, we first observe that the nonlocal game defined in the second round of \Cref{def:Classical_Int} admits no perfect classical strategy. Indeed, the protocol embeds instances of the Mermin--Peres nonlocal game, whose successful execution requires correlations corresponding to Pauli-string measurements, which cannot be perfectly reproduced classically \cite{bravyi2020quantum}.

Moreover, the first round only allows the prover to generate classical information independently of the second-round questions. Since no entanglement is available, the resulting strategy reduces to a classical strategy for the underlying nonlocal games without communication. In particular, the prover cannot adapt its answers to the specific constraint tested in the second round. Perfect success would therefore imply a perfect classical strategy for the Mermin--Peres game without communication, which is impossible \cite{peres1991two,mermin1990simple}. Achieving the required global coordination instead requires circuit depth at least $\Omega(\log n)$.

To obtain an explicit bound, we isolate a single Mermin--Peres constraint within the parallel repeated instances appearing in the composed constraint system. It is well known that no classical strategy wins this game with probability exceeding $8/9$ \cite{bravyi2020quantum}. Furthermore, the composed constraint system contains $2^{m+1}$ constraints involving only $z$-type variables, but more than $6^{m/2}$ constraints arising from the parallel Mermin--Peres component. Since the latter dominate asymptotically, the overall success probability of any classical strategy is governed by the Mermin--Peres constraints and therefore cannot asymptotically exceed the corresponding classical value.

This proves that classical provers cannot achieve perfect success, and in particular cannot do so using $o(\log n)$-depth circuits. \hfill\qedsymbol

\section*{Acknowledgements}

Mdo thanks Angelos Bampounis and Matthew Coudron for useful feedback. SS acknowledges support from the Royal Society through a University Research Fellowship. XZ acknowledges Hong Kong Research Grant Council (RGC) and the Chancellor's Research Fellowship scheme provided by the University of Technology Sydney, and thanks Honghao Fu for useful feedback.

\printbibliography

\appendix

\section{Game robustness}\label{App:Rigidity}

For the semi-quantum games and corresponding protocols introduced in \Cref{subsec:Quantum_depth} and \Cref{sec:Classical_int}, we establish rigidity results: in the ideal case ($\varepsilon = 0$), any perfect strategy must implement the operator-valued solutions to the corresponding constraint systems, as characterized in \Cref{lem:Z2_faithful_rep} and \Cref{lem:Extended_Z2_faithful_rep}, respectively. However, these exact rigidity statements do not directly extend to the robust setting, where the players may succeed with probability $1-\varepsilon$ for some $\varepsilon > 0$. In particular, it is not clear that near-optimal strategies remain close to the corresponding operator-valued solutions.

A key feature of our setting is that the game admits a group-theoretic description: valid strategies correspond to representations of a solution group, and the intended (honest) strategy realizes a distinguished, rigid representation of this group. This viewpoint suggests following a standard approach in the literature, whereby approximate satisfaction of the defining relations is converted into closeness to an exact representation via stability results for approximate representations. 

We seek to leverage the following stability theorem for approximate representations of finite groups. 
\begin{theorem}[Informal statement from \cite{gowers2017generalizations}] Let $G$ be a finite group and $f : G\mapsto U(\mathbb C^n)$ be such that $||f(x)f(y)−f(xy)||_2 \leq \varepsilon \sqrt{n}$ for all $x,y \in G$. Then there exists $m \leq (1 + \varepsilon^2)n$, an isometry $V : \mathbb C^n \mapsto \mathbb C^m$, and a unitary representation $\sigma :G \mapsto U( \mathbb C^m)$, such that $||f(x)−V^\dagger \sigma(x)V||_2 \in O(\varepsilon \sqrt{n})$ for every $x\in G$.
\end{theorem}

A variant, due to Vidick, adapts this to the context of nonlocal games, wherein one must work with \emph{state-dependent} notions of distance, rather than operator norms.
\begin{lemma}[Formalized and proven in \cite{coladangelo2017robust}]\label{lemma:Vidick}
Let $G$ be a finite group, and let $f : G \to U(\mathcal{H}_A)$ be such that $f(x^{-1}) = f(x)^\dagger \,.$ Let $\rho_{AB}$ be a state on $\mathcal{H}_A \otimes \mathcal{H}_B$ and
\begin{equation*}
\E_{x,y \in G} D_\rho  \left ( f(x) f(yx)^\dagger f(y) \otimes I_B\ \Big \|\  I_{AB} \right ) \le \eta\ .
\end{equation*}
Then there exists a Hilbert space $\widehat{\mathcal{H}}_A$, an isometry $V : \mathcal{H}_A \to \widehat{\mathcal{H}}_A \,,$ and a representation $\tau : G \to U(\widehat{\mathcal{H}}_A)$ such that
\begin{equation*}
\E_{x \in G} D_\rho \, \left (  V f(x) V^\dagger \otimes I_B\ \Big \|\  \tau(x)  \otimes I_B \right ) \le \eta.
\end{equation*}
where $D_\rho(A\ \|\ B)=\mathsf{Tr}(\rho (A-B)^\dagger(A-B))$.
\end{lemma}

Thus if the relations defining the group are approximately satisfied (in a state-dependent sense), then the approximate operators $f(x)$ are close to a genuine representation of the group, up to an isometry. This would allow us to conclude that any near-optimal strategy must be close to the representation enforced by our construction. 

In the following two subsections we establish the ingredients required to apply \Cref{lemma:Vidick} and prove robustness for each of the constructions introduced in  \Cref{subsec:Quantum_depth} and \Cref{sec:Classical_int}.

\subsection{Faithful \texorpdfstring{$\mathbb Z_2^m$}{Z\_2\^{}m} operator-valued BCS}\label{subsec:rob_Q}

Recall that in \Cref{subsec:Quantum_depth}, the verifier prepares and distributes a bipartite quantum state to two non-communicating provers, who are assumed not to share any prior entanglement. Hence, the joint state available to the provers at the start of each round
is exactly the verifier-prepared state, tensored with local ancillary registers.

In this setting, robustness reduces to an operator-rigidity statement: any near-optimal strategy must implement measurements that are close, on the verifier-prepared states, to the intended observables.

\begin{definition}[Robust operator rigidity]\label{def:OperatorRigidity}
Let $G$ be a semi-quantum game in which the verifier sends classical question labels $(x,y)$ together with a corresponding bipartite quantum state $\rho_{x,y}$. Suppose that $G$ self-tests a strategy $S = \{\{A_x\}_x, \{B_y\}_y\}$. We say that $G$ is $\delta(\varepsilon)$-rigid if for any strategy $\{\tilde A_x, \tilde B_y\}$ achieving success probability at least $\omega^*(G) - \varepsilon$, there exist local isometries $V_A, V_B$ and an auxiliary state $\rho_{\mathrm{junk}}$ such that for all $(x,y)$,
\begin{align}
\Big\| (V_A \otimes V_B)\big(\tilde A_x \otimes \tilde{B}_y\big)\rho_{x,y}\otimes \rho_{junk}&\big(\tilde{A}_x \otimes \tilde{B}_y\big)^\dagger (V_A^\dagger \otimes V_B^\dagger)\nonumber\\ 
 &-\big(A_x \otimes B_y\big)\rho_{x,y}\big(A_x \otimes B_y\big)^\dagger \otimes \rho_{\mathrm{junk}} \Big \|_2\le \delta(\varepsilon),
\end{align}
for a decaying function $\delta$.
\end{definition}

Starting from an approximate strategy for \Cref{def:QI_LBCS}, based on the semi-quantum game induced by the operator-valued constraint system in \Cref{def:OVCS_Z2}, and achieving a winning probability $p_{\mathrm{win}}$, we obtain that the corresponding observables satisfy the following bounds.

\begin{lemma}\label{lemma:prob_to_dist}
Let a two-prover strategy be given by observables $\{\tilde A_v\}_v$ and $\{\tilde B_v\}_v$. 
Let $p_{\mathrm{con}}$ be the probability of passing the consistency checks, $p_{\mathrm{sat}}$ the probability of passing the constraint satisfaction checks, and $p_{\mathrm{win}}$ the overall success probability. Then $p_{\mathrm{win}} \le \min\{p_{\mathrm{con}}, p_{\mathrm{sat}}\}$.
Moreover, the following bounds hold:
\begin{equation}
\E_{r,v,\,c}\ \frac{1}{4} \left [D_\rho\left(\tilde{A}_v^r \otimes \mathcal{C}_c\ \big(\tilde{B}_{g(v,c)}\big)\mathcal{C}_c^\dagger\ \Big\|\ I \right)^2\right ] \le 1- p_{\mathrm{win}},
\end{equation}
\begin{equation}
\E_{r'}\ \frac{1}{4}\left[ D_\rho\left(\prod_{v \in r} \tilde A_v \otimes I_B\ \Big\|\ (-1)^{\lambda_{r'}} I\right)^2 \right] \le 1 - p_{\mathrm{win}}.
\end{equation}
Here $c$ is sampled uniformly from $\{\mathsf{CNOT}, \mathsf{Swap}, \mathsf{X}, I\}$, $g:\{\{0,1\}^m,\{\mathsf{CNOT}, \mathsf{Swap}, \mathsf{X}, I\}\}\mapsto \{0,1\}^m$ relabels indexes based on Clifford conjugations rules defined by the affine symmetry relations in \Cref{def:OVCS_Z2}, $r$ ranges over the constraints of the system \Cref{def:OVCS_Z2}, and $r'$ on all except the involution and affine symmetry constraints, $\lambda_r \in \{\pm 1\}$ is the target value of constraint $r$, and $\rho = \ket{\Phi}\bra{\Phi}_{A,B} \otimes \rho_{\mathrm{junk}}$.
\end{lemma}

The proof follows directly from \Cref{def:QI_LBCS}, and the use of the binary observables by the provers. 

\medskip

Note that Bob’s observables are necessarily non-contextual: since he receives only a single variable and no information about Alice’s constraint, his strategy is specified by a single family $\{\tilde B_v\}_v$. Consistency checks then imply that these operators approximately satisfy the constraint relations, in particular, approximately commuting whenever they appear together.

In contrast, Alice’s observables are indexed by both a variable and a restriction: $\tilde A_v^{(r)}$ denotes Alice’s observable for variable $v$ when she is asked restriction $r$. In the approximate setting, this allows for the possibility that these observables are contextual with respect to the chosen restriction. Nevertheless, for the robustness proof, we will work with a single operator per variable. To this end, for each variable $v$, we fix one restriction $r_v$ containing $v$ and define $\tilde A_v := \tilde A_v^{(r_v)}$. This choice is ultimately immaterial, as the argument would proceed identically for any such selection, and the resulting statements for the chosen observables extend to all others. Indeed, by consistency, any operator $\tilde A_v^{(r)}$ corresponding to another restriction $r$ is close to the same Bob operator $\tilde B_v$, and hence also close to the chosen representative $\tilde A_v$.

The following lemma shows that the fixed representatives $\{\tilde A_v\}_v$, together with Bob’s observables $\{\tilde B_v\}_v$, approximately satisfy the relations of the solution group. Thus, both define approximate operator solutions to the same system. 

\begin{lemma}\label{lemma:set_of_dif}
Let $\{\tilde B_v\}$ be Bob’s observables and $\{\tilde A_v^{(r_v)}\}$ Alice’s observables. Assume the strategy succeeds with probability at least $1-\varepsilon$. Then $\{\tilde B_v\}$ is an approximate operator solution in the sense that
\begin{equation}
\sum_{r'} D_\rho \left(\prod_{v \in r'} I_A \otimes  \tilde{B}_{v} ^\dagger\ \Big \| \  (-1)^{\lambda_{r'}}I\right)  \le\  O(m ^3 2^{2m} \sqrt{\varepsilon}),
\end{equation}

\begin{equation}
\sum_{v,v',c}D_\rho \left(I_A \otimes [\tilde{B}_v,\mathcal{C}_c\tilde{B}_{v'}\mathcal{C}_c^\dagger]\ \Big \| \ I \right)\le O(m^6 2^{2m} \sqrt{\varepsilon}).
\end{equation}

The same bounds hold for Alice operators $\{\tilde A_v^{(r_v)}\}$ in place of Bob’s observables $\{\tilde B_v\}$.
\end{lemma}
\begin{proof}
We start from
\begin{equation}
\E_{r,v,\,c}\ \frac{1}{4} \left [D_\rho\left(\tilde{A}_v^r \otimes \mathcal{C}_c\ \big(\tilde{B}_{g(v,c)}\big)\mathcal{C}_c^\dagger \  \Big\|\  I \right)^2\right ] \le \varepsilon. 
\end{equation}
We deduce that,
\begin{equation*}
\sum_{r,v,c} D_\rho\left(\tilde{A}_v^r \otimes \mathcal{C}_c\ \big(\tilde{B}_{g(v,c)}\big)\mathcal{C}_c^\dagger\  \Big\|\ I \right) \le 2|r||v||c|\sqrt{\varepsilon}. 
\end{equation*}
\begin{equation}\label{lemma:approx_AvBv}
\sum_{r,v,c} D_\rho\left(\tilde{A}_v^r \otimes I_B\   \Big\|\ I_A \otimes \mathcal{C}_c\ \big(\tilde{B}_{g(v,c)}^\dagger\big)\mathcal{C}_c^\dagger \right) \le  m^3 2^{2m}\sqrt{\varepsilon}.
\end{equation}
\noindent The first equation follows from Cauchy-Schwarz and the second from $D_\rho(UZ\ \|\ I)=D_\rho(Z\ \|\ U^\dagger)$.

We can equally determine the following expression from \Cref{lemma:prob_to_dist}, 
\begin{equation}\label{lemma:approx_Av}
\sum_{r'}\  D_\rho\left(\prod_{v \in r'} \tilde A_v\  \Big\|\  (-1)^{\lambda_{r'}} I\right)  \le 2^m \sqrt{\varepsilon}.
\end{equation}

\noindent This allows us to show that Bob's operators satisfy approximately the all the constraints,
\begin{align*}
    \sum_{r'} D_\rho\left(\prod_{v \in r'} I_A \otimes \tilde B_v\ \Big\|\ (-1)^{\lambda_{r'}} I\right) \leq  \sum_{r'} D_\rho\left(\prod_{v \in r'} I_A \otimes \tilde B_v\ \Big\|\ \prod_{v \in r'} \tilde A_v^{r_v} \otimes  I_B \right) \\ + \sum_{r'}\  D_\rho\left(\prod_{v \in r'} \tilde A_v\  \Big\|\  (-1)^{\lambda_{r'}} I\right) \notag \\
    \leq \sum_{r',v} D_\rho\left(\tilde{A}_v^r \otimes I_B \ \Big\|\ I_A \otimes \tilde{B}_{v}^\dagger \right) + \sum_{r'}\  D_\rho\left(\prod_{v \in r'} \tilde A_v\ \Big\|\ (-1)^{\lambda_{r'}} I\right)
    \leq (2^m+ m^3 2^{2m})\sqrt{\varepsilon}.
\end{align*}

The first inequality follows from $D_\rho(Z_1 \| Z_3)\leq D_\rho (Z_1\|Z_2)+ D_\rho (Z_2 \| Z_3)$. The second uses $D_\rho(UZ \| I)=D_\rho(Z \| U^\dagger)$, together with the bound $D_\rho(\prod_i A_i \otimes I_B\ \|\ \prod_i I_A \otimes B_i )\leq \sum_i D_\rho (A_i \otimes I_A\ \|\ I_A \otimes B_i)$ after expanding the products. The third uses \Cref{lemma:approx_Av} and \Cref{lemma:approx_AvBv}, and drops the summation over the Clifford labels in the latter bound. 

Subsequently, for the commutation of Bob's operators we have that, 
\begin{equation*}
    D_\rho\left( I_A \otimes [\tilde  B_v,\mathcal{C}_c \tilde B_{v'}\mathcal{C}_c^\dagger]\ \big \| \ I\right) \leq 2 D_\rho\left(\tilde A_v^{r_v} \otimes \tilde B_v \ \big \| \ I\right)+2 D_\rho\left(\tilde A_{g(v',c)}^{r_v} \otimes \mathcal{C}_c \tilde B_{v'}  \mathcal{C}_c^\dagger\ \big \|  \ I\right)
\end{equation*}
This follows, again from $D_\rho(Z_1\|Z_3)\leq D_\rho (Z_1\|Z_2)+ D_\rho (Z_2 \| Z_3)$ and triangle inequalities. This allows us to bound the expression on average as follows,
\begin{equation}
\sum_{v,c} D_\rho \left(I_A \otimes [\tilde B_v,\mathcal{C}_c \tilde B_{v'}\mathcal{C}_c^\dagger]\ \big \| \  I\ \right) \leq 2 m^3  \sum_{r,v,c} D_\rho\left(\tilde{A}_v^r \otimes \mathcal{C}_c\ \big(\tilde{B}_{g(v,c)}\big)\mathcal{C}_c^\dagger\  \big\| \ I \right)\leq m^62^{2m}\sqrt{\varepsilon}.
\end{equation}

Finally, we can now use Bob's contextual operators and bounds, to derive the same bounds on the pre-selected set of observables $\{\tilde A_v^{(r_v)}\}$ by Alice.
\end{proof}

So far, we have assigned observables only to the generators of the group, together with bounds showing that these approximately satisfy the relations tested in the nonlocal game. To apply the stability lemma, however, we require observables associated with arbitrary group elements.

To obtain a well-defined extension, we fix a canonical form for each group element and define its observable as the ordered product of the observables corresponding to the generators appearing in this canonical form. This avoids ambiguities arising from different decompositions of the same element in the approximate setting.

\begin{lemma}\label{lemma:can_form}
Every element $x \in G_{\mathbb Z_2^m}$ admits a canonical decomposition
\begin{equation}
can(x) = \operatorname{can}_{\mathbb Z_2^m}(z)\operatorname{can}_{\mathcal{P}_m}(p)\operatorname{can}_{\mathrm{AGL}(m,2)}(g),
\end{equation}
for some $z \in \mathbb Z_2^m$, $p \in \mathcal P_m$, $g \in \mathrm{AGL}(m,2)$, with each factor is expressed in its respective canonical form.
\end{lemma}
\begin{proof}
Any element of $G_{\mathbb Z_2^m}$ is a word in generators from $\mathbb Z_2^m$, $\mathcal P_m$, and $\mathrm{AGL}(m,2)$. Using the relations of \Cref{def:OVCS_Z2}, we reorder these generators into a fixed canonical order.

First, elements of $\mathbb Z_2^m$ can be commuted through $\mathcal P_m$. Indeed, the $X$-type Paulis act on $\mathbb Z_2^m$ according to the prescribed conjugation relations, while the $Z$-type Paulis are identified with elements of $\mathbb Z_2^m$ (as these instantiate as the $\mathsf{C}^{m}(\mathsf{Z})$ operators discussed before). Next, elements of $\mathrm{AGL}(m,2)$ are represented by Clifford circuits generated by
$\{\mathsf{CNOT},\mathsf{Swap},\mathsf{X},\mathsf{I}\}$. By the affine symmetry relations, these generators can likewise be commuted through $\mathbb Z_2^m$, modifying only the corresponding indices.

Finally, Clifford operators conjugate Pauli operators to Pauli operators, allowing all elements of $\mathrm{AGL}(m,2)$ to be moved to the right of the Pauli sector. Since both $\mathcal P_m$ and $\mathrm{AGL}(m,2)$ admit canonical forms \cite{Clifford_sim}, every element of $G_{\mathbb Z_2^m}$ can therefore be written uniquely as intended.
\end{proof}

The canonical decomposition allows us to extend the approximate operator assignment to all group elements in a well-defined way, providing the input required for the stability lemma.

\begin{definition}\label{def:apx_map}
Let $f_A : G_{\mathbb Z_2^m} \to U(\mathcal{H}_A)$ and $f_B : G_{\mathbb Z_2^m} \to U(\mathcal{H}_B)$ be defined by
\begin{equation*}
f_A(x) =
\begin{cases}
- I & \text{if } x = J, \\
\tilde A_x^{r_x} & \text{if } \operatorname{can}(x)\in \mathbb Z_2^n , \\
x \text{ or }\mathsf{Clifford}(x)  & \text{if } \operatorname{can}(x)\in \mathcal{P}_m \text{ or } \mathrm{AGL}(m,2)\text{ resp.} , \\
\displaystyle \prod_{e \in \operatorname{can}(x)} f_A(e) & \text{otherwise},
\end{cases}
\end{equation*}
with $can(x)$ as defined in \Cref{lemma:can_form} and  $\mathsf{Clifford}(x)$ being a unique Clifford circuit defined according to the canonical form introduced in \cite{Clifford_sim}. $f_B$ is defined equally with the difference being that $f_B(x)=B_e$ if $\operatorname{can}(x) = e$.
\end{definition}

Using the maps $f_A$ and $f_B$ together with the bounds established in \Cref{lemma:set_of_dif}, we can bound the distances between arbitrary group elements, providing the input required to apply the stability lemma and state our main lemma.

\begin{lemma}\label{lem:RQP}
\Cref{def:QI_LBCS} robustly self-tests the operator strategy of~\Cref{lem:Z2_faithful_rep}, in the sense of~\Cref{def:OperatorRigidity}, with robustness $O(m^9 2^{5m} \sqrt{\varepsilon})$ .
\end{lemma}
\begin{proof}
We start by showing that
\begin{equation}\label{eq:pre_sta}
    D_\rho \left( f_A(x) f_A(yx^{-1}) f_A(y) \otimes I_B\ \big \|  \ I \right) \leq O(m^6 2^{3m} \sqrt{\varepsilon})
\end{equation}
\begin{equation}
    D_\rho \left(I_A \otimes  f_B(x) f_B(yx^{-1}) f_B(y)\ \big \| \ I \right) \leq O(m^6 2^{3m} \sqrt{\varepsilon})
\end{equation}
Consider the word
$\operatorname{can}(x)\operatorname{can}(yx^{-1})\operatorname{can}(y)$.
Since $x(yx^{-1})y=1$, reducing this word to the identity amounts to commuting subgroup components into canonical order.

Only commutations involving $\mathbb Z_2^m$ contribute nontrivially to the error. Indeed, by construction (\Cref{def:apx_map}), the $\mathcal P_m$ and $\mathrm{AGL}(m,2)$ sectors are implemented exactly. Thus, errors arise only from the $\mathbb Z_2^m$ generators appearing in the canonical forms, each contributing according to the bounds of \Cref{lemma:approx_AvBv}.

Subsequently, the reduction uses at most $O(m^3)$ commutators of type $[\mathrm{AGL}(m,2),\mathbb Z_2^m]$, at most $O(m^2)$ commutators of the type $[\mathcal P_m,\mathbb Z_2^m]$ and $O(2^m)$ of the type  $[\mathbb Z_2^m,\mathbb Z_2^m]$\footnote{We write $[A,B]$ for commutators of the form $[a,b]$ with $a\in A$ and $b\in B$.}. All remaining commutations are exact and therefore incur no error. Applying the bounds from \Cref{lemma:set_of_dif} together with repeated use of $D_\rho(U_1U_2\|I) \le D_\rho(U_1\|I)+D_\rho(U_2\|I)$. Performing this reduction first for $\operatorname{can}(yx^{-1})\operatorname{can}(y)$
and then for the resulting product with $\operatorname{can}(x)$ yields the claimed bounds.

Now by \Cref{lemma:Vidick}, we obtain that there exist two isometries $V_A$ and $V_B$ such that,
\begin{equation}\label{eq:post_stab}
    \E_x\ D_\rho\left( f_A(x) \otimes I_B\ \big \|\  V_A \sigma(x) V_A^{\dagger} \otimes I_B \right)  \leq O(m^6 2^{3m}\sqrt{\varepsilon})
\end{equation}
with $\sigma$ being a representation of $G_{\mathbb Z_2^m}$ and $\rho=\!_{A,B}\ket{\Phi}\!\bra{\Phi}_{A,B} \otimes \rho_{\mathrm{junk}}$.

Finally, we need to show that the previous bounds also hold pointwise. First, we consider the following bound,
\begin{equation*}
    \E_x\ D_\rho\left(V_A f_A(zx)V_A^{\dagger} \otimes ( \sigma(zx)^{-1}  \otimes I)_B\ \big\|\ I  \right)  \leq O(m^6 2^{3m}\sqrt{\varepsilon}).
\end{equation*}
\noindent This follows from the previous bound on Alice, with a change of variables from $x\mapsto zx$ which does not alter the bound due to its being over the expected value of $x$. 

Subsequently, we use that $\sigma(zx)^{-1}=\sigma(x)^{-1}\sigma(z)^{-1}$ and \Cref{eq:post_stab} to deduce, 
\begin{equation}
    \E_x\ D_\rho\left(V_A f_A(zx)f_A(x)^{\dagger}V_A^{\dagger} \otimes ( \sigma(z) \otimes I)_B\ \big\|\ I  \right)  \leq O(m^6 2^{3m}\sqrt{\varepsilon}).
\end{equation}
Now we can just use \Cref{eq:pre_sta} and the triangle inequality to deduce, 
\begin{equation}
    \E_x\ D_\rho\left(V_A f_A(z)V_A^{\dagger} \otimes ( \sigma(z) \otimes I)_B\ \big\|\ I  \right)  \leq O(m^6 2^{3m}\sqrt{\varepsilon}).
\end{equation}
Finally, we can just drop the expectation operator while obtaining our final robustness bounds. Note that for Bob the same deductions will hold.
\end{proof}

\subsection{Mermin--Peres extended faithful \texorpdfstring{$\mathbb Z_2^m$}{Z\_2\^{}m} operator-valued BCS}\label{subsec:robust5}

The robustness proof for \Cref{sec:Classical_int} follows the same overall strategy as before. However, in the present setting, we must first show that the classically delegated setting remains, up to small error, equivalent to the ideal Clifford-rotated EPR-pair setting considered there. 

Recall that this difficulty results from neither the state preparation nor the Clifford operations being directly implemented by the verifier. Instead, in the three-prover protocol, the Clifford generators are delegated to Charlie via gate teleportation and are therefore only approximately certified through the interaction itself. This obstructs the argument from \Cref{lem:RQP}, where the reduction to the stability lemma relied on exact Clifford relations to robustly reduce arbitrary words in $\mathcal{P}_m \rtimes_{\alpha} \mathrm{GL}(m,2)$ to canonical form. In the present setting, only state-dependent approximate Clifford relations are available, and it is therefore no longer immediate that the same reduction remains valid.

We now explain how to overcome this issue by reducing any approximate strategy for the three-prover protocol to an effective two-prover strategy. Concretely, we write the strategy as $S = \bigl\{\{\rho_c\}_c,\{A_x\}_x,\{B_y\}_y\bigr\}$. Here, $\rho_c$ denotes the post-teleportation state shared by Alice and Bob conditioned on the Clifford operator delegated to Charlie, labelled by $c$. Importantly, the questions in the resulting nonlocal game remain state-dependent, exactly as in \Cref{subsec:Quantum_depth}.

\begin{lemma}\label{lemma:prob_to_dist2}
Let a two-prover strategy be given by observables $\{\tilde A_v\}_v$ and $\{\tilde B_v\}_v$ and a question-dependent state $\rho_c$.
Let $p_{\mathrm{con}}$ be the probability of passing the consistency checks, $p_{\mathrm{sat}}$ the probability of passing the constraint satisfaction checks, and $p_{\mathrm{win}}$ the overall success probability. Then $p_{\mathrm{win}} \le \min\{p_{\mathrm{con}}, p_{\mathrm{sat}}\}$.
Moreover, the following bounds hold:
\begin{equation}
\E_{r,v,c}\ \frac{1}{4} \left [D_{\rho_c}\left(\tilde{A}_v^r \otimes \tilde{B}_{g(v,c)}\ \big\|\  I \right)^2\right ] \le 1- p_{\mathrm{win}},
\end{equation}
\begin{equation}
\E_{r'}\ \frac{1}{4}\left[ D_{\rho_c}\left(\prod_{v \in r} \tilde A_v \otimes I_B \ \big\|\   (-1)^{\lambda_r} I\right)^2 \right] \le 1 - p_{\mathrm{win}}.
\end{equation}
Here $r$ ranges over the constraints of the Mermin--Peres extended faithful $\mathbb Z_2^n$ operator-valued constraint system defined in \Cref{subsec:composing}, while $r'$ ranges over all constraints except the involution and affine symmetry relations. For each constraint $r$, the value $\lambda_r \in \{\pm 1\}$  denotes its target value. Furthermore, $c$ is sampled uniformly from $\{\mathsf{CNOT}, \mathsf{Swap}, \mathsf{X}, I\}$ and $g:\{0,1\}^m \times \{\mathsf{CNOT}, \mathsf{Swap}, \mathsf{X}, I\} \mapsto \{0,1\}^m$ denotes the relabelling map induced by Clifford conjugation according to the affine symmetry relations of the constraint system.
\end{lemma}

The main difference between \Cref{lemma:prob_to_dist} and the previous lemma is that the resulting operator distances remain state-dependent, even after the Clifford conjugations. This is unavoidable in the classically delegated setting, since the conditioned states $\rho_c$ may deviate from the ideal Clifford-rotated EPR states tested by the protocol. To overcome this issue, we derive a bound relating operator distances evaluated on the approximate states to the corresponding distances evaluated on a fixed ideal reference state, together with a trace-distance bound between the approximate and ideal states.

\begin{lemma}[State-dependent stability under state perturbations]\label{lemma:sta_chang_D}
Let \(\rho,\sigma\) be density operators and let \(X,Y\) be unitary operators. Then
\begin{equation}
        D_\rho(X \| Y)^2 \leq  D_\sigma(X \| Y)^2 + 4\|\rho -\sigma\|_1. 
\end{equation}
\end{lemma}
\begin{proof}
We start by writing the difference
\begin{align*}
\ | D_\rho(X\| Y)^2 -D_\sigma(X \| Y)^2\ | &= \ |\mathsf{Tr}(\rho (X-Y)^\dagger(X-Y))-\mathsf{Tr}(\sigma (X-Y)^\dagger(X-Y))\ |\\
    &= \ |\mathsf{Tr}(\rho U^\dagger U)-\mathsf{Tr}(\sigma U^\dagger U)\ |=\ |\mathsf{Tr}(\rho U^\dagger U-\sigma U^\dagger U)\ |\\ &= \ |\mathsf{Tr}((\rho-\sigma) U^\dagger U) \ |
\end{align*}
Next we use the Schatten norm, $\|ST\|_1\leq \|S\|_p \|T\|_q$ with $1=1/p+1/q$, obtaining 
\begin{equation*}
 |\mathsf{Tr}((\rho-\sigma) U^\dagger U) | \leq \|\rho-\sigma \|_1   \|U^\dagger U \|_{\infty}
\end{equation*} 
\noindent for $p=1$ and $q=\infty$.

We now bound $\|U^\dagger U \|_\infty\leq \|U \|_\infty^2= \|X-Y \|_\infty^2\leq (\|X \|_\infty+\|Y \|_\infty)^2\leq 4$. Thus, $\ | D_\rho(X\| Y)^2 -D_\sigma(X\| Y)^2 | \leq 4\|\rho-\sigma \|_1.$
\end{proof}

We now derive bounds relating the states arising from the approximate strategy to the corresponding ideal reference states. To this end, we take the opportunity to establish the full collection of robustness statements for the Clifford-basis test (\Cref{def:Cliff_EPR_self}), including both the states and the observables.

\paragraph{Proof of \Cref{lem:clifford_basis_rigidity}.} The proof will, as before, reduce the analysis of the protocol in \Cref{fig:delegated_state_preparation} to the two non-communicating parties Alice and Bob, with Charlie’s behaviour absorbed into the state shared by Alice and Bob at the beginning of the second round. We therefore define an approximate strategy achieving a winning probability $p_{\mathrm{win}}$ as $S'=\{\{\rho_c\}_c,\{A_x\}_x,\{B_y\}_y\}$. Furthermore, for each state $\rho_c$, let $\mathcal{C}_c$ denote the ideal Clifford operator corresponding to the operation requested from Charlie. We define the rotated state
\begin{equation}
\rho_{cc^{-1}} := (\mathcal{C}_c \otimes I)\rho_c(\mathcal{C}_c \otimes I)^\dagger
\end{equation}
for each Clifford label $c$. 

We now observe that, when restricted to a fixed state $\rho_{cc^{-1}}$, and after relabelling Bob's observables $\tilde B_v$ according to conjugation by the Clifford operator $\mathcal C_c$, the resulting statistics are exactly reduced to those of the parallel Mermin--Peres self-testing scenario. For example, in the constraint-system notation, suppose Alice is asked to assign values to variables $v_{j,1}^A,v_{j,2}^A,v_{j,3}^A$ (among variables that define a parallel Mermin--Peres constraint), which need to satisfy the constraint of $v_{j,1}^Av_{j,2}^Av_{j,3}^A=1$, and that Bob is asked to assign a value to $v_{j,1}^B$. Here, we specify the variables on each side with a superscript for clarity. Now suppose Charlie is asked to perform a $\mathsf{CNOT}$ which is delegated to Alice's pair of qubits $j$ via teleportation. Then, the verifier will check if Bob's assignment to $v_{j,1}^B$ is the same as Alice's assignment to $v_{j,3}^A$.

Consequently, each such setting is equivalent to the original parallel Mermin--Peres game up to the corresponding Clifford relabeling, so that there exists for each $c$ a local isometry $V_c = V_A^c\otimes V_B^c$ acting on the joint Hilbert space of Alice and Bob such that, for Pauli strings $\vecs,\vect,\vecu,\vecv\in\{0,1\}^{2m}$, there exist observables $A_{\vecs,\vect}$ and $B_{\vecu,\vecv}$ on Alice's and Bob's local Hilbert spaces satisfying
\begin{equation}
\Bigl|\langle \phi_c|X(\vecs)Z(\vect)\otimes X(\vecu)Z(\vecv)|\phi_c\rangle-\langle \psi_c| A_{\vecs,\vect}\otimes B_{\vecu,\vecv}|\psi_c\rangle\Bigr|\leq O(m^2\sqrt{\varepsilon})
\end{equation}
where $|\phi_c\rangle := V_c(|\psi_c\rangle)$. Moreover, 
\begin{equation}
\bra{\phi_c} (I^{k} \otimes \mathcal{C}_c \otimes I^{m'})(\ket{\Phi}  \bra{\Phi}_{AB})(I^{m'} \otimes \mathcal{C}_c^\dagger \otimes I^{k} )\otimes I_{junk} \ket{\phi_c} \geq 1-O(m^2\sqrt{\varepsilon})
\end{equation}
\noindent where $\Phi_{AB}$ is the $2m$-fold tensor product of Bell pairs $\bigotimes_{i=1}^{2m}\ket{\Phi^+}_{A_iB_i}$, $k\in[2m]$, and $m'=2m-k-|c|$ based on bounds determined by \cite{coudron2016parallel}.

Importantly, the corresponding local isometries must in fact coincide, i.e., $V_c = V_{c'}$ for all $c,c'\in\{\mathsf{Swap},\mathsf{CNOT},\mathsf{X},\mathsf{I}\}$, over all possible qubit subsystems of \(\mathcal H_B\). Indeed, the isometry represents the degree of freedom relating the physical observables to their extracted ideal counterparts. Since Charlie cannot communicate with Alice and Bob, he cannot coordinate any adaptive change of this degree of freedom between different Clifford-conditioned executions of the protocol. Consequently, Alice and Bob must employ a single family of observables throughout. The extracted observables must therefore remain consistent across all post-selected settings, forcing the associated rigidity isometries to coincide. We conclude that the rigidity isometry can be chosen uniformly over all Clifford labels \(c\), yielding the full statement of \Cref{lem:clifford_basis_rigidity}.\hfill\qedsymbol\\

The previous proven bounds for \Cref{def:Cliff_EPR_self} allow us to control the distance between the state prepared through Charlie’s actions and the ideal reference state considered in \Cref{def:QI_LBCS}. More importantly, it allows us to rewrite all the state-dependent distances arising from \Cref{def:Classical_Int} with respect to the ideal Clifford-rotated EPR states. From this, we can derive the following inequalities for the generators and relators.

\begin{lemma}\label{lemma:set_of_dif2}
Let $\{\tilde B_v\}$ be Bob’s observables and $\{\tilde A_v^{(r_v)}\}$ Alice’s observables. Assume the strategy succeeds with probability at least $1-\varepsilon$. Then $\{\tilde B_v\}$ is an approximate operator solution in the sense that
\begin{equation}
\sum_{r'} D_\rho \left(\prod_{v \in r'} I_A \otimes  \tilde{B}_{v} ^\dagger\ \Big \| \  (-1)^{\lambda_{r'}}I\right)  \le\  O(m^52^{4m}\sqrt{\varepsilon}),
\end{equation}

\begin{equation}
\sum_{v,v',c}D_\rho \left(I_A \otimes [\tilde{B}_v,\mathcal{C}_c\tilde{B}_{v'}\mathcal{C}_c^\dagger]\ \Big \| \ I \right)\le O(m^82^{4m}\sqrt{\varepsilon}).
\end{equation}

The same bounds hold for Alice operators $\{\tilde A_v^{(r_v)}\}$ in place of Bob’s observables $\{\tilde B_v\}$,  and $\rho = \ket{\Phi}\bra{\Phi}_{A,B} \otimes \rho_{\mathrm{junk}}$.
\end{lemma}
\begin{proof}
Starting from, 
\begin{equation*}
\E_{r,v,\,c}\ \frac{1}{4} \left [D_{\rho_c}\left(\tilde{A}_v^r \otimes \tilde{B}_{g(v,c)} \Big\| I \right)^2\right ] \le \varepsilon,
\end{equation*}
obtained from \Cref{lemma:prob_to_dist2}, and considering that 
\begin{equation*}
D_{\rho_c}\left(\tilde{A}_v^r \otimes \tilde{B}_{g(v,c)}\ \Big\|\ I \right)^2 =  D_{\mathcal C_c\rho_c\mathcal C_c^{\dagger}}\left(\tilde{A}_v^r \otimes \mathcal C_c\tilde{B}_{g(v,c)}\mathcal C_c^{\dagger}\ \Big\|\ I \right)^2.
\end{equation*}

\noindent We can use \Cref{lemma:sta_chang_D} to write all distances in function of a single reference state, which for us will be $\rho= \ket{\Phi}\bra{\Phi}_{A,B} \otimes \rho_{\mathrm{junk}}$,
\begin{equation}
D_{\rho}\left(\tilde{A}_v^r \otimes \mathcal C_c\tilde{B}_{g(v,c)}\mathcal C_c^{\dagger}\ \Big\|\ I \right)^2 \leq  D_{\rho_c}\left(\tilde{A}_v^r \otimes \tilde{B}_{g(v,c)}\ \Big\|\ I \right)^2+4\left \|\mathcal C_c \rho_c\mathcal C_c^{\dagger}-\rho\right \|_1.
\end{equation}

Thus, 
\begin{equation*}
\sum_{r,v,c} D_{\rho_c} \left(\tilde{A}_v^r \otimes \mathcal{C}_c\ \big(\tilde{B}_{g(v,c)}\big)\mathcal{C}_c^\dagger\  \Big\|\ I \right)\leq 2|r||v||c|\sqrt{\varepsilon}(m^2+1). 
\end{equation*}
\begin{equation}\label{eq:upd_state_dist}
\sum_{r,v,c} D_{\rho_c} \left(\tilde{A}_v^r \otimes I_B\  \Big\|\  I_A \otimes \mathcal{C}_c^\dagger\ \big(\tilde{B}_{g(v,c)}^ \dagger\big)\mathcal{C}_c\right)\leq m^52^{4m}\sqrt{\varepsilon}. 
\end{equation}

At this point the proof just follows as in \Cref{lemma:set_of_dif} with the bound of \Cref{eq:upd_state_dist} in place of \Cref{lemma:approx_AvBv}.
\end{proof}

Finally we state our main robustness lemma for \Cref{sec:Classical_int}, 
\begin{lemma}\label{lemma:rob_sec5}
\Cref{def:Classical_Int} robustly self-tests the operator strategy of~\Cref{lem:Extended_Z2_faithful_rep}, in the sense of~\Cref{def:OperatorRigidity}, with robustness $O(m^{11}2^{7m}\sqrt{\varepsilon})$.
\end{lemma}
\begin{proof}
By \Cref{lemma:set_of_dif2}, any strategy succeeding with probability at least $1-\varepsilon$ yields observables that approximately satisfy the defining relations of the Mermin--Peres extended faithful $\mathbb Z_2^m$ operator-valued constraint system, on the fixed reference state $\rho=|\Phi\rangle\!\langle\Phi|_{AB}\otimes\rho_{\mathrm{junk}}$. Thus, the hypotheses needed in the proof for \Cref{def:QI_LBCS} are recovered, with the bounds of \Cref{lemma:set_of_dif2} replacing those of \Cref{lemma:set_of_dif}.

Repeating the same canonical-form reduction and applying \Cref{lemma:Vidick}, together with the rigidity of \Cref{lem:Extended_Z2_faithful_rep}, gives the robustness bound.
\end{proof}

\section{Circuit depth for multi-controlled phase gates}\label{append:Toffoli}

\subsection{Depth lower bound for protocol success}

Note that, because the verifier supplies the states from \Cref{def:QI_LBCS} and the protocol is robust (as shown in \Cref{lem:RQP}), the depth lower bounds for the observables certified by the protocol reduce directly to the depth lower bounds for the corresponding exact or approximate multi-controlled phase operators. Indeed, the state preparation is fixed by the verifier, while the operator self-testing statement of \Cref{def:OperatorRigidity} implies that the relevant operator isometries, together with the corresponding conjugated observables, must be realized by the provers themselves. Therefore, since the shared EPR pairs can be prepared in two layers of gates, no isometry can reduce the gate or depth complexity: the isometry and the effective observables implemented by the provers together realize operators that are precisely equivalent to the required multi-controlled phase operators.

The dequantized protocol of \Cref{def:Classical_Int}, however, requires additional analysis. Unlike the quantum-input setting, the rigidity statement is specified only up to local isometries, which in principle allows part of the gate and depth complexity to be shifted from the measurement stage to the state-preparation stage. As a result, the computational complexity is not \emph{a priori} localized to either the first or the second round of the protocol. For the purposes of a depth hierarchy, however, it is desirable to identify a specific round that necessarily incurs the depth cost. We therefore refine the analysis to show that any successful strategy must realize an $\Omega(\log m)$-depth computation in the second round of the protocol in the exact setting, and likewise in the approximate setting for sufficiently small error. 

A related issue concerns the use of ancillary systems. Since the self-testing statement certifies the target observables only up to local isometries, one might imagine implementing the required correlations via a joint unitary acting on a larger Hilbert space, followed by partial measurements, thereby reproducing the desired statistics with a shallower circuit. We explicitly rule out this possibility within our framework, showing that ancillary systems do not provide a mechanism for circumventing the depth lower bound.

Together, these observations yield the cleaner depth hierarchy stated in the main theorems.

\paragraph{Exact case.} We begin by analyzing the exact case of \Cref{def:Classical_Int}. 
We first describe a general circuit to implement the quantum strategy in the nonlocal game. On a shared entangled state on systems $A$ and $B$, the provers each append their own system with an ancillary system, which without loss of generality can be taken as $\ket{\mathbf{0}}_{A'}$ and $\ket{\mathbf{0}}_{B'}$. 

Upon receiving the questions, the provers apply a unitary operation $\tilde{A}_{q_A}$ and $\tilde{B}_{q_B}$ on their own systems, respectively. Finally, they measure a part of their ancillary systems $A'$ and $B'$ on the computational basis and obtain the measurement outcomes.

By the rigidity statement in the self test, there exist $V_{AA'}$ and $V_{BB'}$ such that for every $(q_A,q_B)$,
    \[
        (V_{AA'}\tilde{A}_{q_A}\otimes V_{BB'}\tilde{B}_{q_B})(\ket{\Phi_{AB}^+}\otimes\ket{\mathbf{0}}_{A'B'})=(A_{q_A}\otimes B_{q_B})\ket{\Phi_{AB}^+}\otimes\ket{\mathbf{0}}_{A'B'}.
    \]

Here, we denote the ideal multi-controlled phase operators with respect to $m$-bit question labels $q_A$ and $q_B$ as $A_{q_A}$ and $B_{q_B}$ (corresponding to $z_{\mathbf{i}}$ in \Cref{def:OVCS_Z2}).
Since the verifier can access the provers' behaviour only through the resulting measurement statistics, the self-testing statement leaves an ambiguity corresponding to local isometries $V_{AA'}$ and $V_{BB'}$. Our goal is therefore to show that no choice of such isometries can substantially reduce the complexity of implementing the entire physical family of operators, $\{\tilde{A}_{q_A}\}$ and $\{\tilde{B}_{q_B}\}$.

To this end, we first observe that the self-tested state $\Phi_{AB}$ is a maximally entangled state, up to local Clifford transformations. 
This allows us to adapt the argument of \cite{cleve2014characterization} to deal with measurement strategies employing ancillary systems and remove the state dependence in the self-testing statement.
In particular, any admissible strategy can be represented by an isometry that factors into the registers supporting the EPR pairs and an ancillary subsystem. Consequently, any complexity arising from the ancillary degrees of freedom can be separated from the certified observables, and the relevant gate and depth complexity is therefore essentially captured by the latter.

\begin{proposition}\label{prop:ToffoliIsom}
    The set of operators $\{\tilde{A}_{q_A}\}$ is isomorphic to the set of operators $\{A_{q_A}\}$ in the sense that $\tilde{A}_{q_A}=V_{AA'}^\dagger(A_{q_A}C_A\otimes\ket{\mathbf{0}}\bra{\mathbf{0}}_{A'}\oplus C_{A,\bot})$ for every $q_A$, where $C_A$ is a constant unitary operator on system $A$, and $C_{A,\bot}$ is a constant unitary operator on the subsystem orthogonal to $\mathcal{H}_A\cup\{\ket{\mathbf{0}}_{A'}\}$. A similar result holds for $\{\tilde{B}_{q_B}\}$ and $\{B_{q_B}\}$.
\end{proposition}

\begin{proof}
    To simplify the notation in the derivations, we denote 
    \begin{equation}\label{eq:NotationSimplify}
    \begin{aligned}
        \tilde{E}_{q_A}&=V_{AA'}\tilde{A}_{q_A},\quad \tilde{E}_{q_B}=V_{BB'}\tilde{B}_{q_B}, \\
        E_{q_A}&=A_{q_A}\otimes\ketbra{\mathbf{0}}{\mathbf{0}}_{A'},\quad E_{q_B}=B_{q_B}\otimes\ketbra{\mathbf{0}}{\mathbf{0}}_{B'}.
    \end{aligned}
    \end{equation}
    Using the simplified notation, the self-testing constraint is denoted as
    \begin{equation}\label{eq:SimplifiedSelfTest}
      (\tilde{E}_{q_A}\otimes\tilde{E}_{q_B})(\ket{\Phi^+}_{AB}\otimes\ket{\mathbf{0}}_{A'B'})=(E_{q_A}\otimes E_{q_B})(\ket{\Phi^+}_{AB}\otimes\ket{\mathbf{0}}_{A'B'}).
    \end{equation}
    Tracing out systems $BB'$ in \Cref{eq:SimplifiedSelfTest}, as $E_{q_B}$ is unitary, the reduced density matrix becomes
    \begin{equation}
        \rho_{AA'}=\frac{1}{2^{m}}\tilde{E}_{q_A}({I}\otimes\ketbra{\mathbf{0}}{\mathbf{0}}_{A'})\tilde{E}_{q_A}^\dag 
        =\frac{1}{2^{m}}E_{q_A}E_{q_A}^\dag\otimes\ketbra{\mathbf{0}}{\mathbf{0}}_{A'}
        =\frac{1}{2^{m}}{I}_{A}\otimes\ketbra{\mathbf{0}}{\mathbf{0}}_{A'}.
    \end{equation}
    Therefore, there exists some unitary $M_{q_A}$ on system $A$ such that
    \begin{equation}
        \tilde{E}_{q_A}=M_{q_A}\otimes\ketbra{\mathbf{0}}{\mathbf{0}}_{A'}\oplus C_{A,\bot},
    \end{equation}
    with $C_{A,\bot}$ a constant unitary on the subsystem orthogonal to $\mathcal{H}_A\cup\{\ket{\mathbf{0}}_{A'}\}$.
    A similar argument can be applied to system $B$, hence there exist some unitary operators $M_{q_A}$ and $M_{q_B}$ such that
    \begin{equation}
        (M_{q_A}\otimes M_{q_B})\ket{\Phi^+}_{AB}=(A_{q_A}\otimes B_{q_B})\ket{\Phi^+}_{AB}.
    \end{equation}
    Using the transpose trick over $\ket{\Phi^+}_{AB}$, we have $M_{q_A}M_{q_B}^{\mathrm{T}}=A_{q_A}B_{q_B}^{\mathrm{T}}$.    Noticing that $B_{q_B}$ is the multi-controlled phase operation that is real and diagonal in the computational basis, hence $M_{q_A}=A_{q_A}B_{q_B}(M_{q_B}^{\mathrm{T}})^{-1}$, of which the left hand side is independent of label $q_B$; also notice that such an equation needs to hold for every $q_A,q_B$. For this to hold, we must have $B_{q_B}(M_{q_B}^{\mathrm{T}})^{-1}=C_A$
    with $C_A$ some constant unitary operator (we label it with $A$ for notation consistency though it is derived on system $B$). Note that such a gauge freedom is inevitable due to the fact that for any unitary $C$, $(C\otimes C^*)\ket{\Phi^+}_{AB}=\ket{\Phi^+}_{AB}$; operationally, such a unitary may be applied to $\ket{\Phi^+}$ before measurements without changing the state. After reorganizing the formula, we have that for any $(q_A,q_B)$,
\begin{equation}\label{eq:IsometryEqual}
\tilde{A}_{q_A}=V_{AA'}^\dag(A_{q_A}C_A\otimes\ketbra{\mathbf{0}}{\mathbf{0}}_{A'}\oplus C_{A,\bot}),\quad
\tilde{B}_{q_B}=V_{BB'}^\dag(B_{q_B}C_B\otimes\ketbra{\mathbf{0}}{\mathbf{0}}_{B'}\oplus C_{B,\bot}).
\end{equation}
\end{proof}

Next, we show the depth lower bounds to implement the set of operators $\{\tilde{A}_{q_A}\}$ ($\{\tilde{B}_{q_B}\}$ resp.).
It is worth noting that an arbitrary self-testing isometry could in principle rotate the computational basis to another orthonormal basis, for instance the Fourier/parity basis, thereby mapping the singleton phase-flip operators $I-2\ket{q}\bra{q}$ to singleton phase flips with respect to that new basis. 
However, it cannot map them to, for instance, Pauli parity observables $Z(\mathbf{s})$, since unitary conjugation preserves the spectrum. More generally, for any orthonormal basis $\{\ket{u_q}\}$ obtained from the self-tested realization, the conjugated observables remain rank-one reflections of the form $I-2\ket{u_q}\bra{u_q}$. 
The relevant question is whether the isometry can transform the entire family into another basis in which all of these rank-one reflections admit a substantially lower-depth implementation. Establishing or ruling out such a possibility requires an argument beyond the spectral properties alone.
For this purpose, we shall combine the use of a light cone argument and a counting argument.

\begin{lemma}\label{prop:Toffoli_lower}
Any circuit that implements the set of operators $\{\tilde{A}_{q_A}\}$ with single- and two-qubit gates requires a circuit depth $d=\Omega(\log{m})$. A similar result holds for $\{\tilde{B}_{q_B}\}$.
\end{lemma}

\begin{proof}
    The self-testing reference $\{A_{q_A}\}$ are multi-controlled phase operators. For convenience, we abbreviate $q_A$ as $q$ and denote $A_{q_A}={I}-2\ketbra{q}{q}$, where $\{\ket{q}\}$ forms an orthonormal complete basis of $\mathcal{H}_A$ ($q_A$ specifies an $m$-bit string per \Cref{def:OVCS_Z2} and $A_{q_A}$ is given in \Cref{lem:Z2_faithful_rep}). Then,
    \[
        A_{q}C_A=({I}-2\ketbra{q}{q})C_A=C_A-2\ketbra{q}{q}C_A.
    \]
    Since $C_A$ is a unitary operator, $\ket{\tilde{q}}=C_A^\dag\ket{q}$ also forms an orthonormal complete basis. Taking it to \Cref{eq:IsometryEqual}, when $\ket{\tilde{p}}$ is input to the principal system $A$ and $\ket{\mathbf{0}}_{A'}$ is input to the ancillary system $A'$,
    \begin{equation}
        \tilde{A}_{q}(\ket{\tilde{p}}\otimes\ket{\mathbf{0}}_{A'})=V_{AA'}^{\dag}(A_qC_AC_A^\dag\ket{p}\otimes\ket{\mathbf{0}}_{A'})
        =(-1)^{\delta_{q,p}}V_{AA'}^{\dag}(\ket{p}\otimes\ket{\mathbf{0}}_{A'})
        :=(-1)^{\delta_{q,p}}\ket{\bar{p}}_{AA'},
    \end{equation}
    in which $A_q\ket{p}=(-1)^{\delta_{q,p}}\ket{p}$. Therefore, $\tilde{A}_{q}$ is given by $2^m$ linearly independent operators:
    \begin{equation}\label{eq:DimCapacity}
        \tilde{A}_{q}=\sum_{p=0}^{2^m-1}(-1)^{\delta_{q,p}}\ket{\bar{p}}\bra{\tilde{p},\mathbf{0}}_{AA'},
    \end{equation}
    which spans a linear space with dimension $2^m$. Specifically, the phase flip given by $(-1)^{\delta_{q,p}}$ requires reading the $m$-bit string value carried by the input basis state, calculating the function value $(-1)^{\delta_{q,p}}$, and encoding the value to one of $2^m$ orthogonal states among $\ket{\bar{p}}$.

    Now, suppose the set of unitaries $\{\tilde{A}_q\}$ can be implemented with a maximum depth $d$. We denote the total number of qubits in $A$ and $A'$ as $M$. Using a light cone argument, each $\tilde{A}_q$ is restricted to be acting on a subset of qubits $\mathcal{S}_q$ over $A$ and $A'$ with 
    \begin{equation}\label{eq:lightcone}
        |\mathcal{S}_q|\leq2^d:=w,
    \end{equation}
    and we can hence decompose $\tilde{A}_q$ as $\tilde{A}_q=\tilde{A}_{q|\mathcal{S}_q}\otimes{I}_{\bar{\mathcal{S}_q}}$,
    where $\tilde{A}_{q|\mathcal{S}_q}$ is a unitary operator acting only on the qubits in $\mathcal{S}_q$, and $\bar{\mathcal{S}_q}$ represents the rest $(M-w)$ qubits. Applying the spectral decomposition to $\tilde{A}_{q|\mathcal{S}_q}$, it specifies at most $2^w$ distinct eigenspaces. With $M$ qubits in total, there are $\binom{M}{w}$ ways to choose $w$ qubits. To meet the space dimension requirement for \Cref{eq:DimCapacity}, we must have
    \begin{equation}
        \binom{M}{w}2^w\geq 2^m.
    \end{equation}
    Therefore, combining the restriction of \Cref{eq:lightcone}, with $M=\mathsf{poly}(m)$, we have $d=\Omega(\log{m})$.
\end{proof}

\paragraph{Approximate case.}
The same reduction must now be carried out in the robust setting to transfer the depth lower bounds to approximate realizations of the multi-controlled phase operators independent of any isometry. 
After this is done, intuitively, the argument remains fundamentally unchanged. Even approximate implementations of highly nonlocal operations require information to propagate across many input qubits, and this propagation cannot be substantially compressed when the approximation error is sufficiently small. Consequently, approximate realizations of the certified observables continue to exhibit essentially the same depth requirements as their exact counterparts. We formalize this intuition by combining the robustness statement with approximate unitary-synthesis lower bounds, thereby obtaining depth lower bounds for all sufficiently accurate winning strategies.

\begin{lemma}\label{thm:robustdepth}
    Suppose there exist $V_{AA'}$ and $V_{BB'}$ such that for every $(q_A,q_B)$, 
    \begin{equation}\label{eq:ApproxSelfTest}
        \|(V_{AA'}\tilde{A}_{q_A}\otimes V_{BB'}\tilde{B}_{q_B})(\ket{\Phi_{AB}^+}\otimes\ket{\mathbf{0}}_{A'B'})-(A_{q_A}\otimes B_{q_B})\ket{\Phi_{AB}^+}\otimes\ket{\mathbf{0}}_{A'B'}\|_1\leq\delta.
    \end{equation}
    Then, for $\delta< 2\sqrt{2}\cdot2^{-m/2}$, there exists at least one $q_A$ ($q_B$ resp.) such that the circuit depth for the implementation of $\tilde{A}_{q_A}$ ($\tilde{B}_{q_B}$ resp.) is $\Omega(\log{m})$.
\end{lemma}

\begin{proof}
    We note that a large part of the full proof for this theorem is similar to the proofs for \Cref{prop:ToffoliIsom} and \Cref{prop:Toffoli_lower}. We do not repeat such contents for simplicity and only highlight the different parts. In addition, we shall follow the same notation simplification as above. By taking
    \begin{equation}
    \begin{aligned}
      \hat{E}_{q_A}&=V_{AA'}\tilde{A}_{q_A}(I\otimes\ket{\mathbf{0}}_{A'}),\quad \hat{E}_{q_B}=V_{BB'}\tilde{B}_{q_B}(I\otimes\ket{\mathbf{0}}_{B'}), \\
      \hat{A}_{q_A}&=A_{q_A}\otimes\ket{\mathbf{0}}_{A'},\quad \hat{B}_{q_B}=B_{q_B}\otimes\ket{\mathbf{0}}_{B'},
    \end{aligned}
    \end{equation}
    we can use \Cref{eq:ApproxSelfTest} and the transpose trick with respect to $\ket{\Phi^+}_{AB}$ to obtain an upper bound on the Frobenius-norm distance between the following operators:
    \begin{equation}
        \frac{1}{2^m}\|\hat{E}_{q_A}\hat{E}_{q_B}^{\mathrm{T}}-\hat{A}_{q_A}\hat{B}_{q_B}^{\mathrm{T}}\|_F^2\leq\delta^2< 8\cdot2^{-m}.
    \end{equation}
    Then,
    \begin{equation}\label{eq:FrobDistRadius}
        \|\hat{E}_{q_A}\hat{E}_{q_B}^{\mathrm{T}}\hat{B}_{q_B}-\hat{A}_{q_A}\hat{B}_{q_B}^{\mathrm{T}}\hat{B}_{q_B}\|_F =\|\hat{E}_{q_A}(\hat{E}_{q_B}^{\mathrm{T}}\hat{B}_{q_B})-\hat{A}_{q_A}\|_F:=\|\hat{E}_{q_A}W_{q_B}-\hat{A}_{q_A}\|_F\leq 2^{m/2}\delta<2\sqrt{2},
    \end{equation}
    which holds for every $q_A$ and $q_B$. For the ideal operators $\hat{A}_{q_A}$, given $q_A\neq q_A'$, $\|\hat{A}_{q_A}-\hat{A}_{q_A'}\|_F^2=8$. As \Cref{eq:FrobDistRadius} guarantees that every operator $\hat{E}_{q_A} W_{q_B}$ resides inside a local error sphere of radius strictly smaller than the separation between $\hat{A}_{q_A}$ and $\hat{A}_{q_A'}$, consequently, these error spheres are mutually disjoint. Furthermore, since the $2^m$ ideal operators $\hat{A}_{q_A}$ are mutually orthogonal, to guarantee the same linear independence among the operators without violating \Cref{eq:FrobDistRadius}, we must have $\dim(\mathrm{span}\{\hat{E}_{q_A}W_{q_B}\})\geq 2^m$. By this, we can follow a same dimension on the subspaces spanned by $\hat{A}_{q_A}$ as in \Cref{prop:Toffoli_lower} and derive the depth lower bound.
\end{proof}

\subsection{Depth upper bound}\label{append:Toff_upper}

We recall the explicit construction of \cite{nie2024quantum} for implementing the generalized Toffoli gate with a single ancillary qubit and derive the corresponding depth bound.

\begin{lemma}\label{thm:upperToffoli}
Let $\gamma(n)$ denote the circuit depth required to implement $\mathsf{C}^{n}(\mathsf{X})$ using single- and two-qubit gates and one(possibly dirty) ancillary qubit. Then $\mathsf{C}^{n}(\mathsf{X})$ admits an exact implementation of size $O(n)$ and depth
\begin{equation}\label{eq:Toffolidepth}
\gamma(n)\leq (2+2\gamma(4)+\gamma(3))(\lceil\log_2(n+4)\rceil-3)+\gamma(4).
\end{equation}
\end{lemma}
\begin{proof}
The proof follows by analyzing the recursive circuit implementation of \cite[Fig.\ 3]{nie2024quantum}.
    \begin{enumerate}
        \item The recursive definition works for an even-valued $n$. If $n$ is an odd number, then we need to round it up to the nearest even number. That is, we will take $n$ to $N$ with $N = n+(n \mod2)$, which operationally corresponds to adding a dummy qubit.

        \item The generalized Toffoli gate is decomposed into two smaller generalized Toffoli gates in parallel over a subgroup of the input control qubits. Each smaller gate takes $(N/2-2)$ qubits as control qubits, borrows one control qubit as an ancillary qubit, and computes the $\mathsf{AND}$ values over another input qubit serving as the target in this smaller Toffoli gate. 

        \item To effectively reduce the circuit depth, the circuit is implemented in a recursive manner. Notice that the implementation of a Toffoli gate can be owed to two phases: with the help of an ancillary qubit, the first phase computes the $\mathsf{AND}$ of the values carried by all the qubits on the target qubit, and the second phase carries out an uncompute operation to restore the original state of the ancillary qubit. With the final merge operation sandwiched in the middle of the circuit, its prior circuit performs only the first phase; its posterior circuit performs only the second phase; and the two parts of the circuit are recursively decomposed into smaller parts. 
    \end{enumerate}
    The merge operation is implemented by a generalized Toffoli gate with three control qubits. In the compute phase (resp. the uncompute phase), a generalized Toffoli gate with four control qubits and a layer of Pauli-$\mathsf{X}$ gates are also employed to help upload the computation results in the subgroups to the final target qubit. 
    In the end, in this circuit implementation of the generalized Toffoli gate using one ancillary qubit, the recursive derivation of the circuit depth, $\gamma(n)$, is given by
    \begin{equation}
        \gamma(n) = \gamma\left(\frac{N}{2}-2\right)+2+2\gamma(4)+\gamma(3)= \gamma\left(\lceil\frac{n}{2}\rceil-2\right)+2+2\gamma(4)+\gamma(3),
    \end{equation}
    where $\gamma(3)$ is the depth to implement a generalized Toffoli gate with three control qubits, and $\gamma(4)$ is the depth to implement a generalized Toffoli gate with four control qubits. By solving this recursive function, for every integer $n\geq4$ (note that the recursive construction starts from the case of $n=4$), we obtain our stated upper bound. 
\end{proof}

For (generalized) Toffoli gates with two, three, and four control qubits, respectively, using the Clifford$+\mathsf{T}$ gate set and a clean ancilla, we can upper-bound their circuit depths by
\begin{equation}
    \gamma(2)\leq 6,\quad \gamma(3)\leq 22, \quad \gamma(4)\leq 36,
\end{equation}
and one can refer to the textbook of \cite{nielsen2010quantum} for an explicit circuit implementation of the standard Toffoli gate with two control qubits. Moreover, $\mathsf{C} ^n (\mathsf X)$ and $\mathsf C^n(\mathsf Z)$ differ by at most two Hadamard gates acting on the target qubit, which can be parallelized with the remaining layers of the circuit constructed in \cite{nie2024quantum}. Consequently, the same depth bound applies to multi-controlled phase gates.

\section{The 3-Qubit Boolean hypercube BCS}\label{app:example}
Consider the constraint system of \Cref{def:OVCS_Z2}, for $m=3$. In this case, the representation matrices of each of the operators $z_{\veci}$ will be of the form $\diag(\pm 1,\ldots,\pm 1)$. This structure is enforced by the constraints $\langle z_{\veci}^2 = I$, $z_{\veci} z_{\vecj} = z_{\vecj} z_{\veci}, \forall \veci,\vecj \in \{0,1\}^m\rangle$ and $\prod_{\veci: \veci\cdot \vecs = 1} z_{\veci} = Z(\vecs)\ \ \forall \vecs \in \{0,1\}^m$. So, let the representation matrix of $z_{000}$ equivalently $z_{\mathbf{0}}$ be $\diag(y_{000},y_{001},y_{010},\ldots, y_{111})$. This matrix commutes with all the possible $\mathsf{Swap}$s and $\mathsf{CNOT}$s, as
\begin{align}
    \mathsf{Swap}_{a,b} (z_{000}) \mathsf{Swap}_{a,b} = z_{\sigma_{a,b}(000)}=z_{000} \ \ \ a\leq b\leq m, \text{and}   \\ 
     \mathsf{CNOT}_{a,b}  (z_{000}) \mathsf{CNOT}_{a,b} = z_{000+\mathbf{0}[a]\mathbf{e}_b}=z_{000+0\cdot\mathbf{e}_b}=z_{000},\quad \forall a,b\text{ s.t. } a\neq b\leq m. 
\end{align}

But notice that conjugating by a swap, one can swap entries on the diagonal of $z_{000}$; for instance
\begin{align}
  \mathsf{Swap}_{1,2}:&~ y_{001} \leftrightarrow y_{010} \\ 
  \mathsf{Swap}_{2,3}:&~ y_{010} \leftrightarrow y_{100} \\
  \mathsf{Swap}_{1,3}:&~ y_{100} \leftrightarrow y_{001}\;.
\end{align}
Thus if $\mathsf{Swap}_{1,2}(z_{000})=\mathsf{Swap}_{2,3}(z_{000})=\mathsf{Swap}_{1,3}(z_{000})=z_{000}$, then $y_{001}=y_{010}=y_{100}$. Similarly, $y_{011}=y_{101}=y_{110}$. Therefore, we already have that $z_{000}=\diag(a,b,b,c,b,c,c,d)$ for $a,b,c,d\in \{\pm 1\}$.

Subsequently we will make use of the $\mathsf{CNOT}$, such that 
\begin{align}
  \mathsf{CNOT}_{1,2}:&~ y_{001} \leftrightarrow y_{011} \\ 
  \mathsf{CNOT}_{2,3}:&~ y_{011} \leftrightarrow y_{111}
\end{align}
such that $b=c=d$, 
so that we have $z_{000}=\diag(a,b,b,b,b,b,b,b)$ for $a,b \in \{\pm 1\}$.

Now we can use the restriction $X_a\, z_{\veci}\, X_a = z_{\veci\oplus \mathbf{e}_a}\ \  \forall \veci\in\{0,1\}^m,\; a\in\{1,\dots,m\}$ to show that all $z_{\veci}$ are equal to $z_{000}$, just by moving the position of the value $a$ along the diagonal. For that see that, $X_a\, z_{000}\, X_a = z_{000\oplus \mathbf{e}_a}=z_{\mathbf{e}_a}\ \ \forall \veci\in\{0,1\}^m,\; a\in\{1,\dots,m\}$. Additionally, we know that,
\begin{align}
  \mathsf{X}_{a} :&~ y_{000} \leftrightarrow y_{\mathbf{e}_a} \\ 
  \mathsf{X}_{a}:&~ y_{010} \leftrightarrow y_{(010\oplus \mathbf{e}_a)}\;,
\end{align}
so we obtain also, for instance, that $z_{010}=\diag(b,b,a,b,b,b,b,b)$.

Finally, we have that $\prod_{\veci \in \{0,1\}^m} z_{\veci} = -I$, which implies that 
$\prod_{\veci\in\{0,1\}^m} z_{\veci}[000]=-1$.
However, because $z_{000}[000]=y_{000}=a$, for any $\mathbf{p}\neq \mathbf{0}$ $z_\mathbf{p}[000]$ is equal to $z_{000}[\mathbf{k}]$ for some other string $\mathbf{k} \neq 000$. In other words, it is equal to another value on the diagonal of $z_{000}$, and so we obtain that $a\cdot b^{2^m-1}=-1$. Finally, because $2^m-1$ is odd,  $a=-b$.

\end{document}